\documentclass[a4paper,11pt,twoside]{article}

\usepackage{a4wide}
\usepackage{amsmath,amssymb,amsfonts,amsthm}
\usepackage{moreverb,rotating,graphics}
\usepackage[hidelinks]{hyperref}
\usepackage[utf8]{inputenc}
\usepackage{epsfig}
\usepackage[english]{babel} 
\usepackage{framed,fancybox}
\usepackage{tikz,pgfplots}
\usepackage{caption}
\usepackage{subcaption}
\usepackage{mathrsfs}
\usepackage{esint}
\usepackage{dsfont}

\numberwithin{equation}{section} 

\newtheorem{theorem}{Theorem}[section]
\newtheorem{proposition}[theorem]{Proposition}
\newtheorem{remark}[theorem]{Remark}
\newtheorem{lemma}[theorem]{Lemma}
\newtheorem{corollary}[theorem]{Corollary}
\newtheorem{definition}[theorem]{Definition}

\newcommand\1{{\ensuremath {\mathds 1} }} 

\def\C{{\mathbb C}}

\def\bbI{{\mathbb I}}

\def\N{{\mathbb N}}

\def\R{{\mathbb R}}
\def\RR{{\mathbb R}}

\def\Z{{\mathbb Z}}

\def\ba {{ \bold a}}

\def\bb{{\bold b}}

\def\bc{{\bold c}}

\def\bk{{\bold k}}

\def\bm{{\bold m}}

\def\bQ{{\bold Q}}
\def\bq{{\bold q}}
\def\bR{{\bold R}}

\def\bx{{\bold x}}
\def\by{{\bold y}}
\def\bz{{\bold z}}

\def\bnull{{\bold 0}}

\def\rd{{\mathrm{d}}}
\def\re{{\mathrm{e}}}
\def\ri{{\mathrm{i}}}

\def\Re{{\mathrm{Re }\, }}

\def\Tr{{\rm Tr}}
\def\VTr{ \underline{ \rm Tr \,}}

\def\cB{{\mathcal B}}
\def\cC{{\mathcal C}}
\def\cD{{\mathcal D}}
\def\cE{{\mathcal E}}
\def\cF{{\mathcal F}}

\def\cH{{\mathcal H}}

\def\cJ{{\mathcal J}}
\def\cK{{\mathcal K}}
\def\cL{{\mathcal L}}

\def\cN{{\mathcal N}}
\def\cP{{\mathcal P}}
\def\cQ{{\mathcal Q}}

\def\cS{{\mathcal S}}
\def\cT{{\mathcal T}}

\def\cZ{{\mathcal Z}}

\def\fa{{\mathfrak{a}}}

\def\fm{{\mathfrak{m}}}

\def\fS{{\mathfrak S}}

\newcommand{\sS}{\mathscr{S}}
\newcommand{\sC}{\mathscr{C}}

\newcommand{\bra}{\langle}
\newcommand{\ket}{\rangle}

\newcommand{\BZ}{{\Gamma^\ast}} 
\newcommand{\WS}{\Gamma} 
\newcommand{\RLat}{\mathcal{R}^*} 
\newcommand{\Lat}{\mathcal{R}} 

\newcommand{\per}{\mathrm{per}}

\newcommand{\loc}{\mathrm{loc}}

\newcommand{\supp}{\mathrm{supp}}
\newcommand{\gap}{\mathrm{gap}}

\def\sqw{\hbox{\rlap{\leavevmode\raise.3ex\hbox{$\sqcap$}}$%
\sqcup$}}
\def\cqfd{\ifmmode\sqw\else{\ifhmode\unskip\fi\nobreak\hfil
\penalty50\hskip1em\null\nobreak\hfil\sqw
\parfillskip=0pt\finalhyphendemerits=0\endgraf}\fi}

\hypersetup{
    colorlinks,
    linkcolor={blue!50!blue},
    citecolor={red}
}
\renewcommand{\eqref}[1]{(\ref{#1})}

\usepackage[titletoc]{appendix}

\newcommand{\av}[1]{\left| #1 \right|}

\begin{document}

\title{Supercell calculations in the reduced Hartree-Fock model for crystals with local defects.}

\author{
  David Gontier, Salma Lahbabi
}

\date{}

\maketitle

\begin{abstract}
In this article, we study the speed of convergence of the supercell reduced Hartree-Fock~(rHF) model towards the whole space rHF model in the case where the crystal contains a local defect.  We prove that, when the defect is charged, the defect energy in a supercell model converges to the full rHF defect energy with speed $L^{-1}$, where $L^3$ is the volume of the supercell. The convergence constant is identified as the  Makov-Payne correction term when the crystal is isotropic cubic. The result is extended to the non-isotropic case.
\end{abstract}


\section{Introduction}

The numerical simulation of crystals is a very active area of research in solid state physics, material science and nano-electronics. Although the simulation of perfect crystals is well-understood nowadays, the treatment of local defects in solids is still a major issue~\cite{Pisani1994, Stoneham2001}. 

\medskip

The state-of-the-art method to simulate crystals (with or without local defects) is the \emph{supercell} method. It consists in considering a large box $\Gamma_L = L \Gamma$, with periodic boundary conditions. When the crystal is an insulator or a semi-conductor without defects, the supercell model converges rapidly to the whole space model, as it has been numerically observed in the work of Monkhorst and Pack~\cite{Monkhorst1976}. Actually, in~\cite{GL2015}, we consider  the reduced Hartree-Fock model~\cite{Solovej1991}, which is obtained from the generalized Hartree-Fock model~\cite{LiebSim1977} by removing the exchange term, and we prove that the convergence is exponential in this case. Precisely, for a nuclear charge density $\mu_\per$, which is a periodic function, if we denote by $I_{\mu_\per}$ the energy per unit volume of the whole space model and $I^L_{\mu_\per}$ the energy of the crystal restrained to the box $\Gamma_L$, then 
there exist constants $C \ge 0$ and $A > 0$
such that
 \[
\left| \frac{1}{\av{\Gamma_L}} I_{\mu_\per}^L - I_{\mu_\per} \right| \le C \re^{-AL}.
 \]

When the crystal contains local defects, the nuclear charge density of the crystal is of the form $\mu_\per + \nu$, where  $\nu$ is a smooth function with compact support representing the charge density of the defect. 
The defect energy $J_\nu$ is formally defined as the difference between the energy of the crystal with the defect and the energy of the crystal without the defect. In the supercell model, it is given by
\[
	J_\nu^L := I^L_{\mu_\per + \nu} - I^L_{\mu_\per}.
\]
It has been proved in~\cite{Cances2008} that $J_\nu^L$ converges to  $J_\nu$. When the defect is charged ($q:=\int_{\RR^3}\nu\neq 0$), the convergence of $J^L_{\nu}$ to $J_\nu$ is slow with respect to the size of the supercell $L$. Numerically, one finds that the convergence rate is of the order $L^{-1}$. This slow convergence comes from two effects. First, the supercell method induces spurious interactions between the defect and its periodic images. Then, one always needs to add a jellium background to compensate charged defects in order to satisfy the periodic boundary conditions imposed on the electrostatic potential, so that the reference ``zero energy'' is shifted as $L$ goes to infinity.

\medskip

The main result of this paper is that when the defect is small (see Theorem~\ref{th:chargedDefects} for the exact assumption), it holds that
\begin{equation} \label{eq:speed_cv_Lm1}
	 J_\nu = J_\nu^L + \dfrac{\beta q^2}{L} + O(\left\| \nu \right\| \exp^{-\alpha L} , \left\| \nu \right\|^2 L^{-3} , \left\| \nu \right\|^3),
\end{equation}
where the constant $\beta$ can be computed explicitly as a by-product of the supercell calculation. Approximating $J_\nu$ by $J_\nu^L+\beta q^2/L$ rather that only $J_\nu^L$ therefore speeds up the convergence at a negligible computational cost. When the crystal is isotropic cubic (see Definition~\ref{def:isotropic_cubic}), this constant is of the form $\beta \approx \fm / 2\epsilon$, where  $\epsilon$ is the macroscopic dielectric constant of the perfect crystal (see also~\cite{Adler1962, Wiser1963, Cances2010}), and $\fm$ is  the Madelung constant of the crystal. We recover the term predicted by Leslie and Gillan in~\cite{Leslie1985}, further developed by Makov and Payne~\cite{Makov1995}, and observed numerically in simulations~\cite{Komsa2012}. We therefore give a rigorous mathematical proof to the ``phenomenological approach of Leslie and Gillan in which the potential is reduced by the dielectric constant''~\cite{Makov1995}. These last articles have been the starting point to a large variety of methods to improve supercell calculations in the presence of charged defects. Let us mention the Freysolt, Neugebauer and Van de Walle method~\cite{Freysoldt2009, Freysoldt2010}, the Lany and Zunger method~\cite{Lany2008, Lany2009} and the Taylor and Bruneval method~\cite{Taylor2011}.
In the non isotropic case, the correcting term was proposed in~\cite{Murphy2013} using physical considerations. \\

The article is organized as follows. In Section~\ref{sec:rHF_MainResults}, we recall the reduced Hartree-Fock model for both the perfect crystal and the crystal with local defects, together with their corresponding supercell models, and we state our main results. In Section~\ref{sec:linear}, we identify the linear response of the crystal with respect to the defect. In Section~\ref{sec:BlochRegularity}, we recall the theory of Bloch transform, and use it to sketch the main steps of the proof. The details of the proofs are presented in Section~\ref{sec:detailsProofs}. 
Some complementary results about the convergence of Riemann sums are given in the appendices.


\section{Presentation of the models and main results}
\label{sec:rHF_MainResults}

\subsection{The rHF model for perfect crystals}
\label{ssec:rHF}

We introduce in this section the rHF model for a perfect crystal following the work in~\cite{Catto1998, Catto2001, Cances2008}. A perfect crystal is a periodic arrangement of atoms. We denote by $\Lat = \ba_1 \Z + \ba_2 \Z + \ba_3 \Z$ the underlying periodic lattice (in $\R^3$), where $(\ba_1, \ba_2, \ba_3)$ are linearly independent vectors in $\R^3$. The reciprocal lattice is denoted by $\RLat = \ba_1^* \Z + \ba_2^* \Z + \ba_3^* \Z$, where the vectors $\ba_k^* \in \R^3$, $1 \le k \le 3$, are chosen such that $\ba_k \cdot \ba_l^* = 2 \pi \delta_{kl}$. The unitary cell is $\WS :=  \ba_1 [-1/2, 1/2) + \ba_2 [-1/2, 1/2) + \ba_3 [-1/2, 1/2)$, and the reciprocal unitary cell is $\BZ :=  \ba_1^* [-1/2, 1/2) + \ba_2^* [-1/2, 1/2) + \ba_3^* [-1/2, 1/2)$. We denote by $P_j := - \ri \partial_{x_j}$, $j \in \{ 1,2,3\}$,  the $j$-th momentum operator, and by $(-\Delta) := \sum_{j=1}^3 P_j^2$ the Laplacian operator on the usual complex valued Lebesgue space $L^2(\R^3)$, seen here as an Hilbert space with its natural inner product. We also introduce the usual complex valued $\Lat$-periodic Sobolev and Lebesgue spaces
\[
	H^s_\per(\WS) := \left\{ f \in H^s_\loc, \ f \ \text{is $\Lat$-periodic} \right\} \quad \text{and} \quad 
	L^p_\per(\WS) := \left\{ f \in L^p_\loc, \ f \ \text{is $\Lat$-periodic} \right\}.
\]
If $f \in L^2_\per(\WS)$, the normalized Fourier coefficients of $f$ are denoted by $c_\bk(f)$. Formally,
\[
	\forall \bk \in \RLat, \quad c_\bk(f)  = \dfrac{1}{| \WS |^{1/2} } \int_{\WS} f(\bx) \re^{-\ri \bk \cdot \bx} \rd \bx,
\]
and it holds
\[
	\| f \|_{L^2_\per(\WS)}^2 = \sum_{\bk \in \RLat} | c_\bk |^2
	\quad \text{and} \quad
	f(\bx) = \dfrac{1}{\left| \WS \right|^{1/2}} \sum_{\bk \in \RLat} c_\bk(f) \re^{\ri \bk \cdot \bx} \quad \text{a.e. and in} \ L^2_\per(\WS).
\]

\medskip

The charge density of the nuclei (together with the core electrons in mean-field models) of a perfect crystal is well-approximated by an $\Lat$-periodic real-valued function $\mu_\per$. In this article, we will assume that $\mu_\per$ is regular (say $\mu_\per \in L^2_\per(\WS)$), but more singular functions may be also treated~\cite{Blanc2003}. 

\medskip

Bulk properties of the crystal can be understood by studying the so-called reduced Hatree-Fock (rHF) model. This model has been rigorously derived from the rHF model for finite molecular systems by means of thermodynamic limit procedure by Catto, Le Bris and Lions~\cite{Catto2001}. Later, Cancès, Deleurence and Lewin~\cite{Cances2008} proved that this model is also the limit of the supercell rHF model (see Section~\ref{sec:supercellRHF}). 
\medskip

For $\bR \in \Lat$, we denote by $\tau_\bR$ the translation operator by the vector $\bR$: $\tau_\bR f(\bx) = f(\bx - \bR)$. We introduce the set of admissible density matrices
\begin{equation} \label{eq:VTr1}
	\cP := \left\{ \gamma \in \cS(L^2(\R^3)), \ 0 \le \gamma \le 1, \ \forall \bR \in \Lat, \ \tau_\bR \gamma = \gamma \tau_\bR , \ \VTr \left( \gamma \right) +\VTr\left( -\Delta\gamma \right)< \infty \right\},
\end{equation}
where $\VTr(-\Delta \gamma)$ is a short-hand notation for $\sum_{j=1}^3 \VTr(P_j \gamma P_j)$, $\cS(\cH)$ denotes the space of the bounded self-adjoint operators on the Hilbert space $\cH$ and $\VTr$ denotes the trace per unit volume, defined for any locally trace class operator $A$ that commutes with $\Lat$-translations, by (see also Equation~\eqref{eq:VTr} below for an alternative definition)
\begin{equation} \label{eq:VTr_def1}
\VTr\left(  A \right):= \lim_{L \to \infty} \dfrac{1}{L^3} \Tr\left( \1_{L\Gamma} A \1_{L\Gamma} \right).
\end{equation}
Any $\gamma \in \cP$ is locally trace-class, and can be associated an $\Lat$-periodic density $\rho_{\gamma} \in L^2_\per(\WS)$. For $\gamma \in \cP$, the reduced Hartree-Fock energy is given by
\begin{equation} \label{eq:cE_per}
	\cE_{\mu_\per}(\gamma) := \dfrac12 \VTr \left( -\Delta \gamma \right) + \dfrac12 D_1 \left( \rho_\gamma - \mu_\per, \rho_\gamma - \mu_\per \right).
\end{equation}
The first term of~\eqref{eq:cE_per} corresponds to the kinetic energy, and the second term represents the Coulomb energy per unit volume. To describe the latter term, we introduce the Green kernel of the $\Lat$-periodic Poisson equation \cite{Lieb1977}, denoted by $G_1$ and satisfying the equation
\[
	\left\{ \begin{array}{c}
		\displaystyle -\Delta G_1 =  4 \pi \left( \sum_{\bk \in \Lat} \delta_\bk - \left| \Gamma \right|^{-1} \right) \\
		G_1 \ \text{ is $\Lat$-periodic.}
		\end{array} \right. 
\]
The expression of $G_1$ is given in the Fourier basis by
\begin{equation} \label{eq:G1Fourier}
	G_1(\bx) = c_1 + \dfrac{4 \pi}{| \WS |} \sum_{\bk \in \RLat \setminus \{ \bnull\} } \dfrac{\re^{\ri \bk \cdot \bx}}{| \bk |^2},
\end{equation}
where $c_1 = | \WS |^{-1} \int_\WS G_1$ can be \textit{a priori} any fixed constant. In one of the first article on the topic~\cite{Lieb1977}, the authors chose to set $c_1 = 0$, but other choices are equally valid (see~\cite{Cances2008} for instance). We will set $c_1 = 0$ for simplicity, and highlight the role of $c_1$ in the main results (see Remark~\ref{rem:role_of_cL}). The Coulomb energy per unit volume is defined, for $f, g \in L^2_\per(\WS)$, by
\begin{equation} \label{eq:defD1}
	D_1 (f,g) := \int_\WS (f\ast_\Gamma G_1) (\bx) g(\bx) \rd \bx = 4 \pi \sum_{\bk \in \RLat \setminus \{ \bnull \} } \dfrac{\overline{c_\bk(f)} c_\bk(g)}{| \bk |^2}.
\end{equation}
where $(f\ast_\Gamma G_1) (\bx):= \int_\WS f(\by)G_1(\bx-\by) \rd \by$. \\
Finally, the periodic rHF ground state energy is given by
\begin{equation*} 
	 \inf \left\{ \cE_{\mu_\per}(\gamma), \ \gamma \in \cP_\per, \ \int_\WS \rho_\gamma = \int_\WS \mu_\per \right\}.
\end{equation*}
It has been proved in \cite[Theorem 1]{Cances2008} that this minimization problem admits a unique minimizer $\gamma_0$, which is the solution to the self-consistent equation
\begin{equation} \label{eq:sc}
	\left\{ \begin{array}{lll}
		\gamma_0 & = & \mathds{1}(H_0\leq \varepsilon_F) \\
		 H_{0} & = & - \frac12 \Delta + V_0 \\
		V_0 & = & (\rho_{\gamma_0} - \mu_\per) \ast_\WS G_1.
		\end{array} \right.
\end{equation}
Here, $\varepsilon_F$, called the Fermi level or the Fermi energy, is the Lagrange multiplier corresponding to the charge constraint $\int_\WS \rho_{\gamma_0} = \int_\WS \mu_\per$. Throughout the article, we make the following assumption:
\[
	\boxed{ \text{\textbf{(A1)} The system is an insulator, in the sense that $H_0$ has a spectral gap around $\varepsilon_F$.}}
\]
Without loss of generality, we assume that $\varepsilon_F$ is located in the middle of the gap, and we denote by $g > 0$ the size of the gap. In this case, the minimizer $\gamma_0$ defined in~\eqref{eq:sc} is also the unique minimizer of the grand canonical ensemble problem (see~\cite[Theorem 1]{Cances2008})
\begin{equation} \label{eq:Iper}
	I_{\mu_\per} := \inf \left\{ \cE_{\mu_\per}(\gamma) - \varepsilon_F \VTr(\gamma) , \ \gamma \in \cP \right\}.
\end{equation}


\subsection{Assumption on the defect}
\label{sec:defect}

In this article, we are interested in cases where the crystal contains a local defect. The nuclear charge distribution of such a system is taken of the form $\mu_\per + \nu$, where $\nu$ is a compactly supported function that models the nuclear charge of the local defect. More specifically, we fix $L^\supp \in \N^*$, and consider defects $\nu \in L^2(\R^3)$ with support in $L^\supp\WS$. We introduce, for $\eta > 0$, the set
\begin{equation} \label{eq:ball}
	\cN( \eta) := \left\{ \nu \in L^2(\R^3), \ \nu \mathds{1}_{L^\supp \WS} = \nu, \ \| \nu \|_{L^2(\R^3)} \le \eta \right\}.
\end{equation}
Other types of defects may be considered, but they add extra mathematical complications, and do not provide any further insight on the nature of our main result.


\subsection{The rHF energy of a local defect}

The rHF model of a local defect embedded into a reference perfect crystal was introduced and studied in~\cite{Cances2008}. The main idea of this article is to decompose the ground state density matrix as
\begin{align}\label{eq:local_defect}
	\gamma=\gamma_0+Q_\nu,
\end{align}
and rewrite the formal minimization problem in a problem where the variable is $Q_\nu$. The authors have then proved that the corresponding energy was indeed the limit of the supercell rHF energy of the defect (see Section~\ref{sec:supercellRHFDefects}). Following~\cite{Cances2008}, we introduce the convex set
\begin{equation}\label{eq:K}
	\begin{array}{c}
	 \cK=\left\{Q \in \cS(L^2(\R^3)),\ -\gamma_0 \leq Q \leq 1-\gamma_0 , \ (1-\Delta)^\frac12 Q \in \fS_2(L^2(\R^3)),\right.\\
		\left. \ (1 -\Delta)^\frac12 Q^{\pm\pm} (1-\Delta)^\frac12 \in \fS_1(L^2(\R^3)) \right\},
	\end{array}
\end{equation}
where $Q^{++}:=(1-\gamma_0)Q(1-\gamma_0)$, $Q^{--}:=\gamma_0 Q \gamma_0$,
and $\fS_p$ denotes the $p$-th Schatten class~\cite{Simon2005}. In particular $\fS_1$ is the set of trace class operators and $\fS_2$ is the set of Hilbert-Schmidt operators. It has been proved in~\cite{Cances2008} that although a generic $Q$ in $\cK$ in not trace class, it can be associated a density $\rho_Q \in \cC(\R^3)\cap L^2(\R^3)$, where $\cC(\R^3)$ is the Coulomb space 
\begin{equation} \label{eq:CoulombSpace}
	\cC := \left\{ f \in \mathscr{S}'(\R^3), \ \| f \|_{\cC(\R^3)}^2 := D(f,f) := 4 \pi \int_{\R^3} \dfrac{ | \widehat{f}(\bk)| ^2}{ |\bk|^2} \rd \bk < \infty\right\},
\end{equation}
where $\sS'$ is the Schwartz space of tempered distributions, and $\widehat{f}$ is the normalized Fourier transform of $f$, defined for $f \in L^1(\R^3)$ by
\[
	\widehat{f}(\bk) := \frac{1}{(2\pi)^{3/2}}\int_{\R^3} f(\bx) \re^{-\ri \bk \cdot \bx} \rd \bx.
\]
It holds $L^{6/5}(\R^3) \hookrightarrow \cC$, and taking $L^2(\R^3)$ as the pivoting space, the dual of $\cC$ is the Beppo-Levi space
\begin{equation} \label{eq:BLSpace}
	\cC' := \left\{ V \in L^6(\R^3), \ \nabla V \in \left( L^2(\R^3) \right)^3 \right\},
\end{equation}
which can also be seen as an Hilbert space when endowed with the inner product $\bra f , g \ket_{\cC'} := \int_{\R^3} \overline{\nabla f} \cdot \nabla g$.

For $Q \in \cK$, we introduce
\begin{equation}\label{eq:defectEnergy} 
	\cF_\nu(Q)=\Tr  \left( | H_0-\varepsilon_F |  (Q^{++}-Q^{--}) \right) +\frac{1}{2}D(\rho_Q-\nu,\rho_Q-\nu) - \int_{\R^3} V_0 \nu,
\end{equation} 
where $H_0$ and $V_0$ were both defined in~\eqref{eq:sc}. Note that $\nu \in \cC$ since $L^2(L^\supp \WS) \hookrightarrow L^{6/5} (L^\supp \WS) \hookrightarrow  L^{6/5}(\R^3) \hookrightarrow \cC$, so that the Coulomb energy in~\eqref{eq:defectEnergy} is well-defined. The rHF energy of the defect $\nu$ is then defined by the minimization problem
\begin{equation} \label{eq:Jnu}
	J_\nu := \inf  \left\{ \cF_\nu(Q), \ Q \in \cK \right\}.
\end{equation}
The existence of minimizers for this problem was proved in~\cite[Theorem 2]{Cances2008}\cite[Lemma 5]{Cances2010}.

\begin{theorem}[Existence of ground states for crystals with local defects]
\label{th:Qnu}
Assume that \textbf{(A1)} holds true. There exists $\eta > 0$ such for all $\nu \in \cN(\eta)$, the following holds. There exists a unique minimizer $Q_\nu \in \cK$ of the problem~\eqref{eq:Jnu}. Moreover, this minimizer satisfies the equality $Q_\nu^2 = Q_\nu^{++} - Q_\nu^{--}$,  the equality $\Tr ( Q_\nu^{++}+Q_\nu^{--} ) = 0$ and the Euler-Lagrange equation

\begin{equation}\label{eq:SCF_defect}
 \left\{
	\begin{array}{lll}
		Q_\nu & =& \mathds{1} ( H_\nu \leq \varepsilon_F) - \gamma_0 ,\\
 		H_\nu & =& H_0+ V_\nu ,\\
		V_\nu & = & (\rho_{Q_\nu}-\nu) \ast \dfrac{1}{ | \cdot |}.
	\end{array} \right.
\end{equation}
Finally, the operator $H_\nu$, which acts on $L^2(\R^3)$, is gapped around $\varepsilon_F$ with $\left| H_\nu - \varepsilon_F \right| \ge g/2$.

\end{theorem}

\subsection{The supercell rHF model for perfect crystals}
\label{sec:supercellRHF}

The rHF model for crystals with and without local defects is the thermodynamic limit~\cite{Cances2008} of the supercell rHF model, where the system is confined to a box $\WS_L := L \WS$ with periodic boundary conditions. The corresponding lattice is $\Lat_L := L \Lat$, and the corresponding reciprocal lattice is $\RLat_L := L^{-1} \RLat$. We denote by $L^2_\per(\WS_L)$ the Hilbert space of locally square integrable functions that are $L\Lat$-periodic. The normalized Fourier coefficients of a function $f \in L^2_\per ( \WS_L)$ are defined by
\[
	\forall \bk \in L^{-1} \RLat , \quad c_\bk^L (f) =\frac{1}{| \Gamma_L|^{1/2}} \int_{\Gamma_L} f(\bx) \re^{-\ri \bk \cdot \bx} \rd \bx,
\]
so that,
\[
\| f \|_{L^2_\per(\WS_L)}^2 = \sum_{\bk \in L^{-1}\RLat} |c_\bk^L | ^2
\quad \text{and} \quad
	f(\bx) = \frac{1}{| \Gamma_L|^{1/2}} \sum_{\bk \in L^{-1} \RLat} c_\bk^L(f) \re^{\ri \bk \cdot \bx} \quad \text{ a.e. and in } L^2_\per(\WS_L).
\]

We denote by $P_{j}^L$, $j\in\left\{1,2,3\right\}$, the self-adjoint operator on $L^2_\per (\Gamma_L)$ defined by $c_\bk^L(P_{j}^L f)= k_j c_\bk^L(f)$, and by $(-\Delta^L) := \sum_{j=1}^3 \left( P_{j}^L \right)^2$, so that $c_\bk^L(-\Delta^L f)= | \bk |^2 c_\bk^L(f)$. The set of admissible electronic states for the supercell model is
\[
	\cP^L := \left\{ \gamma^L \in \cS(L^2_\per(\WS_L)), \ 0 \le \gamma^L \le 1,  \  \Tr_{L^2_\per(\WS_L)} \left(  \gamma^L \right) +\Tr_{L^2_\per(\WS_L)}\left(  -\Delta^L \gamma^L \right) < \infty \right\},
\]
where $\Tr_{L^2_\per(\WS_L)}\left(  -\Delta^L \gamma^L \right)$ is a shorthand notation for 
$\sum_{j=1}^3\Tr_{L^2_\per(\WS_L)}\left(  P_{j}^L \gamma^L P_{j}^L \right)$. Any $\gamma^L \in \cP^L$ is locally trace-class, and can be associated a density $\rho_{\gamma^L} \in L^2_\per(\WS_L)$. 

\medskip

For a charge density $\mu \in L^2_\loc(\R^3)$, we denote by $ \mu_L $ the $L \Lat$-periodic nuclear distribution which is equal to $\mu$ on $\Gamma_L$, and by $\cE_{\mu}^L$ the energy functional defined on $\cP^L$ by 
\begin{equation} \label{eq:EL}
	\cE^L_\mu(\gamma^L) :=  \dfrac12  \Tr_{L^2_\per(\WS_L)} (- \Delta^L \gamma^L) + \dfrac12 D_L(\rho_{\gamma^L} - \mu_L, \rho_{\gamma^L} - \mu_L).
\end{equation}
The first term corresponds to the kinetic energy, and the second to the supercell Coulomb energy. To define the latter one, we introduce the $L \Lat$-periodic Green kernel $G_L$ of the Poisson interaction, solution of
\[
	\left\{ \begin{array}{c}
		-\Delta G_L =  4 \pi \left( \sum\limits_{\bk \in L \Lat} \delta_\bk - \left| \WS_L \right|^{-1} \right) \\
		G_L \text{ is } L\Lat\text{-periodic.}
		\end{array} \right. 
\]
The expression of $G_L$ is given in the Fourier basis by
\begin{equation} \label{eq:GLFourier}
	G_L (\bx) = c_L + \frac{4 \pi}{| \WS_L|} \sum_{\bk \in L^{-1} \RLat \setminus \{ \bnull \}} \dfrac{\re^{\ri \bk \cdot \bx}}{| \bk |^2},
\end{equation}
where $c_L = | \WS_L|^{-1} \int_{\WS_L} G_L$. As in~\eqref{eq:G1Fourier}, the constant $c_L$ can be any fixed constant, and we choose to set $c_L = 0$ for simplicity. The supercell Coulomb energy is defined, for $f,g \in \L^2_\per(\WS_L)$ by
\begin{equation} \label{eq:defDL}
	D_L (f,g) := \int_{\WS_L} (f\ast_{\Gamma_L} G_L) (\bx) g(\bx) \rd \bx = 4 \pi \sum_{\bk \in L^{-1}\RLat \setminus \{ \bnull \} } \dfrac{\overline{c_\bk^L(f)} c_\bk^L(g)}{| \bk |^2}.
\end{equation}

\medskip

Finally, the supercell rHF energy of the system in the box of size $L$ with Fermi level $\varepsilon_F$ is given by the minimization problem 
\begin{equation} \label{eq:finite_problem}
	I_{\mu}^L = \inf \left\{ \cE_{\mu}^L (\gamma^L) - \varepsilon_F \Tr_{L^2_\per(\WS_L)} (\gamma^L), \, \gamma^L \in \cP^L\right\}.
\end{equation}

\begin{remark}
	This problem is set on the grand canonical ensemble. The Fermi level $\varepsilon_F$ is the one defined in the previous section (see Assumption~\textbf{(A1)}). This model is therefore different than the ones usually considered in numerical codes, where the charge of $\gamma^L$ is constrained. It is unclear to us what is the relationship between the two models when $\mu \neq \mu_\per$.
\end{remark}

When $\mu = \mu_\per$, then $\mu_\per = \mu_{\per,L}$ for all $L \in \N^*$. In this case, we can state precise results. The following theorem was proved in~\cite[Theorem 4]{Cances2008} and~\cite[Theorem 3.3]{GL2015}.

\begin{theorem}[Thermodynamic limit for perfect crystals]
\label{th:expCV}
	If $\mu = \mu_\per$, then there exists $L^\gap \in \N^*$ such that the following holds true. For $L \ge L^\gap$, the minimization problem $I^L_{\mu_\per}$ defined in~\eqref{eq:finite_problem} has a unique solution $\gamma_{0}^L \in \cP^L$ which commutes with $\Lat$-translations. This minimizer 
	satisfies the self-consistent equation
	\begin{equation} \label{eq:scL}
		\left\{ \begin{aligned}
			\gamma_{0}^L & = \mathds{1} (H_{0}^L < \varepsilon_F) \\
			H_{0}^L & = - \frac12 \Delta^L + V_{0}^L \\
			V_{0}^L & = (\rho_{\gamma_{0}^L} - \mu_\per) \ast_{\WS_L} G_L.
		\end{aligned}
		\right.
	\end{equation}
	Moreover, there exist $C \in \R^+$ and $\alpha > 0$ such that,
	\[
		\forall L \ge L^\gap, \quad \left| I_{\mu_\per} - L^{-3} I_{\mu_\per}^L \right| \le C \re^{-\alpha L} 
		\quad \text{and} \quad
		\left\| \rho_{\gamma_0} - \rho_{\gamma_{0}^L} \right\|_{L^\infty_\per(\WS)} \le C \re^{- \alpha L}.
	\]
	Finally, for $L \ge L^\gap$, the operator $H_{0}^L$, which acts on $L^2_\per(\WS_L)$, satisfies $ \left| H_{0}^L - \varepsilon_L \right| > 3g/4$.
\end{theorem}

\subsection{The supercell rHF energy of a defect}
\label{sec:supercellRHFDefects}

When $\mu \neq \mu_\per$, we have much weaker results. In the case of local defects where $\mu = \mu_\per + \nu$, we can treat $\nu$ as a perturbation of defect-free case. We obtain the following theorem, whose proof is skipped for brevity.
\begin{theorem} \label{th:sc_nu_L}
	There exists $\eta > 0$ and $L^* \in \N$ such that, for all $\nu \in \cN(\eta)$, the following holds. For all $L \ge L^*$, the problem $I^L_{\mu_\per + \nu}$ has a unique minimizer $\gamma_{\nu}^L$. This minimizer satisfies $\Tr_{L^2_\per(\WS_L)}(\gamma_\nu^L) = \Tr_{L^2_\per(\WS_L)}(\gamma_0^L)$, and is the solution to the self-consistent equation
	\begin{equation} \label{eq:sc_nu_L}
		\left\{ \begin{aligned}
			\gamma_{\nu}^L & = \mathds{1} (H_{\nu}^L < \varepsilon_F) \\
			H_{\nu}^L & = - \frac12 \Delta^L + V_{\nu}^L \quad \text{acts on } L^2_\per(\WS_L) \\
			V_{\nu}^L & = (\rho_{\gamma_{\nu}^L} - \mu_\per - \nu_L) \ast_{\WS_L} G_L.
		\end{aligned}
		\right.
	\end{equation}
	Moreover, the operator $H_{\nu}^L$ in~\eqref{eq:sc_nu_L} satisfies $\left| H_{\nu}^L - \varepsilon_F \right| \ge g/2$.
\end{theorem}

In the sequel, we take $\eta > 0$ and $L^* \in \N$ as in Theorem~\ref{th:sc_nu_L}.
Without loss of generality, we may choose $L^* \ge \max (L^{\supp}, L^\gap)$. For all $\nu \in \cB(\eta)$, we supercell rHF energy of the defect is defined by
\begin{equation} \label{eq:JnuL}
	J^L_\nu := I^L_{\mu_\per + \nu} - I^L_{\mu_\per}.
\end{equation}

In~\cite[Theorem 5]{Cances2008}, the authors proved that $J_\nu$ defined in~\eqref{eq:Jnu} is the limit of $J^L_\nu$ as $L$ goes to infinity. The purpose of this article is to identify the speed of this convergence. Before stating our results, let us first rewrite $J^L_\nu$ as a minimization problem.

We introduce the convex set
\begin{equation}\label{eq:K}
	\begin{array}{cl }
	 \cK^L = & \left\{ Q^L \in \cS(L^2_\per(\WS_L)), \ -\gamma_0^L \leq Q^L \leq 1-\gamma_0^L , \right.\\
		& \left. \ (1 -\Delta^L)^\frac12 Q^L (1-\Delta^L)^\frac12 \in \fS_1(L^2_\per(\WS_L)) \right\}.
	\end{array}
\end{equation}
For $Q^L \in \cK^L$, we define the supercell rHF energy of $Q^L$ by
\[
	\cF_\nu^L(Q^L) := \Tr_{L^2_\per(\WS_L)} \left( \left[ H_0^L -\varepsilon_F \right]  Q^L \right) + \frac12 D_L \left( \rho_{Q^L} - \nu_L, \rho_{Q^L} - \nu_L \right) - \int_{\WS_L} V_0^L \nu.
\]
A little algebra shows that $J^L_\nu$ is also the solution of the minimization problem
\begin{equation} \label{eq:JnuL_infQ}
	J^L_\nu = \inf \left\{ \cF_\nu^L(Q^L), \ Q^L \in \cK^L\ \right\},
\end{equation}
and that the unique minimizer of this problem is $Q_\nu^L := \gamma_\nu^L - \gamma_0^L$. \\

\subsection{Main results}

Our main result is the following.
\begin{theorem}[Convergence rate of the defect energy]\label{th:chargedDefects}
Suppose that \textbf{(A1)} holds true. There exist $\eta > 0$, $\alpha > 0$, $L^* \in \N^*$ and $C \ge 0$, such that
\begin{equation} \label{eq:Jnu2JnuL}
	\forall \nu \in \cN( \eta ), \quad \forall L \ge L^*, \quad
	J_\nu =J_\nu^L - \dfrac{1}{L} \frac{2 \pi \fa q^2}{| \WS |} + K(L, \nu),
\end{equation}
with
\begin{equation} \label{eq:K}
	\forall \nu \in \cN( \eta), \quad \forall L \in L^*, \quad 
	\left| K(L, \nu) \right| \le C \left( \| \nu \|_{L^2(\R^3)}^3 + \dfrac{\| \nu \|_{L^2(\R^3)}^2}{L^3} + \| \nu \|_{L^2(\R^3)} \re^{- \alpha L} \right).
\end{equation}
Here, $q := \int_{\R^3} \nu$ is the nuclear charge of the defect, and $\fa$ is defined by
\begin{equation} \label{eq:fa}
	\fa := \sum_{\bk \in \RLat } \fint_{\BZ} \left( \frac{1 }{( \bk + \bq)^T M (\bk + \bq) } - \dfrac{ \mathds{1} (\bk \neq \bnull) }{ \bk^T M \bk} \right) \rd \bq,
\end{equation}
where $\fint_\BZ := | \BZ |^{-1} \int_\BZ$, and $M$ is the macroscopic dielectric $3 \times 3$ matrix of the crystal (see Definition~\ref{def:A1} below).
\end{theorem}

\begin{remark}~\label{rem:role_of_cL}
	If we choose not to set $c_L = 0$ in~\eqref{eq:G1Fourier} and~\eqref{eq:GLFourier}, then~\eqref{eq:Jnu2JnuL} becomes
	\[
		J_\nu =J_\nu^L - c_L \frac{ q^2}{2}  - \dfrac{1}{L} \frac{2 \pi \fa q^2}{| \WS |} + \widetilde{K}(L, \nu),
	\]
	where $\widetilde{K}$ satisfies estimates similar than~\eqref{eq:K}.
	We refer to Remark~\ref{rem:theExtraTerm} for the origin of the extra term $-c_L q^2/2$. The convergence rate therefore depends on the choices of $c_L$. 
\end{remark}

The proof of Theorem~\ref{th:chargedDefects} is detailed in the following sections. In practice, the sum appearing in~\eqref{eq:fa} can be evaluated using Ewald summation~\cite{Ewald1921}. When the crystal is isotropic cubic, the expression of $\fa$ can be simplified.

\begin{definition}[Isotropic cubic crystal] \label{def:isotropic_cubic}
Let
\[
		S_1 =  \begin{pmatrix} 0 & 1 & 0 \\ 0 & 0 & 1 \\ 1 & 0 & 0 \end{pmatrix} \quad \text{and} \quad S_2 = \begin{pmatrix} -1 & 0 & 0 \\ 0 & 1 & 0 \\ 0 & 0 & 1 \end{pmatrix}.
\]
	We say that the crystal is isotropic cubic if $\mu_\per(S_1 \bx) = \mu_\per(S_2 \bx) = \mu_\per(\bx)$, for all $\bx \in \R^3$.
\end{definition}
A necessary condition for a crystal to be isotropic cubic is that $\Lat = a \Z^3$ for some $a >0$. 
\begin{proposition} \label{prop:isotropic}
If the crystal is isotropic cubic with $\Lat = a \Z^3$, then $M$ is proportional to the identity matrix with $M = \epsilon \bbI_3$, where $\epsilon \ge 1$ is the macroscopic dielectric constant of the crystal. In this case, it holds that
\begin{equation} \label{eq:MadelungCst}
	\fa = - 2 \pi^2 \dfrac{\fm}{\epsilon | \BZ |},
\end{equation}
where $\fm$ is what the physicians call the Madelung constant, defined by
\begin{equation} \label{eq:def:Madelung}
	\fm := \lim_{\bx \to \bnull} G_1(\bx) - \frac{1}{| \bx |}.
\end{equation}
Here $G_1$ is the Green kernel defined~\eqref{eq:G1Fourier} with $c_1 = 0$.
%
%
%
\end{proposition}
The proof of Proposition~\ref{prop:isotropic} is given in Section~\ref{sec:proof:isotropic}. In~\cite{Lieb1977}, the authors defined another Madelung constant $\fm'$, defined by
\[
	\fm' = \lim_{\bx \to \bnull} F(\bx) - \frac{1}{| \bx |}
	\quad \text{where} \quad
	F(\bx) = \sum_{\bR \in \Lat} f(\bx - \bR)
	\quad \text{and} \quad
	f(\bx) = \dfrac{1}{| \bx |} - \dfrac{1}{| \WS |} \int_{\WS} \dfrac{\rd \by}{| \bx - \by| }.
\]
These two constants are linked by the relation $\fm' = \fm + 2 \pi |\WS |^{-2} \int_\WS \bx^2 \rd \bx$ (see~\cite[Equation (126)]{Lieb1977}).

\medskip

In the isotropic cubic case,~\eqref{eq:Jnu2JnuL} therefore becomes
\[
	\forall \nu \in \cN( \eta ), \quad \forall L \ge L^*, \quad
	J_\nu =J_\nu^L + \dfrac{1}{L} \frac{\fm q^2}{2 \epsilon} + K(L, \nu).
\]
We therefore recover in~\eqref{eq:Jnu2JnuL} the $L^{-1}$ correction term predicted by Leslie and Gillan~\cite{Leslie1985}, and by Makov and Payne~\cite{Makov1995}. In the non-isotropic case, the definition~\eqref{eq:fa} for $\fa$ in was already proposed without proof in~\cite{Murphy2013}.

\medskip

The rest of the paper is devoted to the proof of Theorem~\ref{th:chargedDefects} and Proposition~\ref{prop:isotropic}. The global strategy of the proof is as follows. We will first isolate the linear, quadratic and higher-order terms of $\nu$ in the functionals $J_\nu$ and $J_\nu^L$. We will then prove that the difference coming from the linear part decays exponentially fast with respect to $L$. To study the difference between the quadratic parts, we will Bloch decompose some response operators. We rewrite the difference as a difference between a Riemann sum and a corresponding integral. We finally prove that the $O(L^{-1})$ speed of convergence comes from the lack of regularity of the integrand.


\section{The linear and quadratic contributions of the defect}
\label{sec:linear}

\subsection{The supercell Coulomb operator}

Recall that the spaces $\cC$ and $\cC'$ were defined respectively in~\eqref{eq:CoulombSpace} and~\eqref{eq:BLSpace}. We introduce the Coulomb operator on the whole space $v_c : \cC \to \cC'$, defined in Fourier by
\begin{equation} \label{eq:def:vc}
	\forall \bk \in \R^3 \setminus \{ \bnull \}, \quad \widehat{v_c(f)}(\bk) = 4 \pi \dfrac{\widehat{f}(\bk)}{| \bk|^2}.
\end{equation}
We also introduce the operator $\sqrt{v_c}$ defined in Fourier by
\begin{equation} \label{eq:def:sqrt_vc}
	\forall \bk \in \R^3 \setminus \{ \bnull \}, \quad \widehat{\sqrt{v_c}(f)}(\bk) = \sqrt{4 \pi} \dfrac{\widehat{f}(\bk)}{| \bk|}.
\end{equation}
The following lemma is straightforward.
\begin{lemma} \label{lem:vc}
	The operator $v_c$ is unitary from $\cC$ to $\cC'$, with $v_c^* = v_c^{-1}$. The operator $\sqrt{v_c}$ is unitary from\footnote{To be accurate, one should introduce two operators $\sqrt{v_c}_1$ and $\sqrt{v_c}_2$ with domain $\cC$ and $L^2(\R^3)$ respectively. However, in order to keep the notation simple, we use the same name for these two operators.} $\cC$ to $L^2(\R^3)$ and from $L^2(\R^3)$ to $\cC'$. Finally, it holds $v_c = \sqrt{v_c} \sqrt{v_c}$.
\end{lemma}
For $f, g \in \cC$, it holds that $D(f,g)$ defined in~\eqref{eq:CoulombSpace} is also equal to
\begin{equation} \label{eq:D=vc}
	D(f, g) = \bra f,  v_c(g)\ket_{\cC, \cC'} =\bra v_c(f),  g \ket_{\cC', \cC} = \bra \sqrt{v_c} (f), \sqrt{v_c} (g) \ket_{L^2(\R^3)}.
\end{equation}

\medskip

In order to introduce the supercell equivalent to these objects, we introduce the space of $L^2_\per(\WS_L)$ functions with null mean-value:
\begin{equation*}
	L^2_{0,\per}(\WS_L) = \left\{ f \in L^2_\per(\WS_L), \ c_\bnull^L(f) = 0 \right\}.
\end{equation*}
It is an Hilbert space when endowed with the $L^2_\per(\WS_L)$ inner product. The supercell Coulomb space is
\begin{equation} \label{eq:supercellCoulombSpace}
	\cC_L := \left\{ f \in \sS_\per'(\WS_L), \ c_\bnull^L(f) = 0, \ \| f \|_{\cC_L}^2 := 4 \pi \sum_{\bk \in L^{-1} \RLat \setminus \{ \bnull \}} \dfrac{| c_\bk^L(f) |^2}{| \bk |^2} < \infty \right\}.
\end{equation}
It is a Hilbert space when endowed with the $D_L( \cdot, \cdot)$ inner product defined in~\eqref{eq:defDL}. Taking $L^2_{0,\per}(\WS_L)$ as the pivoting space, its dual is the supercell Beppo-Levi space
\begin{equation} \label{eq:supercellBLSpace}
	\cC'_L := \left\{ f \in L^6_\per(\WS_L), \ c_\bnull^L(f) = 0, \ \| f \|_{\cC'_{L}}^2 := \dfrac{1}{4 \pi} \sum_{\bk \in L^{-1} \RLat} | \bk |^2 {| c_\bk^L(f) |^2} < \infty \right\},
\end{equation}
which can also be seen as an Hilbert space with its natural inner product. Finally, the supercell Coulomb operators $v_{c}^L : \cC_{L} \to \cC_L'$ is defined in Fourier by
\begin{equation} \label{eq:def:vcL}
	 \forall \bk \in L^{-1}\RLat \setminus \{ \bnull \}, \quad
	  c^L_{\bk} \left(v_{c}^L(f) \right) = 4 \pi \dfrac{c_{\bk}^L(f)}{| \bk|^2}
	 \quad \text{and} \quad
	 c^L_\bnull \left( v_{c}^L(f)  \right) = 0.
\end{equation}
We also introduce the operators $\sqrt{v_{c}^L}$ defined in Fourier by
\begin{equation} \label{eq:def:sqrt_vcL}
	 \forall \bk \in L^{-1}\RLat \setminus \{ \bnull \}, \quad
	  c^L_{\bk} \left(\sqrt{v_{c}^L}(f) \right) = \sqrt{4 \pi} \dfrac{c_{\bk}^L(f)}{| \bk|}
	 \quad \text{and} \quad
	 c^L_\bnull \left( \sqrt{v_{c}^L}(f)  \right) = 0.
\end{equation}
\begin{lemma} \label{lem:vcL}
	The operators $v_c^L$ are unitary from $\cC_L$ to $\cC'_L$, so that $(v_c^L)^\ast = (v_c^L)^{-1}$. The operators $\sqrt{v_{c}^L}$ are unitary from $\cC_L$ to $L^2_{0,\per}(\WS_L)$ and from $L^2_{0,\per}(\WS_L)$ to $\cC_L'$. Finally, it holds $v_c^L = \sqrt{v_{c}^L} \sqrt{v_{c}^L}$.
\end{lemma}

\begin{remark} \label{rem:sqrtVc}
	It is convenient to extend the definition of $\sqrt{v_c^L}$ to $L^2_\per(\WS_L)$ by formally dropping the $\bnull$-Fourier coefficient. In this case, $\sqrt{v_c^L}$ is no longer unitary, but is still bounded by~$1$.
\end{remark}

For $f,g \in \cC_L$, it holds
\begin{equation} \label{eq:normC=normC'}
	D_L(f,g) = \bra f, v_c^L(g) \ket_{\cC_L, \cC_L'} = \bra v_c^L(f), g \ket_{\cC_L', \cC_L} =  \left\bra \sqrt{v_c^L}(f), \sqrt{v_c^L} (g) \right\ket_{L^2_{0,\per}(\WS_L)}.
\end{equation}

\begin{remark} \label{rem:theExtraTerm}
	If we choose not to set $c_L = 0$ in~\eqref{eq:GLFourier}, then~\eqref{eq:normC=normC'} should be modifid to read
	\[
		D_L(f,g) = \left\bra \sqrt{v_c^L}(f), \sqrt{v_c^L} (g) \right\ket_{L^2_{0,\per}(\WS_L)}+c_L  | \WS_L | \overline{c_\bnull^L(f)} c_\bnull^L(g).
	\]
	Together with the fact that for $\nu \in \cN(\eta)$ and $L \ge L^*$, it holds $c_\bnull^L(\nu) = \dfrac{q}{| \WS_L |^{1/2}}$,
	this eventually explains the extra term in Remark~\ref{rem:role_of_cL}.
\end{remark}

\subsection{First properties of $Q_\nu^L$ and $Q_\nu$}
In this section, we study the minimizers $Q_\nu$ and $Q_\nu^L$ of~\eqref{eq:Jnu} and~\eqref{eq:JnuL} respectively. In particular, we identify the linear and quadratic contribution of $\nu$ in these minimizers. We only state the results for $Q_\nu^L$, in order to emphasize the dependence in $L$ of our bounds. Similar results hold for $Q_\nu$, but we will not enunciate them for brevity. 

\medskip

We let $\eta > 0$ and $L^\ast \in \N^*$ be as in Theorem~\ref{th:sc_nu_L}, and we introduce $\sC$ a simple positively oriented loop that encloses the spectrum of the operators $H_{\nu}^L$ below $\varepsilon_F$ for all $\nu \in \cN(\eta)$ and all $L \ge L^*$. This is possible thanks to the last property of Theorem~\ref{th:sc_nu_L}. Let $\Sigma \in \R$ be such that
\[
	\forall \nu \in \cN(\eta), \quad  \Sigma \le \inf \sigma \left( H_\nu \right) -1  \quad \text{and} \quad \forall L > L^\ast, \quad \Sigma \le \inf \sigma \left( H_\nu^L \right) -1,
\]
we take $\sC = \sC_1 \cup \sC_2 \cup \sC_3 \cup \sC_4$, where $\sC_1 = [ \varepsilon_F - \ri, \varepsilon_F + \ri]$, $\sC_2 = [ \varepsilon_F +\ri, \Sigma  + \ri]$, $\sC_3 = [\Sigma + \ri , \Sigma  - \ri]$ and $\sC_4 = [\Sigma-\ri, \varepsilon_F - \ri]$. 

\begin{figure}[h!]
\begin{center}
\begin{tikzpicture}
	\fill[red!80] (-3, -0.1) rectangle (0, 0.1);
	\fill[red!80] (2, -0.1) rectangle (5, 0.1);
	\draw[line width=2, red!80] (-3.2, -0.1) -- (-3.2, 0.1);
	\draw[line width=2, red!80] (0.5, -0.1) -- (0.5, 0.1);
	\draw[line width=2, red!80] (0.2, -0.1) -- (0.2, 0.1);
	\draw[line width=2, red!80] (1.5, -0.1) -- (1.5, 0.1);
	\draw[line width=2, red!80] (1.7, -0.1) -- (1.7, 0.1);
	
	\draw[->] (-5, 0) -> (6, 0);
	
	\draw[line width=2] (-4, 0.2) -- (-4, -0.2);
	\node at (-4.2, -0.5) {$\Sigma$};
	
	\draw (1, -1) -- (1, 1) -- (-4, 1) -- (-4, -1) -- (1, -1);
	\draw (-1.8, 1.2) -- (-2, 1) -- (-1.8, 0.8);
	\draw (-2.2, -1.2) -- (-2, -1) -- (-2.2, -0.8);

	\node at (1.3, -0.5) {$\varepsilon_F$};
	\node at (3, 0.5) {$\sigma(H_\nu)$};
	
	\node at (1.3, 0.5) {$\sC_1$};
	\node at (-1, 1.4) {$\sC_2$};	
	\node at (-4.5, 0.5) {$\sC_3$};
	\node at (-1, -1.5) {$\sC_4$};
	
\end{tikzpicture}
\caption{The loop $\sC$.}
\label{fig:sC}
\end{center}
\end{figure}

With this notation, the Cauchy residue theorem states that, for all $\nu \in \cN(\eta)$ and all $L \ge L^\ast$, it holds
\[
	\gamma_\nu^L = \mathds{1}(H_{\nu}^L \le \varepsilon_F) = \dfrac{1}{2 \ri \pi} \oint_{\sC} \dfrac{\rd \lambda}{\lambda - H_{\nu}^L} \rd \lambda.
\]
In particular, since $Q_\nu^L :=  \gamma_\nu^L - \gamma_0^L$, it holds that
\begin{equation} \label{eq:decompositionQnuL}
	Q_\nu^L =  \dfrac{1}{2 \ri \pi} \oint_\sC \left( \dfrac{1}{\lambda - H_{\nu}^L} - \dfrac{1}{\lambda - H_{0}^L} \right) \rd \lambda
		=  \dfrac{1}{2 \ri \pi} \oint_\sC \dfrac{1}{\lambda - H_{\nu}^L} V_\nu^L \dfrac{1}{\lambda - H_{0}^L} \rd \lambda
		= Q_{\nu,1}^L + \widetilde{Q_{\nu,2}^L},
\end{equation}
where we set
\begin{equation} \label{eq:def:Q1}
	Q_{\nu,1}^L := \dfrac{1}{2 \ri \pi} \oint_{\sC}  \dfrac{1}{\lambda - H_{0}^L} V_{\nu}^L \dfrac{1}{\lambda - H_{0}^L} \rd \lambda
\end{equation}
and
\[
	 \widetilde{Q_{\nu,2}^L} :=   \dfrac{1}{2 \ri \pi} \oint_{\sC} \left( \dfrac{1}{\lambda - H_{\nu}^L} V_{\nu}^L \dfrac{1}{\lambda - H_{0}^L} V_{\nu}^L \dfrac{1}{\lambda - H_{0}^L} \right) \rd \lambda.
\]

The decomposition~\eqref{eq:decompositionQnuL} is motivated by the following lemma, which is very similar to~\cite[Lemma 3]{Cances2010}. In the sequel,  we consider the Banach spaces $\cQ^L = \left\{ Q^L \in \cS(L^2_\per(\WS_L)), \ \left\| Q^L \right\|_{\cQ^L} < \infty \right\} $ endowed with the norm (we denote by $\fS_p^L := \fS_p(L^2_\per(\WS_L))$ for clarity)
\[
	 \| Q^L \|_{\cQ^L} :=  \left\| (1 - \Delta^L)^{1/2} Q^L \right\|_{\fS_2^L}
		 +  \sum_{\alpha \in \{ +,-\}} \left\| (1 - \Delta^L)^{1/2} Q^{\alpha \alpha, L} (1 - \Delta^L)^{1/2} \right\|_{\fS_1^L}.
\]
Here, we denoted by $Q_\nu^{++,L} := (1 - \gamma_0^L) Q_\nu^L (1 - \gamma_0^L)$ and $Q_\nu^{--,L} := \gamma_0^L Q_\nu^L \gamma_0^L$.
\begin{lemma} \label{lem:controlQ}
There exists $C \in \R^+$ such that
\begin{equation} \label{eq:controlV}
	\forall \nu \in \cN(\eta), \ \forall L > L^*, \quad \left\| V_{\nu}^L \right\|_{\cC_L'} \le 2 \| \nu \|_{\cC_L} \le C \| \nu \|_{L^2(\R^3)}.
\end{equation}
Moreover, $Q_{\nu,1}^L$ and $\widetilde{Q_{\nu, 2}^L}$ are in $\cQ^L$, and there exists $C \in \R^+$ such that 
\begin{equation} \label{eq:controlQ1Q2}
	\forall \nu \in \cN(\eta), \ \forall L \ge L^*, \quad
	\left\| Q_{\nu,1}^L \right\|_{\cQ^L} \le C  \| \nu \|_{L^2(\R^3)}
	\quad \text{and} \quad
	\left\| \widetilde{Q_{\nu, 2}^L} \right\|_{\cQ^L} \le C  \| \nu \|_{L^2(\R^3)}^2.
\end{equation}
Finally, it holds
\[
	\forall \nu \in \cN(\eta), \quad \forall L \ge L^*, \quad Q_{\nu,1}^L \in \fS_1^L \quad \text{and} \quad \Tr_{L^2_\per(\WS_L)} \left( Q_{\nu,1}^L \right) = 0.
\]
\end{lemma}
The proof of Lemma~\ref{lem:controlQ} is postponed until Section~\ref{sec:proof:controlQ}. As a consequence, we see that $\widetilde{Q_{\nu,2}^L}$ contains only high order contributions in $\nu$. 

\medskip

The next lemma is a transposition of~\cite[Proposition 1]{Cances2008} in the supercell case. The proof follows the one in~\cite{Cances2008} upon replacing the Kato-Seiler-Simon inequality by the periodic Kato-Seiler-Simon inequality (see Corollary~\ref{cor:periodicKSS} below).

\begin{lemma} \label{lem:controlRhoQ}
There exists $C \in \R^+$ such that, for all $L \ge L^*$, all $Q^L \in \cQ^L$ and all $V^L \in \cC_L'$, it holds
\[
	\left| \left\bra \rho_{Q^L}, V^L \right\ket_{\cC_L, \cC_L'} \right| := 
	\left| \Tr_{L^2_\per(\WS_L)} \left( Q^L V^L   \right) \right| \le C \left\| Q^L \right\|_{\cQ^L} \left\| V^L \right\|_{\cC_L'}.
\]
\end{lemma}
In other words, the map $Q^L \mapsto \rho_{Q^L}$ is continuous from $\cQ^L$ to $\cC_L$ with a continuity bound that is independent of $L$.


\subsection{Linear response operators}

In order to study the operators $Q_{\nu,1}^L$ defined in~\eqref{eq:def:Q1}, we introduce, for $L \ge L^*$, the supercell irreducible polarizability operator $\chi^L : \cC_L' \to \cC_L$ defined by
\begin{equation} \label{eq:def:chiL}
	\chi^L : V^L \in \cC_L \mapsto \rho \left[ \dfrac{1}{2 \ri \pi} \oint_\sC \dfrac{1}{\lambda - H_0^L} V^L \dfrac{1}{\lambda - H_0^L} \rd \lambda \right] ,
\end{equation}
where we denoted by $\rho[Q^L]$ the density of $Q^L$ (in the sense of Lemma~\ref{lem:controlRhoQ}). Following the proof of Lemma~\ref{lem:controlQ}, we obtain the following lemma. In the sequel, we denote by $\cB(E,F)$ the Banach space of bounded operators from the Banach space $E$ to the Banach space $F$, and by $\cB(E) := \cB(E,E)$.
\begin{lemma} \label{lem:propertiesChiL}
	For all $L \ge L^*$, the operator $\chi^L$ is bounded from $\cC_L'$ to $\cC_L$, and there exists $C \in \R^+$ such that, for all $L \ge L^*$, it holds $\left\| \chi^L \right\|_{\cB(\cC_L', \cC_L)} \le C$.
\end{lemma}

From the definitions~\eqref{eq:sc_nu_L} and~\eqref{eq:def:Q1} and the decomposition~\eqref{eq:decompositionQnuL}, we get
\[
	\rho_{Q_{\nu,1}^L} = \chi^L \left( V_\nu^L \right) = \chi^L v_c^L \left( \rho_{Q_\nu^L} - \nu_L \right) = \chi^L v_c^L \left( \rho_{Q_{\nu,1}^L} - \nu_L \right) + \chi^L v_c^L \left( \rho_{\widetilde{Q_{\nu,2}^L}} \right).
\]
Applying the operator $\sqrt{v_c^L}$ leads after some straightforward manipulations to
\begin{equation} \label{eq:forQ1}
	 \left( 1 - \sqrt{v_c^L} \chi^L \sqrt{v_c^L}  \right) \sqrt{v_c^L} \left( \rho_{Q_{\nu,1}^L} - \nu_L \right) = - \sqrt{v_c^L} \nu_L + \sqrt{v_c^L} \chi^L \sqrt{v_c^L}  \left( \sqrt{v_c^L} \rho_{\widetilde{Q_{\nu,2}^L}} \right).
\end{equation}
In the sequel, we denote by $\cL^L$ the operator
\begin{equation} \label{eq:def:LL}
	\cL^L := - \sqrt{v_c^L} \chi^L \sqrt{v_c^L}.
\end{equation}
We have the following lemma, which is a variant of \textit{e.g.}~\cite[Proposition 2]{Cances2010}. We refer to this article for the proof.
\begin{lemma} \label{lem:propertiesLL}
	For all $L \ge L^*$, the operator $\cL^L$ is a non-negative bounded self-adjoint operator on $L^2_{0,\per}(\WS_L)$, and there exists $C \in \R^+$ such that, for all $L \ge L^*$, it holds $\left\| \cL^L \right\|_{\cB(L^2_{0,\per}(\WS_L))} \le C$. As a consequence, for all $L \ge L^*$, the operator $1 + \cL^L$ is invertible, and $(1 + \cL^L)^{-1}$ is a bounded self-adjoint operator on $L^2_{0,\per}(\WS_L)$ satisfying $\left\| (1 + \cL^L)^{-1} \right\|_{\cB(L^2_{0,\per}(\WS_L))} \le~1$.
\end{lemma}

From Lemma~\ref{lem:propertiesLL}, we deduce that~\eqref{eq:forQ1} can be rewritten as
\[
	 \sqrt{v_c^L} \left( \rho_{Q_{\nu,1}^L} - \nu_L \right) = - (1 + \cL^L)^{-1} \sqrt{v_c^L} \nu_L - (1 + \cL^L)^{-1} \cL^L  \left( \sqrt{v_c^L} \rho_{\widetilde{Q_{\nu,2}^L}} \right).
\]
The main result of this section is the following lemma, which shows that this decomposition indeed separates the linear contribution of $\nu$ in $ \sqrt{v_c^L} \left( \rho_{Q_{\nu,1}^L} - \nu_L \right)$. The proof is a straightforward consequence of Lemma~\ref{lem:controlQ}, Lemma~\ref{lem:propertiesChiL} and Lemma~\ref{lem:propertiesLL}.

\begin{lemma} \label{lem:linearPartQ1}
	There exists $C \in \R^+$ such that, for all $L \ge L^*$ and all $\nu \in \cN(\eta)$, it holds
	\[
		\left\| \sqrt{v_c^L} \left( \rho_{Q_{\nu,1}^L} - \nu_L \right) + (1 + \cL^L)^{-1} \sqrt{v_c^L} \nu_L \right\|_{ L^2_{0,\per}(\WS_L) } \le C \left\| \nu \right\|^2_{L^2(\R^3)}.
	\]
\end{lemma}

We end this section by mentioning that all the results also hold in the whole space setting  (see~\cite{Cances2010}). The operator $\chi : \cC' \to \cC$ defined by
\begin{equation} \label{eq:def:chi}
	\chi: V \in \cC' \mapsto 
	\rho \left[ \dfrac{1}{2 \ri \pi} \oint_\sC \dfrac{1}{\lambda - H_0} {V} \dfrac{1}{\lambda - H_0} \rd \lambda \right]
\end{equation}
is a well-defined bounded operator from $\cC'$ to $\cC$, and the operator 
\begin{equation} \label{eq:def:L}
	\cL := - \sqrt{v_c} \chi \sqrt{v_c}
\end{equation}
is a well-defined non-negative bounded self-adjoint operator on $L^2(\R^3)$, so that $(1 + \cL)^{-1}$ is also a well-defined bounded self-adjoint operator on $L^2(\R^3)$.


\subsection{The linear and quadratic contribution of $\nu$ in $J_\nu^L$ and in $J_\nu$}

We identify in this section the linear and quadratic contributions of $\nu$ in the highly non-linear functionals $J_\nu^L$ and $J_\nu$. Again, we state the arguments only for $J_\nu^L$ to check the dependence of the constants with respect to $L$, but similar results hold true for $J_\nu$. From~\eqref{eq:JnuL_infQ}, we obtain
\begin{equation} \label{eq:Jnu_as_F(Qnu)}
	J_\nu^L := \Tr_{L^2_\per(\WS_L)} \left( \left[ H_0^L -\varepsilon_F \right]  Q_\nu^L \right) + \frac12 D_L \left( \rho_{Q_\nu^L} - \nu_L, \rho_{Q_\nu^L} - \nu_L \right) - \int_{\WS_L} V_0^L \nu_L.
\end{equation}
Since $Q_\nu^L = \gamma_\nu^L - \gamma_0^L$ is the difference of two projectors, we get
\begin{align}
	 \Tr_{L^2_\per(\WS_L)} \left( \left[ H_0^L -\varepsilon_F \right]  Q_\nu^L \right) &=  \Tr_{L^2_\per(\WS_L)} \left( \left| H_0^L -\varepsilon_F \right|  \left(Q_\nu^{++,L} - Q_\nu^{--,L} \right) \right) \nonumber \\
	 & = \Tr_{L^2_\per(\WS_L)} \left( \left| H_0^L -\varepsilon_F \right|  \left(Q_\nu^L\right)^2 \right) \nonumber \\
	 & =  \Tr_{L^2_\per(\WS_L)} \left( \left| H_0^L -\varepsilon_F \right|  \left(Q_{\nu,1}^L\right)^2 \right) \label{eq:Q1^2} \\
	 & \quad + \Tr_{L^2_\per(\WS_L)} \left( \left| H_0^L -\varepsilon_F \right|  \left[ \left(Q_{\nu}^L\right)^2 -  \left(Q_{\nu,1}^L\right)^2 \right] \right) \label{eq:Qnu^2-Q1^2}.
\end{align}

The term in~\eqref{eq:Q1^2} can be treated in a similar way than~\cite[Lemma 3.2]{Lewin2013}:
\begin{lemma} \label{lem:LewRou}
	For all $L \ge L^*$, it holds that
	\[
		 \Tr_{L^2_\per(\WS_L)} \left( \left| H_0^L -\varepsilon_F \right|  \left(Q_{\nu,1}^L\right)^2 \right)
		 = 
		 -\dfrac12 D_L \left( \rho_{Q_{\nu,1}^L} - \nu_L , \rho_{Q_{\nu,1}^L} \right).
	\]
\end{lemma}

On the other hand, the term in~\eqref{eq:Qnu^2-Q1^2} is a high-order term in $\nu$. More specifically, we have the following lemma, whose proof is postponed until Section~\ref{sec:proof:kinetic_ho}.
\begin{lemma} \label{lem:kinetic_ho}
	There exists $C \in \R^+$ such that
	\[
		\forall L \ge L^*, \quad \forall \nu \in \cN(\eta), \quad \left| \Tr_{L^2_\per(\WS_L)} \left( \left| H_0^L -\varepsilon_F \right|  \left[ \left(Q_{\nu}^L\right)^2 -  \left(Q_{\nu,1}^L\right)^2 \right] \right) \right| \le C \| \nu \|_{L^2(\R^3)}^3.
	\]
\end{lemma}

From the equalities~\eqref{eq:Jnu_as_F(Qnu)},~\eqref{eq:Q1^2}-\eqref{eq:Qnu^2-Q1^2} and Lemma~\ref{lem:LewRou}, we deduce, after some algebra, that
\begin{align}
	J_\nu^L & = - \int_{\WS_L} V_0^L \nu_L -\dfrac12 D_L \left( \rho_{Q_{\nu,1}^L} - \nu_L, \nu _L\right) +   \label{eq:linearQuadraticJnuL}\\
	& \quad + \dfrac12 D_L \left( \rho_{Q_{\nu}^L} + \rho_{Q_{\nu,1}^L} - 2 \nu_L, \widetilde{\rho_{Q_{\nu,2}^L}} \right)
	+ \Tr_{L^2_\per(\WS_L)} \left( \left| H_0^L -\varepsilon_F \right|  \left[ \left(Q_{\nu}^L\right)^2 -  \left(Q_{\nu,1}^L\right)^2 \right] \right). \label{eq:ho1}
\end{align}
We further decompose the second term of~\eqref{eq:linearQuadraticJnuL} in view of Lemma~\ref{lem:linearPartQ1}. Using~\eqref{eq:normC=normC'}, we write
\begin{align}
	D_L \left( \rho_{Q_{\nu,1}^L} - \nu_L, \nu _L\right) & = 
	\left\bra \sqrt{v_c^L} \left( \rho_{Q_{\nu,1}^L} - \nu_L \right), \sqrt{v_c^L} \nu_L  \right\ket_{L^2_{0,\per}(\WS_L)} \nonumber \\
		& = \left\bra \left( 1 + \cL^L \right)^{-1} \sqrt{v_c^L} \nu_L, \sqrt{v_c^L} \nu_L  \right\ket_{L^2_{0,\per}(\WS_L)} + \label{eq:linearDL} \\
		& \quad + \left\bra \left[ \sqrt{v_c^L} \left( \rho_{Q_{\nu,1}^L} - \nu_L \right) -  \left( 1 + \cL^L \right)^{-1} \sqrt{v_c^L} \nu_L \right], \sqrt{v_c^L} \nu_L  \right\ket_{L^2_{0,\per}(\WS_L)}. \label{eq:hoDL}
\end{align}

From Lemma~\ref{lem:linearPartQ1}, we can control the term in~\eqref{eq:hoDL} (see Section~\ref{sec:proof:controlHoDL} for the proof).
\begin{lemma} \label{lem:controlHoDL}
	There exists $C \in \R^+$ such that, for all $\nu \in \cN(\eta)$ and all $L \ge L^*$, it holds
	\[
		\left| \left\bra \left[ \sqrt{v_c^L} \left( \rho_{Q_{\nu,1}^L} - \nu_L \right) -  \left( 1 + \cL^L \right)^{-1} \sqrt{v_c^L} \nu_L \right], \sqrt{v_c^L} \nu_L  \right\ket_{L^2_{0,\per}(\WS_L)} \right|
		\le C \| \nu \|_{L^2(\R^3)}^3.
	\]
\end{lemma}

Gathering~\eqref{eq:linearQuadraticJnuL}-\eqref{eq:ho1} with~\eqref{eq:linearDL}-\eqref{eq:hoDL}, and using Lemma~\ref{lem:kinetic_ho} and Lemma~\ref{lem:controlHoDL}, we obtain the following lemma, which identifies the linear and quadratic contributions of $\nu$ in $J_\nu^L$ and $J_\nu$. We enunciate the result for both the supercell case and the full model case.
\begin{lemma}[Linear and quadratic contributions of $\nu$] \label{lem:linearAndQuadratic}
	For all $\nu \in \cB(\eta)$ and all $L \ge L^*$, it holds
	\[
		J_\nu^L  = - \int_{\WS_L} V_0^L \nu_L - \dfrac12  \left\bra \left( 1 + \cL^L \right)^{-1} \sqrt{v_c^L} \nu_L, \sqrt{v_c^L} \nu_L  \right\ket_{L^2_{0,\per}(\WS_L)} + R_\nu^L
	\]
	and
	\[
		J_\nu  = - \int_{\R^3} V_0 \nu -\dfrac12   \left\bra \left( 1 + \cL \right)^{-1} \sqrt{v_c} \nu, \sqrt{v_c} \nu  \right\ket_{L^2(\R^3)} + R_\nu,
	\]
	where we set
	\begin{equation}
	\begin{aligned}
		R_\nu^L := & \dfrac12 D_L \left( \rho_{Q_{\nu}^L} - 2 \nu_L, \widetilde{\rho_{Q_{\nu,2}^L}} \right)
	+ \Tr_{L^2_\per(\WS_L)} \left( \left| H_0^L -\varepsilon_F \right|  \left[ \left(Q_{\nu}^L\right)^2 -  \left(Q_{\nu,1}^L\right)^2 \right] \right) \\
	& \quad	- \dfrac12 \left\bra \left[ \sqrt{v_c^L} \left( \rho_{Q_{\nu,1}^L} - \nu_L \right) -  \left( 1 + \cL^L \right)^{-1} \sqrt{v_c^L} \nu_L \right], \sqrt{v_c^L} \nu_L  \right\ket_{L^2_{0,\per}(\WS_L)},
	\end{aligned}
	\end{equation}
	and
	\begin{equation}
	\begin{aligned}
		R_\nu := & \dfrac12 D \left( \rho_{Q_{\nu}} - 2 \nu, \widetilde{\rho_{Q_{\nu,2}}} \right)
	+ \Tr_{L^2_\per(\WS)} \left( \left| H_0 -\varepsilon_F \right|  \left[ \left(Q_{\nu}\right)^2 -  \left(Q_{\nu,1}\right)^2 \right] \right) \\
	& \quad	- \dfrac12 \left\bra \left[ \sqrt{v_c} \left( \rho_{Q_{\nu,1}} - \nu_L \right) -  \left( 1 + \cL \right)^{-1} \sqrt{v_c} \nu_L \right], \sqrt{v_c} \nu  \right\ket_{L^2(\R^3)}.
	\end{aligned}
	\end{equation}
	Moreover, there exists $C \in \R^+$ such that 
	\begin{equation} \label{eq:cubicTerms}
		\forall \nu \in \cN(\eta), \quad \left\| R_\nu \right\|  \le C \| \nu \|_{L^2(\R^3)}^3 
		\quad \text{and} \quad
		\forall L \ge L^*, \quad \left\| R_\nu^L \right\|  \le C \| \nu \|_{L^2(\R^3)}^3.
	\end{equation}
\end{lemma}

In view of this decomposition, we write that $J_\nu - J_\nu^L = \cJ_{1,\nu}^L +   \cJ_{2,\nu}^L +  \cJ_{3,\nu}^L$, where we set
\begin{align}
	 \cJ_{1,\nu}^L & :=  \int_{\WS_L} V_0^L \nu_L - \int_{\R^3} V_0 \nu \label{eq:cJ1} \\
	 \cJ_{2,\nu}^L & := \dfrac12  \left\bra \left( 1 + \cL^L \right)^{-1} \sqrt{v_c^L} \nu_L, \sqrt{v_c^L} \nu_L  \right\ket_{L^2_{0,\per}(\WS_L)}  -\dfrac12   \left\bra \left( 1 + \cL \right)^{-1} \sqrt{v_c} \nu, \sqrt{v_c} \nu  \right\ket_{L^2(\R^3)} \label{eq:cJ2}\\
	  \cJ_{3,\nu}^L & := R_\nu - R_\nu^L. \label{eq:cJ3}
\end{align}

From~\eqref{eq:cubicTerms}, we obtain that $| \cJ_{3, \nu}^L | \le C \| \nu \|_{L^2(\R^3)}^3$, which then leads to the $\| \nu \|^3_{L^2(\R^3)}$ term in~\eqref{eq:K}. The linear part $\cJ_1^L$ defined in~\eqref{eq:cJ1} is easily controlled thanks to the exponential convergence of the mean-field potentials in the defect-free case~\cite{GL2015}. More specifically, we have the following Lemma, whose proof is postponed until Section~\ref{sec:proof:linear}
\begin{lemma}[Convergence of the linear part]
\label{lem:linear}
There exist $C \in \R^+$ and $\alpha > 0$ such that,
	\[
		\forall \nu \in \cN(\eta), \quad \forall L \ge L^*, 
		\quad
		\left|  \cJ_{1,\nu}^L \right| = \left| \int_{\R^3} \left( V_{0}^L - V_{0} \right) \nu \right| \le C \| \nu \|_{L^2(\R^3)} \re^{-\alpha L}.
	\]
\end{lemma}

The study of the quadratic term $\cJ_2^L$ defined in~\eqref{eq:cJ2} is more involving, and require a precise study of the operators $(1 + \cL^L)^{-1}$ and $(1 + \cL)^{-1}$. This is the topic of the next section.


\subsection{An intermediate operator}

In order to study the $\cJ_{2,\nu}^L$ term, we introduce an intermediate operator. The idea is to notice that in~\eqref{eq:cJ2}, there are two sources of errors. One comes from the fact that the operators $\cL$ and $\cL^L$ are constructed from different models, and the other comes from the fact that the scalar products depend on $L$. 

\medskip

Recall that $H_0 = -\frac12 \Delta + V_0$ defined in~\eqref{eq:sc} acts on $L^2(\R^3)$, and that $V_0$ is $\Lat$-periodic. For $L \in \N^*$, we introduce the operator
\begin{equation} \label{eq:def:widetildeH0}
	\widetilde{H_0^L} := -\frac12 \Delta^L + V_0 \quad \text{acting on} \quad L^2_\per(\WS_L).
\end{equation}
This operator has formally the same form than $H_0$, but acts on the periodic space $L^2_\per(\WS_L)$ instead of the whole space $L^2(\R^3)$. Since ${H_0}$ has a spectral gap of size at least $g$ around $\varepsilon_F$, we deduce (see~\cite[Proposition 3.1]{GL2015}) the following lemma.
\begin{lemma} \label{lem:widetildeH0}
	For all $L \ge L^*$, the operator $\widetilde{H_0^L}$ has a spectral gap of size at least $g$ around~$\varepsilon_F$.
\end{lemma}

We introduce the modified irreducible polarizability operator $\widetilde{\chi^L} : \cC_L' \to \cC_L$, defined by
\begin{equation} \label{eq:def:widetildeChiL}
	\widetilde{\chi^L} : V^L \in \cC_L \mapsto \rho \left[ \dfrac{1}{2 \ri \pi} \oint_\sC \dfrac{1}{\lambda - \widetilde{H_0^L}} V^L \dfrac{1}{\lambda - \widetilde{H_0^L}} \rd \lambda \right] .
\end{equation}

From Lemma~\ref{lem:widetildeH0}, we deduce that $\widetilde{\chi^L}$ is well-defined, and has properties similar to ${\chi^L}$ defined in~\eqref{eq:def:chiL}. We finally define the operator 
\begin{equation} \label{eq:def:widetildeLL}
	\widetilde{\cL^L} := -  \sqrt{v_c^L} \widetilde{\chi^L} \sqrt{v_c^L}.
\end{equation}
This operator shares the properties of $\cL^L$ defined in~\eqref{eq:def:LL}. We then write $\cJ_{2,\nu}^L = \cJ_{2,1,\nu}^L + \cJ_{2,2,\nu}^L$ with
\begin{align}
	 \cJ_{2,1,\nu}^L = & \frac12 \left\bra \left[ \left( 1 + \cL^L \right)^{-1}  - \left( 1 + \widetilde{\cL^L} \right)^{-1} \right] \sqrt{v_c^L} \nu_L, \sqrt{v_c^L} \nu_L  \right\ket_{L^2_{0,\per}(\WS_L)} +  \label{eq:diffQuad1}\\
	\cJ_{2,2,\nu}^L = & \frac12 \left( \left\bra \left( 1 + \widetilde{\cL^L} \right)^{-1} \sqrt{v_c^L} \nu_L, \sqrt{v_c^L} \nu_L  \right\ket_{L^2_{0,\per}(\WS_L)} - \left\bra \left( 1 + \cL \right)^{-1} \sqrt{v_c} \nu, \sqrt{v_c} \nu  \right\ket_{L^2(\R^3)} \right). \label{eq:diffQuad2}
\end{align}

The first term is controlled thanks to the following lemma, whose proof is postponed until Section~\ref{sec:proof:intermediate}.

\begin{lemma} \label{lem:intermediate}
	 There exist $C \in \R^+$ and $\alpha > 0$ such that, for all $L \ge L^*$, it holds that
	 \[
	 	\left|  \cJ_{2,1,\nu}^L \right| = \left| \left\bra \left[ \left( 1 + \cL^L \right)^{-1}  - \left( 1 + \widetilde{\cL^L} \right)^{-1} \right] \sqrt{v_c^L} \nu_L, \sqrt{v_c^L} \nu_L  \right\ket_{L^2_{0,\per}(\WS_L)} \right| \le C \| \nu \|_{L^2(\R^3)}^2 \re^{-\alpha L}.
	 \]
\end{lemma}

It remains to study the convergence of $\cJ_{2,2,\nu}^L$ defined in~\eqref{eq:diffQuad2} towards $0$. This is somehow an easier problem than before, for $\widetilde{\cL^L}$ and $\cL$ have very similar expressions. To study this last convergence, we use the Bloch transforms.


\section{Regularity of Bloch transforms}
\label{sec:BlochRegularity}


\subsection{Bloch transform from $L^2(\R^3)$ to $L^2(\BZ, L^2_\per(\WS))$}
\label{ssec:Bloch}
We recall the definition and basic properties of the Bloch transforms for the sake of completeness. We refer to~\cite[Chapter XIII]{ReedSimon4}) and \cite{GL2015}) for more details. In the sequel, we denote by $L^2_\per := L^2_\per(\WS)$ for clarity. We consider the Hilbert space $L^2(\BZ, L^2_\per)$, endowed with the normalized inner product
 \[
 	\bra f(\bq, \bx), g(\bq, \bx) \ket_{L^2(\BZ, L^2_\per(\WS))} := \fint_{\BZ} \int_{\WS} \overline{f}(\bq, \bx) g(\bq, \bx) \, \rd \bx \, \rd \bq,
 \]
 where we denoted by $\fint_{\BZ} = | \BZ |^{-1} \int_{\BZ}$. The Bloch transform $\cZ$ is defined by
\begin{equation} \label{eq:def:cZ}
	 \begin{array}{lcll}
	\cZ: & L^2(\R^3) & \to & L^2(\BZ, L^2_\per) \\
		 & w & \mapsto & (\cZ w)(\bq,\bx) := w_\bq(\bx) := \displaystyle \sum\limits_{\bR \in \Lat} \re^{-\ri \bq \cdot (\bx + \bR)} w(\bx + \bR).
	\end{array}
\end{equation}
It is an isometry from $L^2(\R^3)$ to $ L^2(\BZ, L^2_\per)$, whose inverse is given by
\[
	\cZ^{-1}:  w_\bq(\bx)  \mapsto  (\cZ^{-1} w)(\bx) := \displaystyle \fint_{\BZ} \re^{ \ri \bq \cdot \bx} w_\bq( \bx) \, \rd \bq.
\]
For $\bm \in \RLat$, we introduce the unitary operator $U_\bm$ acting on $L^2_\per$ defined by 
\begin{equation} \label{eq:def:Um}
	\forall \bm \in \RLat, \quad \forall f \in L^2_\per(\WS), \quad
	\left( U_\bm f \right)(\bx) = \re^{-\ri \bm \cdot \bx} f(\bx).
\end{equation}
From~\eqref{eq:def:cZ}, we can consider $\cZ w$ as a function of $L^2_\loc \left( \R^3, L^2_\per \right)$ with
\begin{equation} \label{eq:function_covariant}
	 \forall w \in L^2(\R^3), \ \forall \bm \in \RLat, \ \forall \bq \in \BZ, \quad 
	 \left( \cZ w \right)(\bq + \bm, \cdot) = w_{\bq + \bm} = U_\bm w_{\bq} = U_\bm \left( \cZ w (\bq, \cdot) \right).
\end{equation}

Let $A$ with domain $\cD(A)$ be a possibly unbounded operator acting on $L^2_\per$. We say that $A$ commutes with $\Lat$-translations if $\tau_\bR A = A \tau_\bR$ for all $\bR \in \Lat$. If $A$ commutes with $\Lat$-translations, then it admits a Bloch decomposition: there exists a family of operators $\left(A_\bq \right)_{\bq \in \BZ}$ acting on $L^2_\per$, such that, if $f \in L^2(\R^3)$ and $g \in \cD(A)$ are such that $f = Ag$, then, for almost any $\bq \in \BZ$, $g_\bq \in L^2_\per$ is in the domain of $A_\bq$, and
\begin{equation}\label{eq:op_diagonal}
 f_\bq = A_\bq g_\bq.
\end{equation} 
In this case, we write
\begin{equation*}
	\cZ A \cZ^{-1} = \fint_{\BZ}^{\oplus} A_\bq \rd \bq \quad \text{(Bloch decomposition of $A$).}
\end{equation*}
From~\eqref{eq:function_covariant}, we can extend the definition of $A_\bq$, initially defined for $\bq \in \BZ$, to $\bq \in \R^3$, by setting 
\begin{equation} \label{eq:rotation}
	\forall \bm \in \RLat, \quad \forall \bq \in \BZ, \quad A_{\bq + \bm} = U_\bm A_\bq U_\bm^{-1},
\end{equation}
so that~\eqref{eq:op_diagonal} holds for almost any $\bq\in \RR^3$. 

If $A$ is locally trace-class, then $A_\bq$ is trace-class on $L^2_\per$ for almost any $\bq \in \R^3$. The operator $A$ can be associated a density $\rho_A$, which is an $\Lat$-periodic function, given by 
\[
	\rho_A = \fint_{\BZ} \rho_{A_\bq} \rd \bq,
\] 
where $\rho_{A_\bq}$ is the density of the trace-class operator $A_\bq$. The trace per unit volume of $A$ (defined in~\eqref{eq:VTr_def1}) is also equal to
\begin{equation} \label{eq:VTr}
	\VTr(A) = \fint_{\BZ} \Tr_{L^2_\per} \left( A_\bq \right) \rd \bq .
\end{equation}

\subsection{The supercell Bloch transform}
\label{sec:supercellBloch}

We present in this section the ``supercell'' Bloch transform, already introduced in~\cite{GL2015}. This transformation goes from~$L^2_\per(\Gamma_L)$ to $\ell^2(\Lambda_L, L^2_\per)$, where $\Lambda_L := \left( L^{-1} \RLat\right) \cap \BZ$, \textit{i.e.} 
\begin{equation} \label{eq:def:LambdaL}
	\Lambda_L := \left\{ \dfrac{k_1}{L} \ba_1^* +  \dfrac{k_2}{L} \ba_2^* +  \dfrac{k_3}{L} \ba_3^*, \ (k_1, k_2, k_3) \in \left\{ \dfrac{-L+ \eta}{2},\dfrac{-L+\eta}{2} + 1, \cdots,  \dfrac{L + \eta}{2} - 1 \right\}^3 \right\},
\end{equation}
with $\eta = 1$ if $L$ is odd, and $\eta = 0$ if $L$ is even, so that there are exactly $L^3$ points in $\Lambda_L$. Similarly, we define $\Lat_L := \Lat \cap \WS_L$, which contains $L^3$ points of the lattice $\Lat$. We introduce the Hilbert space $\ell^2(\Lambda_L, L^2_\per)$ endowed with the normalized inner product
\[
	\bra f(\bQ, \bx), g(\bQ, \bx) \ket_{\ell^2(\Lambda_L, L^2_\per)} := \dfrac{1}{L^3}\sum_{\bQ \in \Lambda_L} \int_{\WS} \overline{f}(\bQ, \bx) g(\bQ, \bx) \, \rd \bx.
\]
The supercell Bloch transform is defined by
\begin{equation} \label{eq:def:cZL}
	\begin{array}{llll}
	\cZ_L: & L^2_\per(\WS_L) & \to &  \ell^2(\Lambda_L, L^2_\per) \\
		& w & \mapsto & (\cZ_L w)(\bQ,\bx) := w_\bQ(\bx) := \displaystyle \sum\limits_{\bR \in \Lat_L} \re^{-\ri \bQ \cdot (\bx + \bR)} w(\bx + \bR).
	\end{array}
\end{equation}
It is an isometry from $L^2_\per(\WS_L)$ to $\ell^2(\Lambda_L, L^2_\per)$ whose inverse is given by
\[
	\cZ_L^{-1}: w_\bQ( \bx) \mapsto (\cZ_L^{-1} w)(\bx) :=  \dfrac{1}{L^3} \sum\limits_{\bQ \in \Lambda_L} \re^{ \ri \bQ \cdot \bx} w_\bQ( \bx).
\]
We can extend $\cZ$ to $\ell^\infty \left( L^{-1} \RLat , L^2_\per \right)$ with
\[
	\forall w \in L^2_\per(\WS_L), \quad \forall \bm \in \RLat, \quad \forall \bQ \in \Lambda_L, \quad w_{\bQ + \bm} = U_\bm w_{\bQ},
\]
where the operator $U_\bm$ was defined in~\eqref{eq:def:Um}.

Let $A^L$ with domain $\cD \left(A^L \right)$ be an operator acting on $L^2_\per(\WS_L)$. If $A$ commutes with $\Lat$-translations, then it admits a supercell Bloch decomposition: there exists a family of operators $(A_\bQ^L)_{\bQ \in \Lambda_L}$ acting on $L^2_\per$ such that if $f = A^L g$ with $f \in L^2_\per(\Gamma_L)$ and $g \in \cD(A^L)$, then for all $\bQ \in \Lambda_L$, $g_\bQ \in \cD(A^L)$ and
\begin{equation} \label{eq:op_diagonal_sc}
 f_\bQ = A_\bQ^L g_\bQ.
\end{equation} 
We write
\[
	\cZ_L A^L \cZ_L^{-1} := \dfrac{1}{L^3} \bigoplus_{\bQ \in \Lambda_L} A^L_\bQ \quad \text{(supercell Bloch decomposition of $A^L$).}
\]
Similarly to~\eqref{eq:rotation}, we extend the definition of $A^L_\bQ$ to $L^{-1} \RLat$ with
\[
	\forall \bm \in \RLat, \quad \forall \bQ \in \Lambda_L, \quad A^L_{\bQ + \bm} = U_\bm A^L_\bQ U_{\bm}^{-1},
\]
so that~\eqref{eq:op_diagonal_sc} holds for all $\bQ \in L^{-1} \RLat$.

Finally, if the operator $A^L$ is trace-class, we define the trace per unit volume by
\begin{equation} \label{eq:VTrL}
	\VTr_L(A^L) = \dfrac{1}{L^3} \Tr_{L^2_\per(\WS_L)}(A^L) = \dfrac{1}{L^3} \sum\limits_{\bQ \in \Lambda_L} \Tr_{L^2_\per}(A^L_\bQ),
\end{equation}
and the associated density is given by 
$
	\displaystyle \rho_{A^L} = \dfrac{1}{L^3} \sum\limits_{\bQ \in \Lambda_L} \rho_{A^L_{\bQ}}
$,
where $\rho_{A^L_\bQ}$ is the density of the trace-class operator $A^L_\bQ$.


\subsection{Bloch transforms of $H_0$ and $\widetilde{H_0^L}$.}
We now derive the Bloch transformations of the different operators that we encountered.
We begin by noticing that the Bloch transforms of $-\Delta$ and $-\Delta^L$ are respectively given by
\[
	\cZ \left( - \Delta \right) \cZ^{-1} = \fint_\BZ^\oplus \left( -\Delta \right)_\bq \rd \bq
	\quad \text{and} \quad
	\cZ_L \left( - \Delta^L \right)   \cZ_L^{-1} = \dfrac{1}{L^3} \bigoplus_{\bQ \in \Lambda_L} \left(-\Delta\right)_\bQ,
\]
where we set
\[
	\forall \bq \in \BZ, \quad \left(-\Delta\right)_\bq = \left| - \ri \nabla^1 + \bq \right|^2 = \sum_{j=1}^3 \left( P_j^1 + q_j \right)^2.
\]
Here, we denoted by $\nabla^1$ the gradient operator on $L^2_\per$, and we recall the operators $P_j^L$, $j \in \{1,2,3\}$ have been defined in Section~\ref{ssec:rHF}. In particular, since the potential $V_0$ defined in~\eqref{eq:sc} is $\Lat$-periodic, the Bloch transform of $H_0$ is
\begin{equation} \label{eq:def:Hq}
	\cZ H_0 \cZ^{-1} = \fint_\BZ^\oplus H_\bq \rd \bq
	\quad \text{with} \quad
	H_\bq := \frac12 \left| - \ri \nabla^1 + \bq \right|^2 + V_0 = - \frac12 \Delta^1 - \ri \bq \cdot \nabla^1 + \frac{|\bq|^2}{2} + V_0,
\end{equation}
and the supercell Bloch transform of the operator $\widetilde{H_0^L}$ defined in~\eqref{eq:def:widetildeH0} is simply
\begin{equation} \label{eq:BlochHq}
	\cZ_L \widetilde{H_0^L} \cZ_L^{-1} = \dfrac{1}{L^3}\bigoplus_{\bQ \in \Lambda_L} H_\bQ.
\end{equation}
In other words, it holds that $\left( \widetilde{H_0^L}\right)_\bQ = \left({H_0}\right)_\bQ$. In view of~\eqref{eq:def:Hq}, we can extend the definition of $H_\bq$ to the whole complex plane. More specifically, for $\bz \in \C^3$, we define
\begin{equation} \label{eq:def:Hz}
	H_\bz := - \frac12 \Delta^1 - \ri \bz \cdot \nabla^1 + \frac{\bz^2}{2} + V_0
	\quad \text{acting on}
	\quad L^2_\per,
\end{equation}
where $\bz^2$ is a short notation for $\bz^T \bz = z_1^2 + z_2^2 + z_3^2$. The map $\bz \mapsto H_\bz$ is an holomorphic family of type $(A)$ (see~\cite[Chapter VII]{Kato2012}). 
For $\bz \in \C^3$, and $\lambda \in \sC$, we introduce
\begin{equation} \label{eq:def:B1q_B2q}
	\begin{aligned}
		{B_1}(\lambda, \bz) & := (1 - \Delta^1) \dfrac{1}{\lambda - H_\bz}, 
		\quad \text{and} \quad
		{B_2}(\lambda, \bz) & := \dfrac{1}{\lambda - H_\bz} (1 - \Delta^1).
	\end{aligned}
\end{equation}
The following lemma was proved in~\cite[Lemma 5.2]{GL2015}.
\begin{lemma} \label{lem:bounds_Bq}
	For all $\bq \in \R^3$, and all $\lambda \in \sC$, the operator $\lambda - H_\bq$ is invertible. For any compact $K \subset \R^3$, there exists $C_K \in \R^+$ such that,
	\begin{equation*}
		\forall \bq \in K, \quad \forall \lambda \in \sC, \quad 
	\left\|  {B_{1,2}}(\lambda, \bq) \right\|_{\cB(L^2_\per)} \le {C_K}.
	\end{equation*}
	Moreover, there exists $A > 0$ such that, for all $\bz \in \R^3 + \ri [-A,A]^3$ and all $\lambda \in \sC$, the operator $\lambda - H_\bz$ is invertible, and there exists $C_K \in \R^+$ such that
	\begin{equation} \label{eq:bounds_B_z}
		\forall \bz \in K + \ri \, [-A, A]^3, \quad \forall \lambda \in \sC, \quad 
		\left\|  {B_{1,2}}(\lambda, \bz) \right\|_{\cB(L^2_\per)}  \le C_K.
	\end{equation}
\end{lemma}

\begin{remark}
In practice, we take the compact $K$ big enough so that, for instance $2 \BZ \subset K$. This is useful in order to consider for instance $B_{1,2}(\lambda, \bq - \bq')$ for $\bq, \bq' \in \BZ$.
\end{remark}

For all $\bq \in \BZ$, the operator $H_\bq$ acting on $L^2_\per$ is a bounded below self-adjoint operator which is compact resolvent. In particular, its spectrum is purely discrete, and accumulates at infinity. In the sequel, we denote by $\varepsilon_{1,\bq} \le \varepsilon_{2,\bq} \le \cdots$ the eigenvalues of $H_\bq$ sorted in increasing order, and by $\left( u_{n, \bq} \right)_{n \in \N^*}$ an associated orthonormal basis of eigenvectors. Since $V_0$ is real-valued, it holds
\begin{equation} \label{eq:symmetries}
	\forall \bq \in \BZ, \quad \forall n \in \N^*, \quad u_{n,-\bq} = \overline{u_{n,\bq}}
	\quad \text{and} \quad
	\varepsilon_{n,-\bq} = \varepsilon_{n,\bq}
\end{equation}
We emphasize that the maps $\bq \mapsto \varepsilon_{n, \bq}$ and $\bq \mapsto u_{n, \bq}$ are not smooth in general. In the sequel, we write
\begin{equation} \label{eq:spectralHq}
	H_\bq = \sum_{n=1}^\infty \varepsilon_{n, \bq} | u_{n, \bq} \ket \bra u_{n, \bq} |,
\end{equation}
with the implicit convention that the Dirac's notation for the bra-ket inner product refers to the one of $L^2_\per$. Finally, we denote by $N$ the common number of eigenvalues of $H_\bq$ below $\varepsilon_F$, so that
\begin{equation} \label{eq:gap}
	\forall \bq \in \BZ, \quad \varepsilon_{N, \bq} \le \varepsilon_F -g 
	\quad \text{and} \quad
	\varepsilon_F + g  \le \varepsilon_{N+1, \bq}.
\end{equation}

\subsection{Bloch transforms of $v_c$ and $v_c^L$.}
We now compute the Bloch transform of the operators $v_c$ and $v_c^L$, which commute with $\Lat$-translations. Let $\left( e_\bk \right)_{\bk \in \RLat}$ denotes the orthonormal Fourier basis of $L^2_\per$:
\begin{equation} \label{eq:def:FourierBasis}
	\forall \bk \in \RLat, \quad \forall \bx \in \WS, \quad e_\bk(\bx) := \dfrac{1}{| \WS |^{1/2}} \re^{\ri \bk \cdot \bx}.
\end{equation}
It is classical to prove that the Bloch transform of the operator $v_c$ is
\begin{equation} \label{eq:BlochVc}
	\cZ v_c \cZ^{-1} = \fint_{\BZ}^\oplus v_{c,\bq} \rd \bq 
	\quad \text{with} \quad 
	v_{c,\bq} = 4 \pi \sum_{\bk \in \RLat} \dfrac{| e_\bk \ket \bra e_\bk |}{| \bk + \bq |^2}.
\end{equation}
Likewise, the Bloch transform of $\sqrt{v_c}$ is
\begin{equation} \label{eq:BlochSqrtVc}
	\cZ \sqrt{v_c} \cZ^{-1} = \fint_{\BZ}^\oplus \left(\sqrt{v_c} \right)_\bq \rd \bq 
	\quad \text{with} \quad 
	\left(\sqrt{v_c} \right)_\bq = \sqrt{4 \pi} \sum_{\bk \in \RLat} \dfrac{| e_\bk \ket \bra e_\bk |}{| \bk + \bq |}.
\end{equation}
From Lemma~\ref{lem:vc}, we obtain that, for $\rho \in \cC$, it holds that
\[
	\| \rho \|_{\cC}^2 = \| \sqrt{v_c} \, \rho \|_{L^2(\R^3)}^2 = \fint_\BZ \| (\sqrt{v_c})_\bq \rho_\bq \|_{L^2_\per}^2 \rd \bq
	= \fint_\BZ \|  \rho_\bq \|_{\cC_\bq}^2 \rd \bq,
\]
where we identify the domain of $(\sqrt{v_c})_\bq$ to be, for $\bq \in \BZ \setminus \{ \bnull \}$,
\[
	\cC_\bq := \left\{ f \in \sS'_\per(\WS), \ \left\| f \right\|_{\cC_\bq}^2 := \left\| (\sqrt{v_c})_\bq f \right\|_{L^2_\per}^2 = 4 \pi \sum_{\bk \in \RLat} \dfrac{| c_\bk(f) |^2}{| \bq + \bk |^2} < \infty \right\}.
\]
Taking $L^2_\per$ as the pivoting space, the dual of $\cC_\bq$ is the space $\cC_\bq'$ defined, for $\bq \in \BZ \setminus \{ \bnull \}$, by
\[
	\cC'_\bq := \left\{ f \in \sS'_\per(\WS), \quad \left\| f \right\|_{\cC'_\bq}^2 := (4 \pi)^{-1} \sum_{\bk \in \RLat} {| c_\bk(f) |^2}{| \bq + \bk |^2} < \infty \right\}.
\]
It is easy to see that for $\bq \in \BZ \setminus \{ \bnull \}$, the operator $(\sqrt{v_c})_\bq$ is an isometry from $\cC_\bq$ to $L^2_\per$ and from $L^2_\per$  to $\cC'_\bq$, so that $({v_c})_\bq$ is an isometry from $\cC_\bq$ to $\cC'_\bq$.

\medskip

In the supercell setting, we obtain $\cZ_L \sqrt{v_c^L} \cZ_L^{-1} = \dfrac{1}{L^3} \bigoplus\limits_{\bQ \in \Lambda_L} \left( \sqrt{v_c^L} \right)_\bQ$ with
\begin{equation} \label{eq:vcLbQ}
	\left( \sqrt{v_c^L} \right)_\bQ= \left\{ \begin{array}{lll}
	\displaystyle		\sqrt{4 \pi} \sum\limits_{\bk \in \RLat} \dfrac{| e_\bk \ket \bra e_\bk |}{| \bk + \bQ |} & \text{if} & \bQ \neq \bnull \\
		& & \\
	\displaystyle	\sqrt{4 \pi} \sum\limits_{\bk \in \RLat \setminus \{ \bnull \}} \dfrac{| e_\bk \ket \bra e_\bk |}{| \bk  |} & \text{if} & \bQ = \bnull.
		\end{array}
		\right.
\end{equation}
In particular, if $\bQ \neq \bnull$, it holds that $( \sqrt{v_c^L} )_\bQ = ( \sqrt{v_c} )_\bQ $, while for $\bQ = \bnull$, $(\sqrt{v_c^L})_\bnull$ is only defined on
\[
	\cC_\bnull := \left\{ f \in \sS'_\per(\WS), \ c_\bnull(f) = 0, \ \left\| f \right\|_{\cC_\bnull}^2 := \left\| (\sqrt{v_c^L})_\bnull f \right\|_{L^2_\per}^2 = 4 \pi \sum_{\bk \in \RLat \setminus \{ \bnull \}} \dfrac{| c_\bk(f) |^2}{| \bk |^2} < \infty \right\}.
\]



\subsection{Bloch transforms of $\chi$ and $\widetilde{\chi^L}$.}

The operators $\chi$ and $\widetilde{\chi^L}$ defined respectively in~\eqref{eq:def:chi} and~\eqref{eq:def:widetildeChiL} commute with $\Lat$-translations, hence admit Bloch transforms. To calculate them, we first recall that the Bloch matrix of a multiplicative operator by a function $V$ is $\left[ V \right]_{\bq, \bq'} = V_{\bq - \bq'}$ and that the Bloch transform of a function $\rho_A$ which is the density of an operator $A$ acting on $L^2(\R^3)$ is formally (see \textit{e.g.}~\cite[Equation (48)]{Cances2010})
\begin{equation} \label{eq:densityBloch}
	\left( \rho_A \right)_\bq := \fint_{\BZ} \rho \left[ A_{\bq', \bq' - \bq} \right] \rd \bq'.
\end{equation}
We deduce that the Bloch transform of the operator $\chi$ defined in~\eqref{eq:def:chi} is $\cZ \chi \cZ^{-1} = \fint_\BZ^\oplus \chi_\bq \rd \bq$ with  (see also~\cite[Proposition 3]{Cances2010})
\begin{equation} \label{eq:chiQ}
	\chi_\bq : V \in \cC'_\bq \mapsto 
	 \fint_\BZ \rho \left[  \dfrac{1}{2 \ri \pi}  \oint_\sC \dfrac{1}{\lambda - H_{\bq'}} {V} \dfrac{1}{\lambda - H_{\bq' - \bq}} \rd \lambda \right]\rd \bq'.
\end{equation}
Since $\chi$ is bounded from $\cC'$ to $\cC$ (see Lemma~\ref{lem:propertiesChiL}), we deduce that for almost all $\bq \in \BZ$, the operator $\chi_\bq$ is bounded from $\cC'_\bq$ to $\cC_\bq$. 

In the supercell setting, the analog of~\eqref{eq:densityBloch} is
\[
	\forall \bQ \in \Lambda_L, \quad \left( \rho_A \right)_\bQ := \dfrac{1}{L^3} \sum_{\bQ' \in \Lambda_L} \rho \left[ A_{\bQ', \bQ' - \bQ} \right].
\]
Together with~\eqref{eq:BlochHq}, we obtain that $ \cZ_L \widetilde{\chi^L} \cZ_L^{-1} = \dfrac{1}{L^3}  \bigoplus\limits_{\bQ \in \Lambda_L} \widetilde{\chi^L_\bQ}$ with
\[
	\widetilde{\chi^L_\bQ} : V \in \cC'_\bQ \mapsto 
	 \dfrac{1}{L^3} \sum_{\bQ' \in \Lambda_L} \rho \left[  \dfrac{1}{2 \ri \pi}  \oint_\sC \dfrac{1}{\lambda - H_{\bQ'}} {V} \dfrac{1}{\lambda - H_{\bQ' - \bQ}} \rd \lambda \right].
\]

\subsection{Bloch transforms of the operators $\cL$ and $\widetilde{\cL^L}$.} 
From the definitions~\eqref{eq:def:L} and~\eqref{eq:def:widetildeLL}, we deduce that the operators $\cL$ and $\widetilde{\cL^L}$ commute with $\Lat$-translations, and that it holds
\[
	\cZ \cL \cZ^{-1} = \fint_{\BZ}^\oplus \cL_\bq \rd \bq
	\quad \text{with}\quad
	\forall \bq \in \BZ, \quad	\cL_\bq = - \left(\sqrt{v_c} \right)_\bq  \chi_\bq \left(\sqrt{v_c} \right)_\bq.
\]
and
\[
	\cZ_L \widetilde{\cL^L} \cZ_L^{-1} = \dfrac{1}{L^3} \bigoplus \cL_\bQ 
	\quad \text{with}\quad
	\widetilde{\cL_\bQ} = - \left( \sqrt{v_c^L} \right)_\bQ \widetilde{\chi_\bQ^L} \left( \sqrt{v_c^L} \right)_\bQ.
\]

The following lemma controls the difference between $\widetilde{\cL^L_\bQ}$ and ${\cL_\bQ}$ (see Section~\ref{sec:proof:LQ_LLQ} for the proof). It is based on the fact that the difference between Riemann sums and the corresponding integral decays exponentially fast for analytic integrands.
\begin{lemma} \label{lem:LQ_LLQ}
	There exist $C \in \R^+$ and $\alpha > 0$ such that, for all $L \ge L^*$ and all $\bQ \in \Lambda_L \setminus \{ \bnull \}$, it holds that
	\[
		\left\|  \widetilde{\cL^L_\bQ} - \cL_\bQ \right\|_{\cB(L^2_\per)} \le C \re^{- \ri \alpha L}.
	\]
\end{lemma}

The rest of this section is devoted to exhibiting some properties of the operators $\cL_\bq$. Recall that, for $f,g \in L^2_\per$,
\begin{align*}
	\bra f | \cL_\bq g \ket 
		 = \dfrac{-1}{2 \ri \pi} \fint_\BZ \oint_\sC \Tr_{L^2_\per(\WS)} \left( 
		\dfrac{1}{\lambda - H_{\bq'}} \left[ \left( \sqrt{v_c} \right)_\bq g \right] \dfrac{1}{\lambda - H_{\bq' - \bq}} \left[ \left( \sqrt{v_c} \right)_\bq \overline{f} \right] 
	\right) \rd \lambda \rd \bq'.
\end{align*}
For $\bz \in \C^3 \setminus \RLat$, we introduce the operator $(\sqrt{v_c})_\bz$ defined in the Fourier basis as
\[
	(\sqrt{v_c})_\bz := \sqrt{4 \pi} \sum_{\bk \in \RLat} \dfrac{| e_\bk \ket \bra e_\bk |}{| \bk + \bz |}.
\]
We also introduce the operator $\cL_\bz$ acting on $L^2_\per$ defined for $f,g \in L^2_\per$ by
\[
	\bra f | \cL_\bz g \ket 
		= \dfrac{-1}{2 \ri \pi} \fint_\BZ \oint_\sC \Tr_{L^2_\per(\WS)} \left( 
		\dfrac{1}{\lambda - H_{\bq'}} \left[ \left( \sqrt{v_c} \right)_\bz g \right] \dfrac{1}{\lambda - H_{\bq' - \bz}} \left[ \left( \sqrt{v_c} \right)_\bz \overline{f} \right] 
	\right) \rd \lambda \rd \bq' .
\]
We will study these operators in two different regimes. For $r>0$, and $\bq_0 \in \R^3$, we denote by $\cB(\bq_0, r) := \left\{ \bq \in \R^3, \ |\bq  - \bq_0| \le r \right\}$, and 
\[
	\Omega_r := \R^3 \setminus \left\{ \bigcup_{\bk \in \RLat} \cB(\bk, r) \right\}.
\]
In other words, any $\bq \in \Omega_r$ is far from $\RLat$. We recall if $E$ is a Banach space, a map $f: \Omega \subset \C^d \to E$ is said to be (strongly) analytic on the open subset $\Omega$ if for all $\bz \in \Omega$, $\partial_{z_j} f \in E$ for $1 \le j \le d$ (\textit{i.e.} $\left\| \partial_{z_j} f \right\|_E < \infty$). In particular, if $E$ and $F$ are Banach spaces, and $A : \Omega \subset \C^d \to \cB(E,F)$ and $f : \Omega\subset \C^d \to E$ are analytic maps on $\Omega$, then $(Af): \bz  \mapsto A(\bz) f(\bz) \in F$ is analytic from $\Omega$ to $F$.

\medskip

We first have the following lemma (see Section~\ref{sec:proof:analLq} for the proof).
\begin{lemma} \label{lem:analLq}
	There exists $A > 0$ such that, for all $r > 0$, the map $\bz \mapsto \cL_\bz$ is well defined and analytic from $\Omega_r +\ri [-A, A]^3$ to $\cB(L^2_\per)$. Moreover, there exists $C \in \R^+$ such that
\[
	\sup_{\bz \in \Omega_r + \ri [-A, A]^3} \left\| \cL_\bz \right\|_{\cB(L^2_\per)} \le C.
\]
\end{lemma}


We now study the operator $\cL_\bq$ as $\bq \to \bnull$. 
To do so, we block-decompose the operators. More specifically, we write $L^2_\per = \{ \C | e_\bnull \ket \} \oplus L^2_{0,\per}$, where $L^2_{0,\per}$ is a short notation for $\{ f \in L^2_\per, c_\bnull(f) = 0\}$. We introduce $P_c = | e_\bnull \ket \bra e_\bnull | $ the orthogonal projection  on the constants, and $P_0 = P_c^\perp$ the orthogonal projection from $L^2_\per$ to $L^2_{0,\per}$. With this decomposition, the operator $(\sqrt{v_c})_\bq$ has the matrix form
\[
	(\sqrt{v_c})_\bq = \left( \begin{array}{c | c } 
	\dfrac{\sqrt{4 \pi}  }{| \bq |} & 0 \\
	\hline \\
	0 & (\sqrt{w_c})_\bq
	\end{array} \right),
\]
where, for all $\bz \in \C^3$, we defined
\[
	(\sqrt{w_c})_\bz := \sqrt{4 \pi} \sum_{\bk \in L^{-1}\RLat \setminus \{ \bnull \}} \dfrac{| e_\bk \ket \bra e_\bk |}{| \bz + \bk |}.
\]
The next lemma is straightforward.
\begin{lemma} \label{lem:wc}
	There exist $C \in \R^+$, $r_1 > 0$ and $A > 0$ such that the map $\bz \mapsto (\sqrt{w_c})_\bz$ is analytic from $\cB(\bnull, r_1) + \ri [-A,A]^3$ to $\cB(L^2_{0,\per})$, with
	\[
		\sup_{\bz \in \cB(\bnull, r_1) + \ri [-A,A]^3} \left\| (\sqrt{w_c})_\bz \right\|_{\cB(L^2_{0,\per})} \le C.
	\]
\end{lemma}

We also block-decompose the operator $\cL_\bz$, and we write
\begin{equation} \label{eq:blockL}
	\cL_\bz = \left( \begin{array}{c | c } 
	\Lambda_\bz & l_\bz^* \\
	\hline \\
	l_\bz & L_\bz
	\end{array} \right),
\end{equation}
with $\Lambda_\bz = P_c \cL_\bz P_c \in \R$, $l_\bz = P_0 \cL_\bz P_c \in L^2_{0,\per}$, $L_\bz = P_0 \cL_\bz P_0 \in \cB(L^2_{0,\per})$. We identify these different objects in the next lemma, whose proof is given in Section~\ref{sec:proof:propcL}.

\begin{lemma}\label{lem:propcL}
There exist $r_1 > 0$ and $A > 0$ such that the following holds true.
\begin{itemize}
	\item[i)] The map $\bz \mapsto L_\bz$ is analytic from $\cB(\bnull, r_1) + \ri [-A, A]^3$ to the Banach space $\cB(L^2_{0,\per})$.
	\item[ii)] For $\bq \in \cB(\bnull, r_1)$, it holds that
	\begin{equation} \label{eq:Lambdaq}
		\Lambda_\bq = \dfrac{\bq^T}{|\bq|} M_1(\bq) \dfrac{\bq}{|\bq|},
	\end{equation}
	where $\bz \mapsto M_1(\bz)$ is an analytic map from $\cB(\bnull, r_1) + \ri [-A, A]^3$ to the space of $3 \times 3$ complex matrices. For all $\bq \in \cB(\bnull, r_1)$, it holds that $M_1(\bq)$ is an hermitian matrix, satisfying $M_1(\bq) = M_1(-\bq)$ and defined by
	\begin{equation} \label{eq:def:Aq}
	M_1(\bq)
	:= 
	\dfrac{8 \pi}{|\WS|} \sum_{n \le N < m} \fint_\BZ \dfrac{ \bra u_{n, \bq'} | (- \ri \nabla^1) u_{m, \bq' - \bq} \ket \bra u_{m, \bq' - \bq} | (- \ri \nabla^{1,T}) u_{n, \bq'} \ket  }{\left( \varepsilon_{m, {\bq' - \bq}} - \varepsilon_{n, {\bq'}} - \frac{| \bq - \bq'|^2}{2} + \frac{| \bq' |^2}{2} \right)^2 | \varepsilon_{m, {\bq' - \bq}} - \varepsilon_{n, {\bq'}} |} \rd \bq'.
	\end{equation}
	\item[iii)] For $\bq \in \cB(\bnull, r_1)$, it holds that
	\[
		l_\bq = \frac{\bq^T}{| \bq |} \bb(\bq),
	\]
	where $\bz \mapsto \bb(\bz)$ is analytic from $\cB(\bnull, r_1) + \ri [-A,A]^3$ to $\left(L^2_{0,\per} \right)^3$. For $\bq \in \cB(\bnull, r_1)$, it holds
	\begin{equation} \label{eq:def:bb}
	\begin{aligned}
	\bb(\bq) & := \dfrac{-2(4\pi)^{1/2} }{| \WS |^{1/2}}\sum_{n \le N < m}\fint_\BZ  \rd \bq' \\
		& \quad
	\left( \dfrac{ \bra u_{m, \bq' - \bq}  | (-\ri \nabla^1) u_{n,\bq'} \ket  }{| \varepsilon_{m, \bq' - \bq} - \varepsilon_{n, \bq'} | \left( \varepsilon_{m, {\bq' - \bq}} - \varepsilon_{n, {\bq'}} - \frac{| \bq |^2}{2} + \bq \cdot \bq' \right)} \right) \left[ (\sqrt{w_{c}})_\bq P_0 \left( u_{m, \bq' - \bq} u_{n, \bq'} \right)  \right].
	\end{aligned}
\end{equation}
\end{itemize}
\end{lemma}


The matrix $M_1(\bnull)$ is an important object for the macroscopic properties of the crystal. This matrix is sometimes called the macroscopic inverse dielectric $3 \times 3$ matrix. Let us state some properties of this matrix (see Section~\ref{sec:proof:A0} for the proof).
\begin{lemma}[Properties of $M_1(\bnull)$] \label{lem:A0}
 It holds that $M_1(\bnull)$ is a $3 \times 3$ positive definite hermitian matrix, and that
	\begin{equation} \label{eq:def:A0}
		M_1(\bnull) =  \dfrac{4 \pi}{2 \ri \pi | \WS |} \oint_{\sC} \VTr \left( \dfrac{1}{(\lambda - H_0)^2} \left(-\ri \nabla \right) \dfrac{1}{(\lambda - H_0)^2} \left(-\ri \nabla^T\right) \right) \rd \lambda.
	\end{equation}
	If the crystal is isotropic cubic (see Definition~\ref{def:isotropic_cubic}), then $M_1(\bnull)$ is proportional to the identity matrix. 
\end{lemma}

Formula~\eqref{eq:def:A0} is new to our knowledge, and provides an elegant definition for the macroscopic inverse dielectric $3 \times 3$ matrix. We emphasize that this matrix is not the inverse of the macroscopic dielectric $3 \times 3$ matrix (see next section).

\subsection{Bloch transform of the operator $(1 + \cL)^{-1}$.}

We finally study the Bloch decomposition of $(1 + \cL)^{-1}$. It holds that
\[
	\cZ (1 + \cL)^{-1} \cZ^{-1} = \fint_{\BZ}^\oplus (1 + \cL_\bq)^{-1} \rd \bq.
\]
Recall that $\cL$ is a bounded self-adjoint positive operator on $L^2(\R^3)$, so that the operators $\cL_\bq$ are bounded self-adjoint positive operator on $L^2_\per$ for almost all $\bq \in \BZ$. In particular, the operator $(1 + \cL_\bq)$ is indeed invertible, and from the inequality $ 1 - x \le (1+x)^{-1} \le 1$, valid for $x \ge 0$, it holds
\begin{equation} \label{eq:ineqLq}
	1 - \cL_\bq \le (1 + \cL_\bq)^{-1} \le 1.
\end{equation}
We study the operators $(1 + \cL_\bq)^{-1} $ in two different regimes. From Lemma~\ref{lem:analLq}, we deduce the following lemma (see Section~\ref{sec:proof:analDq} for the proof). 
\begin{lemma} \label{lem:analDq}
	There exist $C \in \R^+$ and $A > 0$ such that, for all $r > 0$, the map $\bz \mapsto  (1 + \cL_\bz)^{-1}$ is well defined and analytic from $\Omega_r +\ri [-A, A]^3$ to $\cB(L^2_\per)$, and
\[
	\sup_{\bz \in \Omega_r + \ri [-A, A]^3} \left\| (1 + \cL_\bz)^{-1} \right\|_{\cB(L^2_\per)} \le C.
\]
\end{lemma}

We now study the behavior of this operator as $\bq$ goes to $\bnull$. From the block decomposition~\eqref{eq:blockL}, we obtain with the Schur complement that
\[
	\left( 1 + \cL_\bz \right)^{-1} = \left( \begin{array}{c | c } 
	\delta_\bz & d_\bz^* \\
	\hline \\
	d_\bz & D_\bz
	\end{array} \right),
\]
with
\begin{align*}
	\delta_\bz & := \left[ (1 + \Lambda_\bz) - l_\bz^* (1 + L_\bz)^{-1} l_\bz \right]^{-1} , \\
	d_\bz & := - \delta_\bz (1 + L_\bz)^{-1}  l_\bz , \\
	D_\bz & := (1 + L_\bz)^{-1} + \delta_\bz (1 + L_\bz)^{-1} l_\bz l_\bz^* (1 + L_\bz)^{-1}.
\end{align*}
Together with Lemma~\ref{lem:propcL}, we deduce the following lemma, which is an extension of~\cite[Lemma 6]{Cances2010}
\begin{lemma} \label{lem:propDq}
There exist $r_1 > 0$ and $A > 0$ such that the following holds true.
\begin{itemize}
	\item[i)] For all $\bq \in \cB(\bnull, r_1)$, it holds that
	\begin{equation} \label{eq:Lambdaq}
		\delta_\bq = \dfrac{| \bq |^2}{ \bq^T M(\bq) \bq},
	\end{equation}
	where $\bz \mapsto M(\bz)$ is an analytic map from $\cB(\bnull, r_1) + \ri [-A, A]^3$ to the space of $3 \times 3$ complex matrices, defined by
	\begin{equation} \label{eq:def:A1}
		M(\bz) = 1 + M_1(\bz) - \bb^*(\bz) \left( 1 + L_\bz \right)^{-1} \bb(\bz).
	\end{equation}
	 For all $\bq \in \cB(\bnull, r_1)$, $M(\bq)$ is an hermitian $3 \times 3$ matrix, and it holds that $M(\bq) = M(-\bq)$. 
	\item[ii)] For $\bq \in \cB(\bnull, r_1)$, it holds that
	\[
		d_\bq = - \frac{| \bq | \bq^T} {\bq^T M(\bq) \bq} \bc(\bq),
	\]
	where $\bz \mapsto \bc(\bz)$ is an analytic map from $\cB(\bnull, r_1) + \ri [-A, A]^3$ to $\left( L^2_{0,\per} \right)^3$, defined by
	\[
		\forall j \in \{1,2,3,\}, \quad c_j(\bz) := \left( (1 + \cL_\bq )^{-1} b_j(\bq) \right).
	\]
	\item[iii)]  For $\bq \in \cB(\bnull, r_1)$, it holds that
	\[
		D_\bq = (1 + L_\bq)^{-1} + \dfrac{\bq^T C(\bq) \bq}{\bq^T M(\bq) \bq},
	\]
	where the map $\bz \mapsto C(\bz)$ is an analytic map from $\cB(\bnull, r_1) + \ri [-A, A]^3$ to the space of $3 \times 3$ matrices with components in $\cB(L^2_{0,\per})$, defined by
	\[
		\forall i,j \in \{1,2,3,\}, \quad C_{i,j}(\bz) = c_i(\bz) c_j^*(\bz).
	\]
\end{itemize}
\end{lemma}

\begin{remark} \label{rem:A1ge1}
	For $\bq \in \cB(\bnull, r_1)$, we obtain from~\eqref{eq:ineqLq} for $\bq \in \cB(\bnull, r_1)$ that
	\[
		 1 - \dfrac{\bq^T M_1(\bq) \bq}{| \bq |^2} \le \dfrac{| \bq |^2}{\bq^T M(\bq) \bq} \le 1
	\]
We deduce from the second inequality to the leading order in $\bq \to \bnull$ that $M(\bnull) \ge 1$.
\end{remark}

The matrix $M(\bnull)$ is an important object in materials science. Let us give the definition from Adler~\cite{Adler1962} and Wiser~\cite{Wiser1963}.
\begin{definition}[Macroscopic dielectric matrix] \label{def:A1}
		The $3 \times 3$ hermitian matrix $M(\bnull)$, where $M(\cdot)$ is defined in~\eqref{eq:def:A1} is the macroscopic dielectric $3 \times 3$ matrix of the crystal. If the crystal is isotropic cubic, then there exists $\epsilon \ge 1$ such that $M(\bnull) = \epsilon \bbI_3$. The number $\epsilon$ is the macroscopic dielectric constant of the crystal (in the rHF approximation).
\end{definition}

The proof that $M(\bnull)$ is proportional to the identity matrix in the isotropic cubic case is similar to Lemma~\ref{lem:A0}, although more involving.


\subsection{End of the proof of Theorem~\ref{th:chargedDefects}}

We now have all the tools to conclude the proof of Theorem~\ref{th:chargedDefects}. According to Lemma~\ref{lem:linearAndQuadratic}, Lemma~\ref{lem:linear} and Lemma~\ref{lem:intermediate}, it only remains to prove that $\cL_{2,2,\nu}^L$ defined in~\eqref{eq:diffQuad2} satisfies an estimate of the form
\begin{equation} \label{eq:restQuadratic}
	\left| \cJ_{2,2,\nu}^L + \dfrac{2 \pi \fa}{L | \WS |} q^2 \right| \le C \dfrac{\left\| \nu \right\|^2_{L^2(\R^3)}}{L^3},
\end{equation}
where $C \in \R^+$ is independent of $L$, and where we recall that $q = \int_{\R^3} \nu$. From Section~\ref{ssec:Bloch}, it holds that
\[
	\left\bra \left( 1 + \cL \right)^{-1} \sqrt{v_c} \nu, \sqrt{v_c} \nu  \right\ket_{L^2(\R^3)}
	= 
	\fint_\BZ F_\nu(\bq) \rd \bq,
\]
where we defined
\[
	\forall \bq \in \R^3, \quad F_\nu(\bq) := \left\bra \left( 1 + \cL_\bq \right)^{-1} (\sqrt{v_c})_\bq \nu_\bq, (\sqrt{v_c})_\bq \nu_\bq  \right\ket.
\]
Similarly, from Section~\ref{sec:supercellBloch}, it holds that
\[
	\left\bra \left( 1 + \widetilde{\cL^L} \right)^{-1} \sqrt{v_c^L} \nu_L, \sqrt{v_c^L} \nu_L  \right\ket_{L^2_{0,\per}(\WS_L)} 
	= 
	\dfrac{1}{L^3} \sum_{\bQ \in \Lambda_L \setminus \{ \bnull \}} F_\nu^L(\bQ),
\]
with
\[
	F_\nu^L(\bQ) := \left\bra \left( 1 + \widetilde{\cL_\bQ^L} \right)^{-1} (\sqrt{v_c^L})_\bQ \nu_{L,\bQ}, (\sqrt{v_c^L})_\bQ \nu_{L,\bQ}  \right\ket.
\]
Altogether, we obtain that
\[
	\cJ_{2,2,\nu}^L = \frac12 \left( \dfrac{1}{L^3} \sum_{\bQ \in \Lambda_L \setminus \{ \bnull \}} F_\nu^L(\bQ)  -  \fint_\BZ F_\nu(\bq) \rd \bq \right).
\]
We split the difference in two parts, namely
\begin{align}
 \cJ_{2,2,\nu}^L 
 	= & \frac12 \left( \dfrac{1}{L^3} \sum_{\bQ \in \Lambda_L \setminus \{ \bnull \}} \left( F_\nu^L(\bQ) - F_\nu(\bQ) \right) \right) +  \label{eq:FMinusFL}\\
	& \quad + \frac12 \left( \dfrac{1}{L^3} \sum_{\bQ \in \Lambda_L \setminus \{ \bnull \}} F_\nu(\bQ) - \fint_\BZ F_\nu(\bq) \rd \bq \right).  \label{eq:intMinusRiemann}
\end{align}
The term in~\eqref{eq:intMinusRiemann} is the difference between a Riemann sum and the corresponding integral. The term in~\eqref{eq:FMinusFL} compares the functions $F_\nu$ and $F^L_\nu$. This term is controlled thanks to Lemma~\ref{lem:LQ_LLQ} (see Section~\ref{sec:proof:FMinusFL} for the proof).

\begin{lemma} \label{lem:eqFMinusFL}
	There exists $C \in \R^+$ and $\alpha > 0$ such that, for all $L \ge L^*$ and all $\nu \in \cN(\eta)$, it holds that
	\[
		\left| \dfrac{1}{L^3} \sum_{\bQ \in \Lambda_L \setminus \{ \bnull \}} \left( F_\nu^L(\bQ) - F_\nu(\bQ) \right) \right|
		\le
		C \left\| \nu \right\|^2_{L^2(\R^3)} \re^{-\alpha L}.
	\]
\end{lemma}

We now study~\eqref{eq:intMinusRiemann}. According to Lemma~\ref{lem:propDq}, the function $F$ in singular as $\bq$ approaches $\bnull$. In order to isolate the singularity, we construct a periodic cut-off function. Recall that the macroscopic dielectric $3 \times 3$ matrix $M(\bnull)$ satisfies $M(\bnull) \ge 1$ (see Remark~\ref{rem:A1ge1}). By continuity of $M(\cdot)$, there exists $r_2 > 0$ such that $M(\bq) \ge 1/2$ for all $| \bq | \le r_2$. We finally set $r = \min(r_1, r_2)$, where $r_1$ is chosen as in Lemma~\ref{lem:propDq}, and we introduce
\begin{equation} \label{eq:Psi}
	\Psi_\per (\bq) = \sum_{\bk \in \RLat} \psi(| \bq + \bk |)
\end{equation}
where $\psi : \R^+ \to \R^+$ is a non-increasing function satisfying $\psi(x) = 1$ for all $x < r/2$ and $\psi(x) = 0$ for all  $x > r$. We write
\[
	F_\nu(\bq) = F_{\nu,1}(\bq) + F_{\nu,2}(\bq)
	\quad \text{with} \quad
	F_{\nu,1}(\bq) = (1 - \Psi_\per)(\bq) F_\nu(\bq) 
	\quad \text{and} \quad
	F_{\nu,2}(\bq) =  \Psi_\per(\bq)F_\nu(\bq).
\]
The function $F_{\nu,1}$ is smooth on the whole space. We deduce the following result, whose proof is postponed until Section~\ref{sec:proof:F1}.
\begin{lemma} \label{lem:F1}
	for all $p \in \N^*$, there exists $C_p \ge 0$ such that for all $L \ge L^*$ and all $\nu \in \cN(\eta)$, it holds
	\begin{equation} \label{eq:F1}
		\left| \fint_{\BZ} F_{\nu,1}(\bq) \rd \bq - \dfrac{1}{L^3}\sum_{\bQ \in \Lambda_L\setminus \{ \bnull \}} F_{\nu,1}(\bQ)  \right| \le \dfrac{C_p}{L^p} \left\| \nu \right\|_{L^2(\R^3)}^2.
	\end{equation}
\end{lemma}

For the remaining $F_{\nu,2}$ term, we expect a much slower convergence, due to the singularity as $\bq \to \bnull$. Thanks to Lemma~\ref{lem:wc} and Lemma~\ref{lem:propDq}, it holds that, for $\bq \in \BZ$,
\begin{equation} \label{eq:singularF2}
	F_{\nu,2}(\bq) =  \Psi_\per(\bq) \left( 
	\dfrac{ r(\bq)}{\bq^T M(\bq) \bq} \right)
\end{equation}
where we set (we denote by $n_\bq := (\sqrt{w_c})_\bq \nu_\bq$ ot lighten the notation)
\begin{align*}
	r(\bq) & := 4 \pi \left| \bra e_\bnull | \nu_\bq \ket \right|^2 +  \bq \left( 2 \Re \left( \sqrt{4 \pi} \bra e_\bnull | \nu_\bq \ket \left\bra n_\bq | d_\bq \right\ket_{L^2_{0,\per}} 	
	+
	\left\bra D_\bq n_\bq | n_\bq \right\ket_{L^2_{0,\per}}
	\right)
		\right) \\
		& = - 2\sqrt{4 \pi} \Re \left( \bra e_\bnull | \nu_\bq \ket  \left\bra n_\bq | \bq^T \bc(\bq) \right\ket_{L^2_{0,\per}} \right) 
		+
		\left\bra \bq^T C(\bq)\bq \, n_\bq | n_\bq \right\ket_{L^2_{0,\per}}	+ \\
		& \quad 
		+ \bq^T M(\bq) \bq \left\bra \left( 1 + L_\bq \right)^{-1} n_\bq | n_\bq \right\ket_{L^2_{0,\per}}.
\end{align*}
From Lemma~\ref{lem:propDq}, we deduce the following properties of $r(\bq)$.
\begin{lemma} \label{lem:rq}
	The map $\bq \mapsto r(\bq)$ admits an analytical extension on $\cB(\bnull, r_1) + \ri [-A, A]^3$ for some $A > 0$. Moreover, it holds 
	\[
		r(\bnull) = 4 \pi \left| \bra e_\bnull | \nu \ket \right|^2 = \dfrac{4 \pi}{| \WS |} q^2.
	\]
\end{lemma}

We finally apply Lemma~\ref{lem:singularRiemann} to the function $F_{\nu,2}(\bq)$, which gives the rate of convergence of the Riemann sum to the corresponding integral of a function of the form~\eqref{eq:singularF2}. We obtain the following lemma, which concludes the proof of Theorem~\ref{th:chargedDefects}.

\begin{lemma}
	There exists $C \in \R^+$ such that, for all $L \ge L^*$ and all $\nu \in \cN(\eta)$, it holds that
	\[
		\left|\dfrac{1}{L^3} \sum_{\bQ \in \Lambda_L \setminus \{ \bnull \}} F_\nu(\bQ) - \fint_\BZ F_\nu(\bq) \rd \bq
		  + \dfrac{4 \pi \fa}{| \WS |} \dfrac{q^2}{L}\right| \le C \dfrac{\left\| \nu \right\|}{L^3}.
	\]
\end{lemma}

\section{Proofs of the results}
\label{sec:detailsProofs}

\subsection{Proof of Proposition~\ref{prop:isotropic}}
\label{sec:proof:isotropic}

\begin{proof}
Let us first compute $\fm$ defined in~\eqref{eq:def:Madelung}.  From the Fourier representation of $G_1$ in~\eqref{eq:G1Fourier} and the Fourier transform of the Coulomb potential
\[
	\dfrac{1}{| \bx |} = \dfrac{1}{2 \pi^2} \int_{\R^3} \dfrac{\re^{\ri \bk \cdot \bx}}{| \bk |^2},
\]
one obtains (we use the fact that $| \WS | \cdot | \BZ | = (2 \pi)^3$)
\[
	G_1(\bx) - \frac{1}{| \bx |} = \frac{| \BZ |}{2 \pi^2} \sum_{\bk \in \RLat} \fint_{\BZ} \left( \dfrac{\re^{\ri \bk \cdot \bx} \mathds{1}(\bk \neq \bnull)}{| \bk |^2} -  \dfrac{\re^{\ri (\bk + \bq) \cdot \bx}}{| \bk + \bq |^2} \right) \rd \bq.
\]
The sum is absolutely convergent thanks to the multipole expansion (see for instance~\eqref{eq:def:F1} below). The limit $\bx \to \bnull$ then leads to
\[
	\fm = \frac{| \BZ |}{2 \pi^2} \sum_{\bk \in \RLat} \fint_{\BZ} \left( \dfrac{\mathds{1}(\bk \neq \bnull)}{| \bk |^2} -  \dfrac{1}{| \bk + \bq |^2} \right) \rd \bq.
\]
Together with the definition of $\fa$ in~\eqref{eq:fa} with $M = \epsilon \mathbb{I}_3$ (see Definition~\ref{def:A1}), this leads to the desire result.
\end{proof}

\subsection{Proof of Lemma~\ref{lem:controlQ}}
\label{sec:proof:controlQ}

The proof follows the arguments in~\cite[Lemma 3]{Cances2010}. We provide it here to emphasize the role the size of the supercell $L$. We first state a supercell equivalent of~\cite[Lemma 5]{Cances2008} and~\cite[Lemma 1]{Cances2010} with uniform bounds in $L$.
We introduce for $\lambda \in \sC$, $\nu \in \cN(\eta)$ and $L \ge L^\ast$ the operators
\[
	B_1(\lambda, \nu, L) := (1 - \Delta^L) \dfrac{1}{\lambda - H_\nu^L}
	\quad \text{and} \quad
	B_2(\lambda, \nu, L) :=  \dfrac{1}{\lambda - H_\nu^L}(1 - \Delta^L).
\]
The following classical lemma is very useful. It can be proved following for instance the lines of~\cite[Lemma 5.2]{GL2015}. Recall that $\cB(E)$ denotes the Banach space of bounded operators on the Banach space $E$.
\begin{lemma} \label{lem:bounds_B}
	For all $\nu \in \cN(\eta)$, all $L \ge L^\ast$ and all $\lambda \in \sC$, the operator $\lambda - H_\nu^L$ is invertible, and there exists $ {C} \in \R^+$ such that
	\begin{align} \label{eq:bounds_B1_B2}
		\forall \nu \in \cN(\eta), \quad \forall L \ge L^\ast,  \quad \forall \lambda \in \sC, \quad  \left\|  {B_{1,2}}(\lambda, \nu,L) \right\|_{\cB(L^2_\per(\WS_L))} \le {C}.
	\end{align}
\end{lemma}

We deduce the following lemma.
\begin{lemma} \label{lem:commutator} \ \\
	i) For all $L \ge L^\ast$, the map $V^L \mapsto [\gamma_0^L, V^L]$ is continuous from $\cC_L'$ to $\fS_2^L$. Moreover, there exists $C \in \R^+$ such that
	\[
		\forall L \ge L^*, \quad \forall \, V^L \in \cC_L', \quad \left\|  [\gamma_0^L, V^L] \right\|_{\fS_2^L } \le C \| V^L \|_{\cC_L'}.
	\]
	ii) For all $L \ge L^\ast$, all $\nu \in \cN(\eta)$ and all $\lambda \in \sC$, the map $V^L \mapsto (\lambda - H_\nu^L)^{-1} V^L$ is continuous from $\cC_L'$ to $\fS_6^L)$. Moreover, there exists $C \in \R^+$ such that
	\[
		\forall L \ge L^*, \quad \forall \nu \in \cN(\eta), \quad \forall \lambda \in \sC, \quad \forall \, V^L \in \cC_L', \quad \left\|   (\lambda - H_\nu^L)^{-1} V^L\right\|_{\fS_6^L)} \le C \| V^L \|_{\cC_L'}.
	\]
	
\end{lemma}

\begin{proof}[Proof of Lemma~\ref{lem:commutator}]
	To prove the first point, we write 
	\begin{align*}
		[\gamma_0^L, V^L] & = \dfrac{1}{2 \ri \pi} \oint_\sC \left[ \dfrac{1}{\lambda - H_0^L}, V^L \right] \rd \lambda 
			= \dfrac{1}{2 \ri \pi} \oint_\sC \dfrac{1}{\lambda - H_0^L}\left[ \lambda - H_0^L, V^L \right]\dfrac{1}{\lambda - H_0^L} \rd \lambda \\
		& = \dfrac{1}{4 \ri \pi} \oint_\sC B_2(\lambda, 0, L) \dfrac{1}{1 - \Delta^L}\left[ \Delta^L, V^L \right]\dfrac{1}{1 - \Delta^L} B_1(\lambda, 0, L) \rd \lambda.
	\end{align*}
	On the other hand, we have
	\[
		\left[ \Delta^L, V^L \right] = \sum_{j=1}^3 \left[ \left(P_j^L\right)^2, V^L \right] =  \sum_{j=1}^3 P_j^L \left( P_j^L V^L \right) + \left( P_j^L V^L \right) P_j^L.
	\]
	Since $V^L \in \cC'_L$, it follows that $\left( P_j^L V^L \right) \in L^2_\per(\WS_L)$ with $\| P_j^L V^L \|_{L^2_\per(\WS_L)} \le \| V^L \|_{\cC'_L}$, for $j \in \{ 1,2,3 \}$. The result then follows from the fact that $P_j (1 - \Delta^L)^{-1}$ is uniformly bounded for $L \ge L^*$, $j \in \{1,2,3\}$, and the periodic Kato-Seiler-Simon inequality (see Lemma~\ref{lem:periodicKSS} and Corollary~\ref{cor:periodicKSS}). \\
	To prove the second point, we simply write that $(\lambda - H_\nu^L)^{-1} V^L = B_2(\lambda, \nu,L) (1 - \Delta^L)^{-1} V^L$, and use the fact that $V^L \in \cC_L' \hookrightarrow L^6_\per(\WS_L)$ with uniform bounds in $L$, together with Corollary~\ref{cor:periodicKSS} with $p = 6$.
\end{proof}

We now prove Lemma~\ref{lem:controlQ}. First, since $Q_\nu^L$ is the minimizer of~\eqref{eq:JnuL_infQ}, it holds that $\cF_\nu^L(Q_\nu^L) \le \cF_\nu^L(0)$, hence
\begin{equation} \label{eq:sumPositivTerms}
	 \Tr_{L^2_\per(\WS_L)} \left( \left[ H_0^L -\varepsilon_F \right]  Q_\nu^L \right) + \frac12 D_L \left( \rho_{Q_\nu^L} - \nu_L, \rho_{Q_\nu^L} - \nu_L \right) \le D_L(\nu_L, \nu_L).
\end{equation}
Since $Q_\nu^L$ is the difference of two projectors, we have
\begin{align*}
	\Tr_{L^2_\per(\WS_L)} \left( \left[ H_0^L -\varepsilon_F \right]  Q_\nu^L \right) 
		& = \Tr_{L^2_\per(\WS_L)} \left( \left| H_0^L -\varepsilon_F \right|  \left(Q_\nu^{++,L} - Q_\nu^{--,L} \right) \right)  \\
		& = \Tr_{L^2_\per(\WS_L)} \left( \left| H_0^L -\varepsilon_F \right|  \left( Q_\nu^L \right)^2 \right) \ge 0.
\end{align*}
 As a consequence, the two terms in the left-hand side of~\eqref{eq:sumPositivTerms} are positive. We deduce that $\| \rho_{Q_\nu^L} - \nu_L \|_{\cC_L} \le 2 \| \nu_L \|_{\cC_L}$, which is the first point of~\eqref{eq:controlV}. From the embedding $L_\per^{6/5} \hookrightarrow \cC_L$ with uniform bound in $L$, we obtain that there exists $C_1, C_2 \in \R^+$ such that, for all $L \ge L^\supp$, 
 \begin{equation} \label{eq:CL_L2}
 	\forall \nu \in \cN(\eta), \quad
	\left\| \nu \right\|_{\cC_L} \le C_1 \left\| \nu \right\|_{L^{6/5}_\per(L^\supp)} \le C_2 \left\| \nu \right\|_{L^{2}_\per(L^\supp)} = C_2 \left\| \nu \right\|_{L^{2}(\R^3)}.
 \end{equation}
The second part of~\eqref{eq:controlV} follows.

\medskip

To prove~\eqref{eq:controlQ1Q2}, we follow~\cite[Lemma 3]{Cances2010}. We first prove the assertion for $Q_{\nu,1}^L$. In the sequel, we use the notation
\begin{equation} \label{eq:P+-}
	P_+^L = (1 - \gamma_0^L) 
	\quad \text{and} \quad
	P_-^L = \gamma_0^L,
\end{equation}
and for $\alpha, \beta \in \{ +,-\}$, we denote by $Q_{\nu,1}^{\alpha \beta, L} := P_\alpha^L Q_{\nu,1}^L P_\beta^L$. Thanks to the Cauchy residual formula, it holds that $Q_{\nu,1}^{++,L} = Q_{\nu,1}^{--,L} = 0$, so that $Q_{\nu,1}^L = Q_{\nu,1}^{+-,L} + Q_{\nu,1}^{-+,L}$. Let us study the $Q_{\nu,1}^{+-,L}$ term (the study of the $Q_{\nu,1}^{-+,L}$ being similar). It holds
\begin{align*}
	(1 - \Delta^L)^{1/2} Q_{\nu,1}^{+-,L} 
		& = \dfrac{1}{2 \ri \pi} \oint_\sC (1 - \Delta^L) \dfrac{1 - \gamma_0^L}{\lambda - H_0^L} V_\nu^L \dfrac{\gamma_0^L}{\lambda - H_0^L} \\
		& =  \dfrac{1}{2 \ri \pi} \oint_\sC B_1(\lambda,0,L) (1 - \gamma_0^L) \left[ V_\nu^L, \gamma_0^L \right] B_2(\lambda, 0, L) \dfrac{1}{1 - \Delta^L} \rd \lambda.
\end{align*}
From Lemma~\ref{lem:bounds_B}, Lemma~\ref{lem:commutator} and the fact that $(1 - \gamma_0^L)$ and $(1 - \Delta^L)^{-1}$ are uniformly bounded in $\cS(L^2_\per(\WS_L))$ for $L \ge L^*$, we deduce that there exist $C_1, C_2, C_3 \in \R^+$ such that, for all $L \ge L^*$,
\[
	\left\| (1 - \Delta^L)^{1/2} Q_{\nu,1}^{+-,L}  \right\|_{\fS_2^L} \le C_1 \left\|  \left[ V_\nu^L, \gamma_0^L \right] \right\|_{\fS_2^L} \le C_2 \left\| V_\nu^L \right\|_{\cC'} \le C_3 \| \nu \|_{\cC},
\]
where we used~\eqref{eq:controlV} for the last inequality.

\medskip

On the other hand, since $(1 - \Delta^L)^{-1} \in \fS_2^L$, we also deduce that $Q_{\nu,1}^{+-,L} \in \fS_1^L$, so that we can consider its trace, and get $\Tr_{L^2_\per(\WS_L)}(Q_{\nu,1}^{+-,L}) = 0$. The third point of Lemma~\ref{lem:controlQ} follows.

\medskip
We now prove the result for $\widetilde{Q_{\nu,2}^L}$. For $k \in \N$, we denote by $Q_{\nu,k}^L$ and $\widetilde{Q_{\nu,k}^L}$ the operators
\[
	Q_{\nu,k}^L = \dfrac{1}{2 \ri \pi} \oint_\sC \dfrac{1}{\lambda - H_0^L} \left( V_\nu^L \dfrac{1}{\lambda - H_0^L} \right)^k \rd \lambda
	\quad \text{and} \quad
	\widetilde{Q_{\nu,k}^L} = \dfrac{1}{2 \ri \pi} \oint_\sC \dfrac{1}{\lambda - H_\nu^L} \left( V_\nu^L \dfrac{1}{\lambda - H_0^L} \right)^k \rd \lambda.
\]
Note that for $k \ge 2$, it holds $\widetilde{Q_{\nu,2}^L} = \sum_{l=2}^{k-1} {Q_{\nu,l}^L} + \widetilde{Q_{\nu,k}^L}$. We proceed in two steps. 

\medskip

\underline{\textbf{Step 1:} The operators $Q_{\nu,l}^L$, $l \in \{ 1,2,3,4,5\}$ satisfy bounds similar to~\eqref{eq:controlQ1Q2}.} \\
We do the proof for $Q_{\nu,2}^L$, the proof begin similar for the other cases. For $\alpha, \beta, \gamma \in \{ -,+ \}$, we introduce 
\[
	Q_{\nu,2}^{\alpha \beta \gamma, L} := \dfrac{1}{2 \ri \pi} \oint_\sC \dfrac{P_\alpha^L}{\lambda - H_0^L} V_\nu^L \dfrac{P_\beta^L}{\lambda - H_0^L} V_\nu^L \dfrac{P_\gamma^L}{\lambda - H_0^L} \rd \lambda,
\]
where the operators $P^L_\pm$ were defined in~\eqref{eq:P+-}. Thanks to the Cauchy residual formula, it holds $Q_{\nu,2}^{+++, L} = Q_{\nu,2}^{---, L} = 0$. In each other terms, the motif $P_+^L  V_\nu^L P_-^L$ or $P_+^L  V_\nu^L P_-^L$ appears at least once. Thanks to Lemma~\ref{lem:commutator}, there exists $C \in \R^+$ such that, for all $L \ge L^*$,
\[
	\left\| P_+^L  V_\nu^L P_-^L \right\|_{\fS_2^L} = \left\| (1 - \gamma_0^L) [V_\nu^L, \gamma_0^L] \right\|_{\fS_2^L} \le C \| \nu \|_{L^2(\R^3)}
	\quad \text{and} \quad
	\left\| P_-^L  V_\nu^L P_+^L \right\|_{\fS_2^L} \le C \| \nu \|_{L^2(\R^3)}.
\]
The other motifs are bounded thanks to the second point of Lemma~\ref{lem:commutator}. Altogether, we obtain
\[
	\exists C \in \R^+, \quad \forall 2 \le l \le 5, \quad \forall L \ge L^*, \quad \left\| (1 - \Delta^L)^{1/2} Q_{\nu,l}^L \right\|_{\fS_2^L} \le C \| \nu \|^l_{L^2(\R^3)}.
\]
On the other hand, it holds $Q_{\nu,2}^{--,L} = Q_{\nu,2}^{-+-,L}$ and $Q_{\nu,2}^{++,L} = Q_{\nu,2}^{+-+,L}$, so that for these operators, the motif $P_+^L  V_\nu^L P_-^L$ or $P_+^L  V_\nu^L P_-^L$ appears at least twice. Following the same arguments leads to
\[
	\exists C \in \R^+, \quad \forall 2 \le l \le 5, \quad \forall L \ge L^*, \quad \sum_{\alpha \in \{-,+\}} \left\| (1 - \Delta^L)^{1/2} Q_{\nu,l}^{\alpha \alpha, L} (1 - \Delta^L)^{1/2} \right\|_{\fS_1^L} \le C \| \nu \|^l_{L^2(\R^3)}.
\]

\underline{\textbf{Step 2:} The operator $\widetilde{Q_{\nu,6}^L}$, satisfies bounds similar to~\eqref{eq:controlQ1Q2}.} \\
Actually, we can prove that $(1 - \Delta^L)^{1/2} \widetilde{Q_{\nu,6}^L} (1 - \Delta^L)^{1/2} \in \fS_1^L$ with norm uniformly bounded in $L$. This time, we simply use the fact that the motif $V_\nu^L (\lambda - H_0^L)^{-1}$ appears six times, and we use the second point of Lemma~\ref{lem:commutator} to bound the operator $(1 - \Delta^L)^{1/2} \widetilde{Q_{\nu,6}^L} (1 - \Delta^L)^{1/2}$ in $\fS_1^L$ uniformly in $L$. More specifically,
\[
	\exists C \in \R^+, \quad \forall \nu \in \cN(\eta), \quad \forall L \ge L^*, \quad \left\| (1 - \Delta^L)^{1/2} \widetilde{Q_{\nu,6}^L} (1 - \Delta^L)^{1/2} \right\|_{\fS_1^L} \le C \| \nu \|_{L^2(\R^3)}^6.
\]

The proof of~\eqref{eq:controlQ1Q2} follows.

\subsection{Proof of Lemma~\ref{lem:kinetic_ho}}
\label{sec:proof:kinetic_ho}

We use the decomposition~\eqref{eq:decompositionQnuL}, and write
\begin{equation} \label{eq:decompositionQnuL^2}
	\left(Q_{\nu}^L\right)^2 = \left(Q_{\nu,1}^L + \widetilde{Q_{\nu,2}^L} \right)^2 = \left(Q_{\nu,1}^L\right)^2 + \left(\widetilde{Q_{\nu,2}^L}\right)^2 
	+ Q_{\nu,1}^L \widetilde{Q_{\nu,2}^L} + \widetilde{Q_{\nu,2}^L} Q_{\nu,1}^L.
\end{equation}
On the other hand, for any $Q^L_a,  Q^L_b \in \cQ^L$, it holds that
\begin{align*}
 \Tr_{L^2_\per(\WS_L)} \left( \left| H_0^L - \varepsilon_F \right| Q_a^L Q_b^L \right) 
	=  \Tr_{L^2_\per(\WS_L)} & \left(  \left| H_0^L - \varepsilon_F \right|^{1/2} (1 - \Delta^L)^{-1/2} (1 - \Delta^L)^{1/2} Q_a^L . \right. \\
	& \quad \left.	Q_b^L (1 - \Delta^L)^{1/2} (1 - \Delta^L)^{-1/2}  \left| H_0^L - \varepsilon_F \right|^{1/2} \right).
\end{align*}
The operators $(1 - \Delta^L)^{-1/2}  \left| H_0^L - \varepsilon_F \right|^{1/2}$ and $ \left| H_0^L - \varepsilon_F \right|^{1/2} (1 - \Delta^L)^{-1/2} $ are uniformly bounded in $L$ (this can be shown as in Lemma~\ref{lem:bounds_B}) by some constant $C \in \R^+$. We therefore get
\begin{align}
 \left| \Tr_{L^2_\per(\WS_L)} \left( \left| H_0^L - \varepsilon_F \right| Q_a^L Q_b^L \right) \right| 
 & \le C^2 \left\| (1 - \Delta^L)^{1/2} Q_a^L\right\|_{\fS_2^L} \left\| (1 - \Delta^L)^{1/2} Q_b^L\right\|_{\fS_2^L} \nonumber \\
 & \le C^2 \left\| Q_a^L \right\|_{\cQ^L}  \left\| Q_b^L \right\|_{\cQ^L} . \label{eq:continuityKinetic}
\end{align}
Lemma~\ref{lem:kinetic_ho} then follows from~\eqref{eq:continuityKinetic}, the decomposition~\eqref{eq:decompositionQnuL^2} and Lemma~\ref{lem:controlQ}.

\subsection{Proof of Lemma~\ref{lem:controlHoDL}}
\label{sec:proof:controlHoDL}

The Cauchy-Schwarz inequality leads to
\begin{align*}
	& \left| \left\bra \left[ \sqrt{v_c^L} \left( \rho_{Q_{\nu,1}^L} - \nu_L \right) -  \left( 1 + \cL^L \right) \sqrt{v_c^L} \nu_L \right], \sqrt{v_c^L} \nu_L  \right\ket_{L^2_{0,\per}(\WS_L)} \right| \le \\
	& \qquad \qquad \left\| \sqrt{v_c^L} \left( \rho_{Q_{\nu,1}^L} - \nu_L \right) -  \left( 1 + \cL^L \right) \sqrt{v_c^L} \nu_L \right\|_{L^2_{0,\per}(\WS_L)} \left\| \sqrt{v_c^L} \nu_L \right\|_{L^2_{0,\per}(\WS_L)}.
\end{align*}
The first term of the right-hand side is controlled by some $C \left\| \nu \right\|^2_{L^2(\R^3)}$ thanks to Lemma~\ref{lem:linearPartQ1}. Since $\sqrt{v_c^L}$ is an operator from $\cC_L$ to $L^2_\per(\WS_L)$ bounded by $1$ (see Remark~\ref{rem:sqrtVc}), we obtain $\left\| \sqrt{v_c^L} \nu_L \right\|_{L^2_\per(\WS_L)} \le \left\| \nu_L \right\|_{\cC_L}$, and we conclude as in~\eqref{eq:CL_L2}.

\subsection{Proof of Lemma~\ref{lem:linear}}
\label{sec:proof:linear}

Recall that the support of $\nu$ is contained in $L^\supp \WS$ (see Section~\ref{sec:defect}), so that
\[
	\left| \int_{\R^3} \left( V_{0}^L - V_{0} \right) \nu \right| = \left| \int_{L^\supp \WS} \left( V_{0}^L - V_{0} \right) \nu \right| \le 
	\left\| V_{0}^L - V_{0} \right\|_{L^2(L^\supp \WS)} \left\| \nu \right\|_{L^2(\R^3)},
\]
where we used the Cauchy-Schwarz inequality. Since $V_0^L$ and $V_0$ are both $\Lat$-periodic, it holds $\left\| V_{0}^L - V_{0} \right\|_{L^2(L^\supp \WS)} = L^\supp \left\| V_{0}^L - V_{0} \right\|_{L^2(\WS)}$. On the other hand, from~\eqref{eq:sc} and~\eqref{eq:scL}, it holds
\[
	V_{0}^L - V_{0} = \left( \rho_{\gamma_0^L} - \rho_{\gamma_0} \right) \ast_\WS G_1.
\]
The result then follows from the continuous embedding $L^\infty_\per(\WS) \hookrightarrow L^2_\per(\WS)$, the fact that the convolution by $G_1$ is continuous from $L^2_\per(\WS)$ to $L^\infty_\per(\WS)$ and the last part of Theorem~\ref{th:expCV}.

\subsection{Proof of Lemma~\ref{lem:intermediate}}
\label{sec:proof:intermediate}

The proof is a direct consequence of the following lemma.

\begin{lemma}
\label{lem:expCVoperators}
	 There exist $C \in \R^+$ and $\alpha > 0$ such that
	\[
		\forall L \ge L^*, \quad \left\| \left( 1 + \widetilde{\cL^L} \right)^{-1} - \left( 1 + \cL^L\right)^{-1}  \right\|_{\cS(L^2_{0,\per}(\WS_L))} \le C \re^{- \alpha L}.
	\]
\end{lemma}

\begin{proof}[Proof of Lemma~\ref{lem:expCVoperators}]

We first note that $\widetilde{H_0^L} - H_0^L = V_0 - V_0^L = G_1 \ast_{\WS} \left( \rho_{\gamma_0} - \rho_{\gamma_0^L} \right)$. Since the convolution by $G_1$ is bounded on $L^\infty_\per(\R^3)$, we obtain from the last part of Theorem~\ref{th:expCV} that that there exist $C \in \R^*$ and $\alpha >0$ such that
\begin{equation} \label{eq:expCVH0}
	\forall L \ge L^*, \quad \left\| \widetilde{H_0^L} - H_0^L \right\|_{\cS(L^2_{\per}(\WS_L))} \le C \re^{- \alpha L}.
\end{equation}
On the other hand, we have, for $V_a^L, V_b^L \in \cC_L'$, that
\begin{align*}
	& \left\bra \widetilde{\chi^L} V_a^L , V_b^L \right\ket_{\cC_L, \cC_L'} - \left\bra {\chi^L} V_a^L , V_b^L \right\ket_{\cC_L, \cC_L'} \\
	& \quad = \dfrac{1}{2 \ri \pi} \oint_\sC \Tr_{L^2_\per(\WS_L)} \left( \dfrac{1}{\lambda - \widetilde{H_0^L}} V_a^L \dfrac{1}{\lambda - \widetilde{H_0^L}} V_b^L - \dfrac{1}{\lambda - {H_0^L}} V_a^L \dfrac{1}{\lambda - {H_0^L}} V_b^L\right) \rd \lambda \\
	& \quad = \dfrac{1}{2 \ri \pi} \oint_\sC \Tr_{L^2_\per(\WS_L)} \left( \dfrac{1}{\lambda - \widetilde{H_0^L}} \left( \widetilde{H_0^L} - H_0^L \right) \dfrac{1}{\lambda - {H_0^L}}  V_a^L \dfrac{1}{\lambda - \widetilde{H_0^L}} V_b^L \right) \rd \lambda + \\
	& \quad \quad + \dfrac{1}{2 \ri \pi} \oint_\sC \Tr_{L^2_\per(\WS_L)} \left( \dfrac{1}{\lambda - {H_0^L}}   V_a^L \dfrac{1}{\lambda - \widetilde{H_0^L}} \left( \widetilde{H_0^L} - H_0^L \right) \dfrac{1}{\lambda - {H_0^L}} V_b^L \right) \rd \lambda.
\end{align*}
Using estimates similar to the ones used in the proof of Lemmas~\ref{lem:controlQ} and Lemma~\ref{lem:propertiesChiL}, together with the estimate~\eqref{eq:expCVH0}, we deduce that there exist $C \in \R^+$ and $\alpha >0$ such that
\begin{equation} \label{eq:expCVChi}
	\forall L \ge L^*, \quad \left\| \widetilde{\chi^L} - \chi^L \right\|_{\cB(\cC_L', \cC_L)} \le C \re^{- \alpha L}.
\end{equation}
Finally, from the definitions~\eqref{eq:def:LL} and~\eqref{eq:def:widetildeLL}, it holds that
\[
	 \left( 1 + \widetilde{\cL^L} \right)^{-1} - \left( 1 + \cL^L\right)^{-1} = 
	  \left( 1 + \widetilde{\cL^L} \right)^{-1} \sqrt{v_c^L} \left( \widetilde{\chi^L} - \chi^L \right) \sqrt{v_c^L} \left( 1 + \cL^L\right)^{-1}.
\]
The result then follows from Lemma~\ref{lem:vcL}, Lemma~\ref{lem:propertiesLL} and~\eqref{eq:expCVChi}.
\end{proof}

\subsection{Proof of Lemma~\ref{lem:LQ_LLQ}}
\label{sec:proof:LQ_LLQ}

Let us first extend the definition of $\widetilde{\cL^L_\bQ}$, initially defined for $\bQ \in \Lambda_L$, to all $\bq \in \BZ \setminus \{\bnull\}$, with
\[
	\forall \bq \in \BZ \setminus \{\bnull\}, \quad
	\widetilde{\cL^L_\bq} : f \in L^2_\per \mapsto 
	 \dfrac{1}{L^3} \sum_{\bQ' \in \Lambda_L} \rho \left[  \dfrac{1}{2 \ri \pi}  \oint_\sC \dfrac{1}{\lambda - H_{\bQ'}} {(\sqrt{v_c})_\bq f} \dfrac{1}{\lambda - H_{\bQ' - \bq}} \rd \lambda \right].
\]
Let $f,g \in L^2_\per$. For $\bq \in \BZ \setminus \{\bnull\}$, we write for simplicity $F_\bq := (\sqrt{v_c})_\bq f$ and $G_\bq := (\sqrt{v_c})_\bq g$. 
From the definition~\eqref{eq:BlochSqrtVc}, we deduce that for $\bq \in \BZ \setminus \{ \bnull \}$, it holds that $F_\bq, G_\bq \in L^2_\per$ with
\begin{equation} \label{eq:boundFq}
	\left\| F_\bq \right\|_{L^2_\per} \le \dfrac{\sqrt{4 \pi}}{| \bq|} \left\| f_\bq \right\|_{L^2_\per}
	\quad \text{and} \quad
	\left\| G_\bq \right\|_{L^2_\per} \le \dfrac{\sqrt{4 \pi}}{| \bq|} \left\| g_\bq \right\|_{L^2_\per}
\end{equation}

Finally, for $\lambda \in \sC$, we introduce
\[
	\forall \bz \in \C^3, \quad  K^{f,g}_{\lambda,\bq}(\bz) := \Tr_{L^2_\per(\WS)} \left( \dfrac{1}{\lambda - H_{\bz}} F_\bq \dfrac{1}{\lambda - H_{\bz - \bq}} G_\bq \right).
\]
With all these notation, the quantity that we want to control is
\begin{align*}
	\left| \left\bra f \big| \left( \widetilde{\cL^L_\bq} - \cL_\bq \right) g \right\ket \right|
		& = \left| \dfrac{1}{2 \ri \pi} \oint_{\sC} \fint_{\BZ} K^{f,g}_{\lambda,\bq}(\bq') \rd \bq' - \dfrac{1}{L^3} \sum_{\bQ' \in \Lambda_L} K^{f,g}_{\lambda,\bq}(\bQ') \right| \\
		& \le \dfrac{| \sC |}{2 \pi} \sup_{\lambda \in \sC} \left|  \fint_{\BZ} K^{f,g}_{\lambda,\bq}(\bq') \rd \bq' - \dfrac{1}{L^3} \sum_{\bQ' \in \Lambda_L} K^{f,g}_{\lambda,\bq}(\bQ')  \right|.
\end{align*}
We recognize the difference between an integral and a corresponding Riemann sum. Let us study the integrand $K^{f,g}_{\lambda,\bq}$.

\begin{lemma} \label{lem:analK}
	There exist $A > 0$ and $C \in \R^+$ such that, for all $\bq \in \BZ \setminus \{\bnull\}$, all $f,g \in L^2_\per(\WS)$ and all $\lambda \in \sC$, the function $z \mapsto K^{f,g}_{\lambda,\bq}(\bz)$ is an $\RLat$-periodic analytic function on $\R^3 + \ri [-A, A]^3$, with
	\begin{equation} \label{eq:upperK}
		\sup_{\bz \in \R^3 + \ri [-A, A]^3} \left| K^{f,g}_{\lambda,\bq}(\bz) \right| 
		\le  \dfrac{C}{| \bq |^2} \| f \|_{L^2_\per}  \| g \|_{L^2_\per} .
	\end{equation}
\end{lemma}

\begin{proof}[Proof of Lemma~\ref{lem:analK}]
Let us begin with the $\RLat$-periodicity. From the covariant property \eqref{eq:rotation}, we deduce that
\begin{align*}
	\forall \bq' \in \R^3, \quad \forall \bk \in \RLat, \quad K^{f,g}_{\lambda,\bq}(\bq' + \bk) 
	& =  \Tr_{L^2_\per(\WS)} \left( \dfrac{1}{\lambda - H_{\bq' + \bk}} F_\bq \dfrac{1}{\lambda - H_{\bq' + \bk - \bq}} G_\bq \right) \\
		& = \Tr_{L^2_\per(\WS)} \left( U_\bk \dfrac{1}{\lambda - H_{\bq' + \bk}} U_{-\bk} F_\bq U_{\bk} \dfrac{1}{\lambda - H_{\bq' - \bq}} U_{-\bk} G_\bq \right).
\end{align*}
The result follows by rotating the unitary operator $U_{\bk}$ under the trace, and using the fact that for any multiplication operator $V$, it holds that $U_{-\bk} V U_{\bk} = V$. The $\RLat$-periodicity on $\R^3$ will eventually transfer into a $\RLat$-periodicity on $\R^3 + \ri [-A, A]$ by analyticity. \\
Let us prove that these maps are well-defined on some complex strip. We choose $A > 0$ as in Lemma~\ref{lem:bounds_Bq}, and recall that $B_1$ and $B_2$ were defined in~\eqref{eq:def:B1q_B2q}. For $\bz \in 2\BZ + \ri [-A, A]^3$, it holds that
\[
	K^{f,g}_{\lambda,\bq}(\bz) = \Tr_{L^2_\per(\WS)} \left( B_2(\lambda, \bz) \left[ \dfrac{1}{1 - \Delta^1} F_\bq \right] B_2(\lambda, \bz - \bq) \left[ \dfrac{1}{1 - \Delta^1} G_\bq \right] \right).
\]
According to Lemma~\ref{lem:bounds_Bq}, the operator $B_2(\lambda, \bz)$ and $B_2(\lambda, \bz - \bq)$ are uniformly bounded for $\lambda \in \sC$ and $\bz \in 2\BZ + \ri [-A, A]^3$. Together with the periodic Kato-Seiler-Simon inequality (see Corollary~\ref{cor:periodicKSS} with $p = 2$), we deduce that there exists $C \in \R^+$ such that
\[
	\forall \bq \in \BZ \setminus \{ \bnull \}, \quad \forall \bz \in 2 \BZ + \ri [-A,A]^3, \quad \forall f,g \in L^2_\per, \quad \left| K^{f,g}_{\lambda,\bq}(\bz) \right| \le C \| F_\bq \|_{L^2_\per} \| G_\bq \|_{L^2_\per}.
\]
Finally, using~\eqref{eq:BlochSqrtVc}, we deduce~\eqref{eq:upperK}. 

\medskip

We finally prove that the maps are analytic on $2 \BZ + \ri [-A, A]^3$ which, by periodicity, will imply the analyticity on the whole strip $\R^3 + \ri [-A,A]^3$. It is enough to prove that these maps are derivable on $\R^3 + \ri [-A, A]^3$. We notice that, for $j \in \{ 1,2,3 \}$, it holds that
\[
	\partial_{z_j} B_2(\lambda, \bz) = \dfrac{1}{\lambda - H_\bz} \left( - \ri P^1_j + z_j \right) \dfrac{1}B_2(\lambda, \bz) = B_2(\lambda, \bz) \left( \dfrac{- \ri P^1_j + z_j}{1 - \Delta^1} \right) B_2(\lambda, \bz),
\]
so that $\partial_{z_j} B_2(\lambda, \bz)$ is a bounded operator. The result easily follows from estimates similar to the ones used previously.
\end{proof}

This analyticity implies an exponential rate of convergence for the Riemann sum towards the corresponding integral. More specifically, applying Lemma~\ref{lem:RiemannAnalytic} to the functions $K_{\lambda, \bq}^{f,g}$, and using~\eqref{eq:boundFq} leads to
\[
	\exists C \in \R^+, \quad \exists \alpha >0, \quad \forall L \ge L^*, \quad \forall \bq \in \BZ \setminus \{ \bnull \}, 
	\left\| \left( \widetilde{\cL^L_\bq} - \cL_\bq \right)  \right\|_{\cB(L^2_\per(\WS))} \le C \dfrac{1}{| \bq |^2} \re^{-\ri \alpha L}.
\]
If $\bQ \in \Lambda_L$ and $\bQ \neq 0$, then it holds $| \bQ |^{-2} \le L^2$, so that there exist $C' \in \R^+$  and $\alpha'>0$ such that
\[
	\forall L \ge L^*, \quad \forall \bQ \in \Lambda_L \setminus \{ \bnull \} , 
	\left\|   \left( \widetilde{\cL^L_\bQ} - \cL_\bQ \right)  \right\|_{\cS(L^2_\per)} \le C L^2 \re^{-\ri \alpha L} \le C' \re^{- \ri \alpha'L},
\]
which is the desire result.

\subsection{Proof of Lemma~\ref{lem:analLq} }
\label{sec:proof:analLq}
%
%
%
%

We choose $A > 0$ as in Lemma~\ref{lem:bounds_Bq}, and $r > 0$. Let us first prove that for $\bz \in \Omega_r + \ri [-A, A]^3$, the operator $\cL_\bz$ is indeed bounded. For $f,g, \in L^2_\per$, we denote by $F_\bz = (\sqrt{v_c})_\bz f$ and $G_\bq = (\sqrt{v_c})_\bz g$. As in~\eqref{eq:boundFq}, we deduce that $F_\bz$ and $G_\bz$ are in $L^2_\per$, and that there exists $C_r \in \R^+$ independent of $f$ and $g$ such that
\begin{equation} \label{eq:controlFq}
	\forall \bz \in \Omega_r + \ri [-A,A]^3, \quad \left\| F_\bz \right\|_{L^2_{\per}} \le  C_r \left\| f \right\|_{L^2_{\per}}.
\end{equation}
In particular,
\begin{align*}
	\left| \bra f | \cL_\bz g \ket  \right| &
	= \left|
		\dfrac{-1}{2 \ri \pi} \fint_\BZ \oint_\sC \Tr_{L^2_\per(\WS)} \left( 
		\dfrac{1}{\lambda - H_{\bq'}} G_\bz \dfrac{1}{\lambda - H_{\bq' - \bz}} \overline{F_\bz}
	\right) \rd \lambda \rd \bq' \right| \\
	& \le \dfrac{1}{2\pi} \fint_\BZ \oint_{\sC} \left| \Tr_{L^2_\per(\WS)} \left( 
		\dfrac{1}{\lambda - H_{\bq'}} G_\bz \dfrac{1}{\lambda - H_{\bq' - \bz}} \overline{F_\bz}
	\right) \right| \rd \lambda \rd \bq'.
\end{align*}
On the other hand, it holds that
\begin{align*}
	\left| \Tr_{L^2_\per(\WS)} \left( 
		\dfrac{1}{\lambda - H_{\bq'}} G_\bz \dfrac{1}{\lambda - H_{\bq' - \bz}} \overline{F_\bz}
	\right) \right| &	= 
	\left| \Tr_{L^2_\per(\WS)} \left( 
		B_2(\lambda, \bq') \dfrac{1}{1 - \Delta^1} G_\bz  B_2(\lambda, \bq' - \bz)  \dfrac{1}{1 - \Delta^1}  \overline{F_\bz} \right) \right| \\
	& \quad \le C \left\| \dfrac{1}{1 - \Delta^1} G_\bz \right\|_{\fS_2} \left\| \dfrac{1}{1 - \Delta^1} \overline{F_\bz} \right\|_{\fS_2},
\end{align*}
where we used Lemma~\ref{lem:bounds_Bq} for the last inequality. Together with the periodic Kato-Seiler-Simon inequality (see Corollary~\ref{cor:periodicKSS}) and~\eqref{eq:controlFq}, we easily deduce that $\cL_\bz$ is bounded on $L^2_\per(\WS)$. We also deduce that there exists $C \in \R^+$ such that
\[
	\sup_{\bz \in \Omega_r + \ri [-A, A]^3} \left\| \cL_\bz \right\|_{\cB(L^2_\per)} \le C.
\]

Let us prove the analyticity. It is enough to prove that for all $\bz \in \Omega_r + \ri [-A,A]^3$, the operator $\partial_{z_j} \cL_\bz$ is bounded. We obtain
\begin{align}
	\bra f | \partial_{z_j} \cL_\bz g \ket = & \dfrac{-1}{2 \ri \pi} \fint_\BZ \oint_\sC \Tr_{L^2_\per(\WS)} \left( 
		\dfrac{1}{\lambda - H_{\bq'}} \left[ \partial_{z_j}  G_\bz\right] \dfrac{1}{\lambda - H_{\bq' - \bz}} \overline{F_\bz}
	\right) \rd \lambda \rd \bq' \label{eq:partialL1} \\
		&  +  \dfrac{-1}{2 \ri \pi} \fint_\BZ \oint_\sC \Tr_{L^2_\per(\WS)} \left( 
		\dfrac{1}{\lambda - H_{\bq'}}  G_\bz \dfrac{1}{\lambda - H_{\bq' - \bz}}  \left[ \partial_{z_j}  \overline{F_\bz} \right]
	\right) \rd \lambda \rd \bq' \label{eq:partialL2} \\
		& + \dfrac{-1}{2 \ri \pi} \fint_\BZ \oint_\sC \Tr_{L^2_\per(\WS)} \left( 
		\dfrac{1}{\lambda - H_{\bq'}}  G_\bz \left[ \partial_{z_j} \dfrac{1}{\lambda - H_{\bq' - \bz}} \right] \overline{F_\bz}
	\right) \rd \lambda \rd \bq'. \label{eq:partialL3}
\end{align}
Since it holds
\[
	\partial_{z_j} (\sqrt{v_c})_\bz =  \sqrt{4 \pi} \sum_{\bk \in \RLat} | e_\bk \ket  \dfrac{ z_j + k_j} {|\bz + \bk |^3}   \bra e_\bk |,
\]
we obtain as in~\eqref{eq:boundFq} that there exists $C_r \in \R^+$ such that
\begin{equation*}
	\forall \bz \in \Omega_r + \ri [-A,A]^3, \quad \left\| \partial_{z_j}  F_\bz \right\|_{L^2_{\per}} \le  C_r \left\| f \right\|_{L^2_{\per}}.
\end{equation*}
Using estimates similar than previously, we deduce that the operators defined in~\eqref{eq:partialL1} and~\eqref{eq:partialL2} are bounded. Finally, we notice that
\[
	 \partial_{z_j} \dfrac{1}{\lambda - H_{\bq' - \bz}} = \left[ \dfrac{1}{\lambda - H_{\bq' - \bz}} \left(P_j^1 - z_j\right) \right] \dfrac{1}{\lambda - H_{\bq' - \bz}},
\]
where the operator under brackets is bounded. Altogether, we deduce that $\partial_{z_j} \cL_\bz$ is bounded, so that $\bz \mapsto \cL_\bz$ is analytic.

\subsection{Proof of Lemma~\ref{lem:propcL}}
\label{sec:proof:propcL}

The first point is proved in a similar way than the proof of Lemma~\ref{lem:analLq}. Let us prove the second point. It holds
\[
	\Lambda_\bq = \bra e_\bnull | \cL_\bq e_\bnull \ket = 
	\dfrac{-1}{2 \ri \pi} \dfrac{4\pi}{| \bq |^2} \dfrac{1}{| \WS |} \fint_\BZ \oint_\sC \Tr_{L^2_\per(\WS)} \left( 
		\dfrac{1}{\lambda - H_{\bq'}}  \dfrac{1}{\lambda - H_{\bq' - \bq}}
	\right) \rd \lambda \rd \bq' .
\]
Using the spectral representation of $H_\bq$ in~\eqref{eq:spectralHq}, performing the $\oint_\sC$ integration using the Cauchy residual theorem and using~\eqref{eq:symmetries} leads, after some manipulations, to
\begin{equation} \label{eq:Lqkk'}
	\Lambda_\bq 
	 =
	 \dfrac{ 8 \pi}{| \WS | \ | \bq |^2} \sum_{n \le N < m} \fint_\BZ \dfrac{ \left| \left\bra u_{m, \bq' - \bq} | u_{n, \bq'} \right\ket \right|^2 }{| \varepsilon_{m, \bq - \bq'} - \varepsilon_{n, \bq'}|} \rd \bq'.
\end{equation}
Let us develop the numerator of~\eqref{eq:Lqkk'}. We use the identity
\[
	\varepsilon_{n, \bq'} \left\bra u_{m, \bq' - \bq} | u_{n, \bq'}  \right\ket
	= \left\bra u_{m, \bq' - \bq} | H_{\bq'} u_{n, \bq'} \right\ket
	= \left\bra H_{\bq'}  u_{m, \bq' - \bq} |  u_{n, \bq'} \right\ket.
\]
Together with the fact that $H_{\bq'} = H_{\bq' - \bq} - 2 \ri \bq \cdot \nabla + (| \bq' |^2 - |\bq - \bq'|^2)/2$, this leads to
\begin{equation} \label{eq:numerator}
 	\left\bra u_{m, \bq' - \bq} | u_{n, \bq'} \right\ket
 	= \dfrac{  \bq \cdot \left\bra u_{m, \bq -\bq'} | \left( - \ri \nabla^1 \right)  u_{n, \bq'} \right\ket}{\varepsilon_{n, {\bq'}} - \varepsilon_{m, {\bq - \bq'}} + \frac{| \bq |^2}{2} - \bq \cdot \bq'}.
\end{equation}
We now choose $r_1 > 0$ such that
\begin{equation} \label{eq:why_r1}
	\forall \bq \in \cB(\bnull, r_1), \quad \forall \bq' \in \BZ, \quad 
	\left|  \frac{| \bq |^2}{2} - \bq \cdot \bq'  \right| < \frac{g}{2},
\end{equation}
where we recall that $g > 0$ is the gap of the rHF system. Together with~\eqref{eq:gap}, this implies that the denominator of~\eqref{eq:numerator} is negative and away from $\bnull$ for all $\bq, \bq' \in \BZ$. As a result, we deduce from~\eqref{eq:Lqkk'} that $\Lambda_\bq = | \bq |^{-2}\bq^T M_1(\bq) \bq$, with
\[
	M_1(\bq)
	:= 
	\dfrac{ 8 \pi}{| \WS |} \sum_{n \le N < m} \fint_\BZ \dfrac{ \bra u_{n, \bq'} | (- \ri \nabla^1) u_{m, \bq' - \bq} \ket \bra u_{m, \bq' - \bq} | (- \ri \nabla^{1,T}) u_{n, \bq'} \ket  }{\left( \varepsilon_{m, {\bq' - \bq}} - \varepsilon_{n, {\bq'}} - \frac{| \bq |^2}{2} + \bq \cdot \bq' \right)^2 | \varepsilon_{m, {\bq' - \bq}} - \varepsilon_{n, {\bq'}} |} \rd \bq'
\]
which is~\eqref{eq:def:Aq}. We easily deduce from this formula that  $M_1(\bq)$ is a $3 \times 3$ hermitian matrix, and that $M_1(-\bq) = M_1(\bq)$. 

\medskip

Note that we cannot prove directly that $M_1(\bq)$ admits an analytic extension, since the maps $\bq \mapsto \varepsilon_{n, \bq}$ and  $\bq \mapsto u_{n, \bq}$ are not smooth in general. Let us provide an alternative formula for $M_1$. The idea is to \textit{undo} the Cauchy integration. More specifically, from~\eqref{eq:why_r1} and~\eqref{eq:gap}, we deduce that
\[
	\forall \bq \in \cB(\bnull, r_1), \quad \forall \bq' \in \BZ, \quad  \varepsilon_{m, \bq' - \bq} - \frac{| \bq |^2}{2} + \bq \cdot \bq' > \varepsilon_F > \varepsilon_{n, \bq'},
\]
so that, for all $n \le N < m$,
\[
	\dfrac{1}{\left( \varepsilon_{m, \bq' - \bq} - \varepsilon_{n, \bq'} + \bq \cdot \bq' - \frac{| \bq |^2}{2} \right)} = 
	\dfrac{-1}{2 \ri \pi} \oint_\sC \dfrac{1}{\left( \lambda_1 - \varepsilon_{m, \bq' - \bq} - \bq \cdot \bq' + \frac{| \bq |^2}{2} \right)(\lambda_1 - \varepsilon_{n, \bq'})}  \rd \lambda_1.
\]
We plug this expression in~\eqref{eq:def:Aq}, and get
\begin{align*}
	& M_1(\bq) = \dfrac{8 \pi}{|\WS|} \sum_{n \le N < m} \oint_{\sC} \oint_{\sC} \oint_{\sC} \fint_{\BZ} \rd \bq' \rd \lambda \rd \lambda_1 \rd \lambda_2
	 \bra u_{n, \bq'} | (- \ri \nabla^1) u_{m, \bq' - \bq} \ket \bra u_{m, \bq' - \bq} | (- \ri \nabla^{1,T}) u_{n, \bq'} \ket \times \\
	 & \quad \times 
	 \left[ \left( \lambda - \varepsilon_{m, \bq' - \bq} \right)   \left(  \lambda - \varepsilon_{n, \bq'} \right)
	\prod_{j\in \{1,2 \}}  \left(  \lambda_j - \varepsilon_{m, \bq' - \bq} - \bq \cdot \bq' + \frac{| \bq |^2}{2} \right)  \left( \lambda_j - \varepsilon_{n, \bq'}\right)  \right]^{-1}.
\end{align*}
Using the spectral representation~\eqref{eq:spectralHq}, this is also
\begin{align*}
	& M_1(\bq) = \dfrac{4 \pi}{|\WS|} \oint_{\sC} \oint_{\sC} \oint_{\sC} \fint_{\Gamma^\ast}
	 \Tr_{L^2_\per(\WS)} \left(  (-\ri \nabla^1)  \dfrac{1}{ (\lambda - H_{\bq'})  \left( \lambda_1 - H_{\bq' }\right) \left( \lambda_2 - H_{\bq' } \right)} (-\ri \nabla^{1,T})   \right) \\
	& \qquad \left.      \dfrac{ 1 }{(\lambda - H_{\bq' - \bq}) \left(\lambda_1 - H_{\bq' - \bq} - \bq \cdot \bq' + \frac{| \bq |^2}{2} \right) \left(\lambda_2 - H_{\bq' - \bq} - \bq \cdot \bq' + \frac{| \bq |^2}{2} \right)} 
		 \right. 
	   \rd \bq' \rd \lambda \rd \lambda_1 \rd \lambda_2.
\end{align*}
From this last expression, and following the same arguments as in the proof of Lemma~\ref{lem:analLq}, we see that $M_1(\cdot)$ admits an analytical extension on $\cB(\bnull, r_1) + \ri [-A, A]^3$. 

\medskip

We finally prove the third point of the lemma. For $f \in L^2_{0,\per}$, we denote by $F_\bq := (\sqrt{w_{c}})_\bq f$. 
Using similar techniques than before, we obtain
\begin{align*}
	\bra l_\bq | f \ket & = \bra e_\bnull | \cL_\bq f \ket =
	\dfrac{-1}{2 \ri \pi} \dfrac{(4\pi)^{1/2}}{| \bq |} \dfrac{1}{| \WS |^{1/2}} \fint_\BZ \oint_\sC \Tr_{L^2_\per(\WS)} \left( 
		\dfrac{1}{\lambda - H_{\bq'}} F_\bq \dfrac{1}{\lambda - H_{\bq' - \bq}}
	\right) \rd \lambda \rd \bq'  \\
	&
	= \dfrac{2(4\pi)^{1/2}}{| \bq |} \dfrac{1}{| \WS |^{1/2}} \sum_{n \le N < m}\fint_\BZ 
	\dfrac{\bra u_{n, \bq'}  | u_{m, \bq' - \bq} \ket \bra u_{m, \bq' - \bq} | F_\bq u_{n, \bq'} \ket  }{| \varepsilon_{m, \bq' - \bq} - \varepsilon_{n, \bq'} |} \rd \bq' \\
	&
	= -2(4\pi)^{1/2} \dfrac{1}{| \WS |^{1/2}} \dfrac{\bq^T}{| \bq |} \sum_{n \le N < m}\fint_\BZ 
	\dfrac{ \bra u_{n, \bq'}  | (-\ri \nabla^1) u_{m, \bq' - \bq} \ket \bra u_{m, \bq' - \bq} | F_\bq u_{n, \bq'} \ket }{| \varepsilon_{m, \bq' - \bq} - \varepsilon_{n, \bq'} | \left( \varepsilon_{m, {\bq' - \bq}} - \varepsilon_{n, {\bq'}} - \frac{| \bq |^2}{2} + \bq \cdot \bq' \right)} \rd \bq' \\
	&
	= \dfrac{\bq^T}{| \bq |} \bra \bb(\bq) | f \ket,
\end{align*}
where $\bb(\bq)$ was defined in~\eqref{eq:def:bb}. The fact that $\bb$ admits an analytical extension on $\Omega_r + \ri [-A,A^3]$ is proved similarly as in the proof of the analyticity of $A(\bz)$.

\subsection{Proof of Lemma~\ref{lem:A0}}
\label{sec:proof:A0}

Let us first prove~\eqref{eq:def:A0}. From~\eqref{eq:def:Aq}, it holds that
\begin{align} \label{eq:Anull}
	M_1(\bnull) & = \dfrac{8 \pi}{| \WS |} \sum_{n \le N < m} \fint_{\bq' \in \BZ} \dfrac{\bra  u_{n, \bq'} | (- \ri \nabla^1) u_{m, \bq'} \ket \bra u_{m, \bq'} | (-\ri \nabla^{1,T}) u_{n, \bq'} \ket}{|  \varepsilon_{m, \bq'} - \varepsilon_{n,\bq'}|^3} \rd \bq' \nonumber \\
		& =  \dfrac{8 \pi}{| \WS |} \sum_{n \le N < m} \fint_{\bq' \in \BZ} \dfrac{\bra  u_{n, \bq'} | ( - \ri \nabla^1 + \bq') | u_{m, \bq'} \ket \bra u_{m, \bq'} | ( - \ri \nabla^1 + \bq')^T | u_{n, \bq'} \ket}{|  \varepsilon_{m, \bq'} - \varepsilon_{n,\bq'}|^3} \rd \bq'.
\end{align}
For all $a, b \in \R$, with $a \neq b$, it holds that
\[
	\dfrac{1}{(\lambda - a)^2 (\lambda - b)^2 } = \dfrac{1}{(a - b)^2} \left( \dfrac{1}{(\lambda-a)^2} + \dfrac{1}{(\lambda - b)^2} + \dfrac{2}{(a-b)} \left(\dfrac{1}{\lambda - b} - \dfrac{1}{\lambda - a} \right) \right).
\]
We integrate this equality over $\sC$, and use the Cauchy residual formula to get (we suppose for simplicity $a \notin \sC$ and $b \notin \sC$)
\begin{equation} \label{eq:nonTrivialCauchy}
	\dfrac{1}{2 \ri \pi} \oint_{\sC} \dfrac{1}{(\lambda - a)^2 (\lambda - b)^2} \rd \lambda = 	\left\{ \begin{array}{lll}
			0 & \text{if} & a \ \text{and} \ b \ < \varepsilon_F \quad \text{or} \quad a \ \text{and} \ b \ge \varepsilon_F \\
			\dfrac{2}{|a - b|^3} & \text{if} & b < \varepsilon_F < a \quad \text{or} \quad a < \varepsilon_F < b,
		\end{array} \right.
\end{equation}
and~\eqref{eq:nonTrivialCauchy} is also true for $a = b$. Hence,~\eqref{eq:Anull} can be rewritten as
\begin{align*}
	M_1(\bnull) & = \dfrac{4 \pi}{2 \ri \pi | \WS |} \sum_{n,m \in \N} \oint_{\sC} \fint_{\bq' \in \BZ} \dfrac{\bra  u_{n, \bq'} | ( - \ri \nabla^1 + \bq') | u_{m, \bq'} \ket \bra u_{m, \bq'} | ( - \ri \nabla^1 + \bq')^T | u_{n, \bq'} \ket}{(\lambda - \varepsilon_{n, \bq'})^2 (\lambda - \varepsilon_{m, \bq'})^2} \rd \bq' \rd \lambda \\
		& = \dfrac{4 \pi}{2 \ri \pi | \WS |} \oint_{\sC} \fint_{\bq' \in \BZ} \Tr_{L^2_\per(\WS)} \left( \dfrac{1}{( \lambda - H_{\bq'})^2} (- \ri \nabla + \bq') \dfrac{1}{( \lambda - H_{\bq'})^2}  (- \ri \nabla + \bq')^T \right) \rd \bq' \rd \lambda \\
		& = \dfrac{4 \pi}{2 \ri \pi | \WS |} \oint_{\sC} \VTr \left( \dfrac{1}{(\lambda - H_0)^2} ( - \ri \nabla)  \dfrac{1}{(\lambda - H_0)^2} ( - \ri \nabla^T)  \right) \rd \lambda,
\end{align*}
where we used~\eqref{eq:VTr} in the last equality. This proves~\eqref{eq:def:A0}. 

\medskip

We now suppose that the crystal is isotropic cubic, and prove that $M_1(\bnull)$ is proportional to the identity matrix. An easy calculation shows that if a matrix $P$ satisfies
\begin{equation}\label{eq:relation_iso}
P=S_1^TPS_1\quad \text{and} \quad P=S_2^T P S_2,
\end{equation}
where $S_1$ and $S_2$ were introduced in Definition~\ref{def:isotropic_cubic}, then $P$ is proportional to the identity matrix. Let us show that $M_1(\bnull)$ satisfies~\eqref{eq:relation_iso}. 
Let $S$ be either $S_1$ or $S_2$. 
We introduce the operator 
\[
	\begin{array}{lrll}
		\cT : & L^2_\per(\WS) & \to & L^2_\per(\WS) \\
			& f & \mapsto & f(S\bx).
	\end{array}
\]
It holds that $ \nabla \cT= S^T\cT\nabla $ and that $H_0 \cT = \cT H_0$. Together with~\eqref{eq:def:A0}, we arrive at
\begin{align*}
	 M_1(\bnull) & = \dfrac{4\pi}{2 \ri \pi | \WS |} \oint_{\sC} \VTr \left[ \left( \cT^{-1} \dfrac{1}{(\lambda - H_0)^2} \cT \right) \left( \cT^{-1} \nabla \cT \right) \left( \cT^{-1} \dfrac{1}{(\lambda - H_0)^2}  \cT \right) \left( \cT^{-1} \nabla^T \cT \right) \right] \rd \lambda \\
	 & = \dfrac{2 \pi}{2 \ri \pi | \WS |} \oint_{\sC} \VTr \left[ \dfrac{1}{(\lambda - H_0)^2} (S^T \nabla)  \dfrac{1}{(\lambda - H_0)^2}  (\nabla^T S) \right] \rd \lambda \\
	 & = S^T M_1(\bnull) S.
\end{align*}
The result follows.

\subsection{Proof of Lemma~\ref{lem:analDq}}
\label{sec:proof:analDq}

The fact that $(1 + \cL_\bq)$ is invertible comes from the fact that $\cL_\bq$ is a bounded positive self-adjoint operator. For $\bz = \bq + \ri \by$, we have
\[
	\left( 1 + \cL_{\bq + \ri \by} \right) = (1 + \cL_\bq) \left( 1 +  (1 + \cL_\bq)^{-1} \left( \cL_{\bq + \ri \by}  - \cL_\bq \right) \right) )
\]
According to Lemma~\ref{lem:analLq}, the map $\bz \in \cL_\bz$ is analytic on $\Omega_r + \ri [-A, A]^3$. As a result, there exists $A' > 0$ such that
\[
	\forall \bq \in 2 \BZ, \quad \forall \by \in [-A', A']^3, \quad \left\| \cL_{\bq + \ri \by} - \cL_\bq \right\|_{\cB(L^2_\per)} \le \frac12.
\]
As a result, it holds $\left\|  (1 + \cL_\bq)^{-1} \left( \cL_{\bq + \ri \by}  - \cL_\bq \right) \right\|_{\cB(L^2_\per)} \le 1/2 < 1$, so that $\left( 1 + \cL_{\bq + \ri \by} \right)$ is invertible on $2 \BZ + \ri [-A', A']^3$, with
\[
	\forall \bq \in 2 \BZ, \quad \forall \by \in [-A', A']^3, \quad \left\| \left( 1 + \cL_{\bq + \ri \by} \right) \right\|_{\cB(L^2_\per)} \le 2.
\]
Using the covariant property~\eqref{eq:rotation}, we deduce that $\left( 1 + \cL_{\bq + \ri \by} \right)$ is invertible on $\R^3 + \ri [-A', A']^3$ with similar bounds. Finally, since the map $\bz \mapsto ( 1 + \cL_\bz)$ is analytic from $\Omega_r + \ri [-A', A']^{-1}$, then so is the map $\bz \mapsto ( 1 + \cL_\bz)^{-1}$ (see for instance~\cite[Chapter 7, §1.1]{Kato2012}).


\subsection{Proof of Lemma~\ref{lem:eqFMinusFL}}
\label{sec:proof:FMinusFL}

For $\bQ \neq \bnull$, it holds that $(\sqrt{v_c})_\bQ = (\sqrt{v_c^L})_\bQ$. Also, since the support of $\nu$ is in $\WS_{L^\supp}$, we get from~\eqref{eq:def:cZ} and~\eqref{eq:def:cZL} that, for $L \ge L^* (\ge L^\supp)$,
\[
	\nu_{L,\bQ} (\bx) = \sum_{\bR \in L \Lat} \re^{-\ri \bQ \cdot (\bx + \bR)} \nu(\bx + \bR) 
	=\re^{-\ri \bQ \cdot \bx } \nu(\bx)
	= \nu_\bQ(\bx).
\]
As a result, for $\bQ \in \Lambda_L \setminus \{ \bnull \}$, it holds that
\begin{align}
	 \left| F_\nu(\bQ) - F_\nu^L(\bQ) \right| & = \left| \left\bra \left[ (1 + \cL_\bQ)^{-1} - (1 + \widetilde{\cL^L_\bQ})^{-1} \right]  (\sqrt{v_c})_\bQ \nu_\bQ, (\sqrt{v_c})_\bQ \nu_\bQ  \right\ket \right| \nonumber \\
	 & \le \left\| (1 + \cL_\bQ)^{-1} - (1 +  \widetilde{\cL^L_\bQ})^{-1} \right\|_{\cB(L^2_\per)}
	  \left\| (\sqrt{v_c})_\bQ \nu_\bQ \right\|_{L^2_\per}^2. \label{eq:controlF-FL}
\end{align}
From Lemma~\ref{lem:LQ_LLQ} and the fact that $\cL_\bQ$ and $\cL_\bQ^L$ are positive bounded operators, we obtain that there exists $C \in \R^+$ and $\alpha >0$ such that, for all $L \ge L^*$ and all $\bQ \in \Lambda_L \setminus \{ \bnull \}$,
\begin{align*}
	\left\| (1 + \cL_\bQ)^{-1} - (1 +  \widetilde{\cL^L_\bQ})^{-1} \right\|_{\cB(L^2_\per)}
	& = \left\| (1 + \cL_\bQ)^{-1} \left(  \widetilde{\cL^L_\bQ} - \cL_\bQ \right)  (1 +  \widetilde{\cL^L_\bQ})^{-1} \right\|_{\cB(L^2_\per)} \\
	& \le  \left\| \left(  \widetilde{\cL^L_\bQ} - \cL_\bQ \right)  \right\|_{\cB(L^2_\per)} \le C \re^{-\ri \alpha L}.
\end{align*}
Together with~\eqref{eq:controlF-FL}, we obtain
\begin{align*}
	\left| \dfrac{1}{L^3} \sum_{\bQ \in \Lambda_L \setminus \{ \bnull \}} \left( F_\nu(\bQ) - F_\nu^L(\bQ) \right) \right| 
		& \le C \re^{-\ri \alpha L} \dfrac{1}{L^3} \sum_{\bQ \in \Lambda_L \setminus \{ \bnull \}} \left\| (\sqrt{v_c})_\bQ \nu_\bQ \right\|_{L^2_\per}^2 
\end{align*}
As in~\eqref{eq:boundFq}, there exists $C \in \R^+$ such that, for all $L \ge L^*$ , all $\bQ \in \Lambda_L \setminus \{ \bnull \}$, and all $\nu \in \cB(\eta)$, it holds that
\[
	\left\| (\sqrt{v_c})_\bQ \nu_\bQ \right\|_{L^2_\per}^2 \le \dfrac{C}{|\bQ|^2} \left\|  \nu_\bQ \right\|_{L^2_\per} \le C L^2 \left\|  \nu_\bQ \right\|_{L^2_\per},
\]
so that
\[
	\dfrac{1}{L^3} \sum_{\bQ \in \Lambda_L \setminus \{ \bnull \}} \left\| (\sqrt{v_c})_\bQ \nu_\bQ \right\|_{L^2_\per}^2 
	\le C \dfrac{L^2}{L^3}\sum_{\bQ \in \Lambda_L \setminus \{ \bnull \}} \left\| \nu_\bQ \right\|_{L^2_\per}^2 
	\le C L^2 \left\| \nu_L \right\|_{L^2_\per(\WS_L)}^2  = C L^2 \left\| \nu \right\|_{L^2(\R^3)}^2.
\]
The result easily follows.

\subsection{Proof of Lemma~\ref{lem:F1}}
\label{sec:proof:F1}

Since $\nu$ is compactly supported, we deduce that the map $\bq \mapsto \nu_\bq$ is analytic. Together with Lemma~\ref{lem:analDq} and the definition of $\Psi_\per$ in~\eqref{eq:Psi}, we get that $F_{\nu,1} \in C^{\infty}(\R^3)$ and $F_{\nu,1}(\bnull) = 0$. In particular, the missing $\bQ = \bnull$ term in the Riemann sum of~\eqref{eq:F1} can be restored. Also, from the covariant identity~\eqref{eq:rotation} and the periodicity of $\Psi_\per$, we obtain that $F_{\nu,1}$ is $\RLat$-periodic.

\medskip

Finally, we notice that for all $p_1, p_2, p_3 \in \N$, it holds
\[
	\left\| \partial_{q_1}^{p_1} \partial_{q_2}^{p_2} \partial_{q_3}^{p_3} \nu_\bq \right\|_{L^2_\per}
	= \left\| x_1^{p_1} x_2^{p_2} x_3^{p_3 } \nu(\bx) \right\|_{L^2(\R^3)} \le \left( |\WS|_\infty L^{\supp}\right)^{p_1 + p_2 + p_3} \left\| \nu \right\|_{L^2(\R^3)}^2,
\]
where we used the fact that $\nu$ is compactly supported in $L^\gap \WS$, and where we denoted by $| \WS |_{\infty} := \sup \{ | \bx |_\infty, \ \bx \in \WS \}$. The result then follows from the theory of the convergence of Riemann sums for smooth periodic functions (see Lemma~\ref{lem:RiemannR3} below).


\section*{Appendices}
\begin{appendices}
\section{The periodic Kato-Seiler-Simon inequality}

The Kato-Seiler-Simon inequality in $L^p(\R^d)$ spaces is well-understood~\cite{Simon2005}. In our case, we need a periodic version of this inequality. In particular, we prove that the constant of continuity in this case depends on the size of the supercell. We introduce
\[
	\ell^p(L^{-1}\RLat) := \left\{ \left( g^L_\bk \right)_{\bk \in L^{-1}\RLat}, \quad \left\| g^L \right\| := 
	\left( \sum_{\bk \in L^{-1}\RLat} \left| g^L_\bk \right|^p \right)^{1/p} < \infty \right\}.
\]
For $\left( g^L_\bk \right)_{\bk \in L^{-1}\RLat} \in \ell^p(L^{-1}\RLat)$, $p\ge 2$, we denote by $g^L(-\ri \nabla)$ the operator on $L^2_\per(\WS_L)$ defined by
\[
	\forall h \in L^2_\per(\WS_L), \quad \forall \bk \in L^{-1} \RLat, \quad c_\bk^L \left( g^L(-\ri \nabla) h \right) = g^L_\bk c_\bk^L (h).
\]
\begin{lemma} [Periodic Kato-Seiler-Simon inequality]
\label{lem:periodicKSS}
	Let $p \ge 2$, $f^L \in L^p_\per(\WS_L)$ and $\left( g^L_\bk \right)_{\bk \in L^{-1}\RLat} \in \ell^p(L^{-1}\RLat)$. Then the operator $f^L(\bx) g^L(- \ri \nabla)$ is in the Schatten class $\fS_p(L^2_\per(\WS_L))$, and it holds
	\begin{equation} \label{eq:periodicKSS}
		\left\| f^L(\bx) g^L(- \ri \nabla) \right\|_{\fS_p(L^2_\per(\WS_L))} \le \dfrac{1}{|\WS_L|^{1/p}} \| f^L \|_{L^p_\per(\WS_L)} \| g^L \|_{\ell^p(L^{-1}\RLat)}.
	\end{equation}
\end{lemma}

\begin{proof}[Proof of Lemma~\ref{lem:periodicKSS}]
It is enough to prove the lemma for $p=2$ and $p=\infty$, the other cases being deduced from classical interpolation arguments. We start with $p=2$. The kernel of the operator $f^L(\bx) g^L(- \ri \nabla)$ is
\[
	K(\bx, \by) := \left[ f^L(\bx) g^L(- \ri \nabla) \right](\bx, \by) = \dfrac{1}{|\WS_L|} f^L(\bx) \sum_{\bk \in L^{-1}\RLat} g_\bk^L \re^{\ri \bk \cdot(\bx - \by)}.
\]
A straightforward calculation leads to
\begin{align*}
	\left\| f^L(\bx) g^L(- \ri \nabla) \right\|_{\fS_2(L^2_\per(\WS_L))}^2 = \int_{(\WS_L)^2 } \left| K(\bx, \by)\right|^2 = \dfrac{1}{|\WS_L|}  \| f^L \|_{L^2_\per(\WS_L)}^2 \| g^L \|_{\ell^2(L^{-1}\RLat)}^2,
\end{align*}
from which we deduce the result for $p=2$. For $p=\infty$, we notice that $f^L(\bx)$ is an operator of norm $\| f \|_{L^\infty}$, and that $g^L(- \ri \nabla)$ is an operator of norm $\| g^L \|_{\ell^\infty(L^{-1}\RLat)}$. The result follows.

\end{proof}

In the case where $g^L$ represents $(1 - \Delta_L)^{-1}$, or equivalently,
\[
	\forall \bk \in L^{-1}\RLat, \quad g^L_\bk = \dfrac{1}{1 + |\bk|^2},
\]
it holds that $\left( g^L_\bk \right)_{\bk \in L^{-1}\RLat} \in \ell^p(L^{-1}\RLat)$ for $p > 3/2$. Thanks to Lemma~\ref{lem:RiemannR3} applied to the function $(1 + | \cdot |^2)^{-p} \in W^{s,1}(\R^3)$, for all $s > 3/2$,we obtain
\begin{equation} \label{eq:normGL}
	\| g^L \|_{\ell^p(L^{-1}\RLat)}^p =
	\sum_{\bk \in \RLat} \left( \dfrac{1}{1 + |\bk/L|^2} \right)^p =
	 \dfrac{L^3}{| \BZ |} \left( \int_{\R_3} \dfrac{1}{(1 + | \bq |^2)^p} \right) + o(L^{3}).
\end{equation}

We deduce the following useful corollary.

\begin{corollary}
\label{cor:periodicKSS}
For all $p \ge 2$, there exists $C_p$ such that for all $L\in \N^*$ and all $V \in L^p_\per(\WS_L)$, the operator $(1 - \Delta^L)^{-1}V$ is in $\fS_p(L^2_\per(\WS_L))$ with
\[
	\left\| (1 - \Delta^L)^{-1}V \right\|_{\fS_p(L^2_\per(\WS_L))} \le C_p \| V \|_{L^p_\per(\WS_L)}.
\]
\end{corollary}

\section{Convergence of Riemann sums}
We recall in this appendix some classical results about the convergence of Riemann sums. We also prove some non standard results when the integrand is a singular function. We introduce the standard multi-indices notations $\alpha = (\alpha_1, \alpha_2, \alpha_3)$, $ | \alpha | = \sum_{i=1}^3 \alpha_i$, $\alpha ! = \alpha_1 ! \alpha_2 ! \alpha_3 !$, $\bq^\alpha = q_1^{\alpha_1} q_2^{\alpha_2} q_3^{\alpha_3}$, and $\partial_\bq^\alpha = \partial_{q_1}^{\alpha_1} \partial_{q_2}^{\alpha_2} \partial_{q_3}^{\alpha_3}$.
 
 \medskip
 
The following lemma was proved in \text{e.g.}~\cite{GL2015}.

\begin{lemma} \label{lem:RiemannAnalytic} $\,$
	For any $A > 0$, there exist $C \in \R^+$ and $\alpha > 0$ such that for all functions $f: \R^3 + \ri [-A,A]^3 \to \C$ that are analytic on $\R^3 + \ri [-A,A]^3$ and that satisfies $f(\bz + \bk) = f(\bz)$ for all $\bz \in \R^3 + \ri [-A,A]^3$ and all $\bk \in \RLat$, it holds that
	\[
		\left| \fint_{\BZ} f(\bq) \rd \bq - \dfrac{1}{ L^3} \sum_{\bQ \in \Lambda_L} f \left( \bQ \right)  \right| \le C \re^{- \alpha L}.
	\]
\end{lemma}

%

We will also need results on the whole space $\R^3$. We introduce the usual Sobolev space
\[
	W^{p,s}(\R^3) \left\{ f \in L^s(\R^3), \quad \left\| f \right\|_{W^{p,s}(\R^3)} := \sum_{| \alpha | \le p } \left\| \partial_\bq^\alpha f \right\|_{L^s(\R^3)} < \infty \right\}.
\]
We also introduce the following notation for clarity.
\begin{equation} \label{eq:def:Ix}
	I(f) := \dfrac{1}{| \BZ |}  \int_{\R^3} f(\bq) \rd \bq 
	\quad \text{and} \quad
	I_\lambda(f) := \dfrac{1}{\lambda^3}  \sum_{\bk \in \RLat} f \left( \dfrac{\bk}{\lambda} \right)
	\quad \text{for} \quad
	\lambda \in (0, \infty).
\end{equation}
\begin{lemma} [Convergence of Riemann sum in the whole space] \label{lem:RiemannR3}
	For all $p > 3/2$, there exists $C_p \in \R^+$ such that, for all $f \in W^{p, 1}(\R^3)$, it holds that
	\begin{equation} \label{eq:RiemannL1}
	\forall L \in \N^*,\quad	\left| I (f) - I_L(f) \right| \le C_p \dfrac{\left\| f \right\|_{W^{p,1}(\R^3)}}{L^p}.
	\end{equation}
\end{lemma}

\begin{remark}
	The constant $C_p$ may depend on the lattice $\RLat$, but the speed of convergence is independent of the choice of the lattice $\RLat$.
\end{remark}


\begin{proof}[Proof of Lemma~\ref{lem:RiemannR3}]
	Let $\widehat{f}$ be the Fourier transform of $f$. Since $f \in W^{p, 1}(\R^3)$, we deduce that it holds
	\begin{equation} \label{eq:hatf}
		\left| \widehat{f}(\omega) \right| \le (2 \pi)^{-3/2} \dfrac{\left\| f \right\|_{W^{p,1}}}{1 + |\omega |^p}.
	\end{equation}
	We deduce that $\widehat{f} \in L^1(\R^3)$, so that $f$ is continuous. In particular, the point-wise evaluations $f(\bk/L)$ in $I_L(f)$ (see~\eqref{eq:def:Ix}) are well-defined. On the other hand, according to the Poisson summation formula, it holds that
	\[
		I_L(f) = \dfrac{( 2 \pi)^{3/2}}{| \BZ |} \sum_{\bR \in \Lat} \widehat{f} (L \bR),
	\]
	so that
	\[
		I_L(f) - I(f)  = I_L(f) - \dfrac{( 2 \pi)^{3/2}}{| \BZ |} \widehat{f}(\bnull) = \dfrac{( 2 \pi)^{3/2}}{| \BZ |} \sum_{\bR \in \Lat \setminus \{ \bnull \}} \widehat{f}(L \bR) 
		\le \dfrac{1}{| \BZ | L^p} \sum_{\bR \in \Lat \setminus \{ \bnull \}} \dfrac{\left\| f \right\|_{W^{p,1}}}{| \bR |^p},
	\]
	where the last inequality comes from~\eqref{eq:hatf}.
	The proof follows.
\end{proof}

We now study functions $f(\bq)$ that have a singularity as $\bq$ goes to $\bnull$. We introduce the notation
\begin{equation} \label{eq:Ix0}
	\forall \lambda > 0, \quad I_\lambda^0(f) := \dfrac{1}{\lambda^3} \sum_{\bQ \in \RLat \setminus \{ \bnull \}} f \left( \dfrac{\bQ}{\lambda} \right).
\end{equation}
Note that if $f$ is continuous in $\bq = \bnull$, then it holds
\begin{equation} \label{eq:Ix0-Ix}
	I_\lambda(f) = I_\lambda^0(f) + \dfrac{f(\bnull)}{\lambda^3},
\end{equation}
so that $I_\lambda(f) = I_\lambda^0(f)$ if $f(\bnull) = 0$. Also, if  $f(\bnull) \neq 0$, we expect a $L^{-3}$ rate of convergence for the difference $I(f) - I_L^0(f)$.  For $r_1 > 0$, and $p \in \N$, we denote by $C^p(\cB(\bnull, r_1))$ the Banach space of functions that are $p$ times continuously differentiable on $\cB(\bnull, r_1)$, endowed with the norm
\[
	\left\| f \right\|_{C^p(\cB(\bnull, r_1))} := \sup_{\bq \in \cB(\bnull, r_1)} \sup_{| \alpha| \le p} \left| \partial_\bq^\alpha f(\bx ) \right|.
\]
Our main result is the following.

\begin{lemma} \label{lem:singularRiemann}
	Let $r_1 > 0$ and $\bq \mapsto M(\bq)$ be a $C^4$ function on  $\cB(\bnull, r_1)$ to the space of $3 \times 3$ positive definite hermitian matrices satisfying $M(\bq) \ge 1$ and $M(-\bq) = M(\bq)$ for all $\bq \in \cB(\bnull, r_1)$. Let also $\Psi$ be a radial $C^\infty(\R^3)$ function such that $\Psi(\bq) = 1$ if $|\bq| \le r_1/2$ and $\Psi(\bq) = 0$ if $|\bq| \ge r_1$. 
	Then, there exists $C \in \R^+$ such that for all functions $g \in C^4(\cB(\bnull, r_1))$, it holds that
	\[
		\forall L \in \N^*, \quad 
		\left| I(f \Psi) - I_L^0(f \Psi) -  \fa \dfrac{g(\bnull)}{L} \right| \le C\dfrac{\left\| g \right\|_{C^4(\cB(\bnull, r_1))} }{L^3} 
	\]
	where $f(\bq) := \dfrac{g(\bq)}{\bq^T M(\bq) \bq}$ for $\bq \neq \bnull$, and where
	\[
		\fa = \sum_{\bk \in \RLat} \fint_{\BZ} \left( \dfrac{1}{(\bk + \bq)^T M(\bnull) (\bk + \bq)} - \dfrac{\mathds{1}(\bk \neq \bnull)}{\bk^T M(\bnull) \bk} \right) \rd \bq.
	\]
\end{lemma}

\begin{proof}
We perform a Taylor expansion for $f$. In order to do so, we first do the Taylor expansion for $g$ and $M$. We write
\[
	g(\bq) = g(\bnull) + g_1(\bq) + g_2(\bq) + g_3(\bq) + \widetilde{g_4}(\bq)
\]
where we set, for $N \in \{ 1, \ldots, 4 \}$,
\[
	g_N(\bq) := \sum_{| \alpha | = N} \left( \dfrac{1}{\alpha !} \partial_{\bq}^\alpha g(\bq) \big|_{\bq = \bnull} \right) \bq^\alpha
	\quad \text{and} \quad
	\widetilde{g_{N+1}} = g - \sum_{n=1}^N g_N.
\]
Note that $g_N$ is an homogeneous polynomial of degree $N$ and that it holds
\[
	\sup_{| \beta | \le 4} \left| \partial_\bq^\beta \widetilde{g_4}(\bq) \right| \le C \left\| g \right\|_{C^4(\cB(\bnull, r_1))} |\bq|^{4 - \beta}.
\]
Similarly, we write
\[
	M(\bq) = M(\bnull) + M_1(\bq) + M_2(\bq) + M_3(\bq) + \widetilde{M_4}(\bq) 
\]
where $M_N(\bq)$ is a $3 \times 3$ matrix-valued homogeneous polynomial of degree $N$. In the sequel, we write $M := M(\bnull)$ for clarity.
After some calculations, we obtain
\begin{align} \label{eq:f}
	 f(\bq)  = 
	f_0(\bq) + f_1(\bq) + f_2(\bq) + f_3(\bq) +  \widetilde{f_4}(\bq),
\end{align}
with
\begin{align*}
	f_0(\bq) & = \dfrac{g (\bnull)}{\bq^T M \bq}  \\
	f_1(\bq) & =   \dfrac{1}{\bq^T M \bq} \left(g_1(\bq) - g (\bnull) \dfrac{\bq^T A_1(\bq) \bq}{\bq^T M \bq} \right) \\
	f_2(\bq) & = \dfrac{1}{\bq^T M  \bq}  \left( - g(\bnull) \dfrac{\bq^T A_2(\bq) \bq}{\bq^T M \bq} + g(\bnull) \dfrac{\left( \bq^T A_1(\bq) \bq \right)^2}{\left( \bq^T M \bq \right)^2} - g_1(\bq) \dfrac{\bq^T A_1(\bq) \bq}{\bq^T M \bq} + g_2(\bq) \right) \\
	f_3(\bq) & = \dfrac{1}{\bq^T M \bq} \left(
	g_3(\bq) - g_2(\bq) \dfrac{\bq^T M_1(\bq) \bq }{\bq^T M \bq }  
	- g_1(\bq) \dfrac{\bq^T M_2(\bq) \bq }{\bq^T M \bq } + g_1(\bq) \left(  \dfrac{\bq^T M_1(\bq) \bq }{\bq^T M \bq } \right)^3 \right. \\
	& \quad \left. - g(\bnull)  \dfrac{\bq^T M_3(\bq) \bq }{\bq^T M \bq } 
	+ 2 g(\bnull)  \dfrac{\bq^T M_1(\bq) \bq \, \bq^T M_2(\bq) \bq }{( \bq^T M \bq )^2}  - g(\bnull) \dfrac{ \left( \bq^T A_1(\bq) \bq \right)^3}{\left( \bq^T M \bq \right)^3}  \right),
\end{align*}
and where $\widetilde{f_4} := f - f_1 -  f_2 - f_3$ satisfies
\begin{equation} \label{eq:widetildef4}
	\sup_{ | \beta | \le 4} \left| \partial_\bq^\beta \widetilde{f_4}(\bq) \right| \le C \left\| g \right\|_{C^4(\cB(\bnull, r_1))}  |\bq|^{2 - \beta},
\end{equation}
for some constant $C \in \R^+$ independent of $g$. The idea of the proof is to write
\begin{equation} \label{eq:TaylorForRiemann}
	I(f \Psi) - I_L^0(f \Psi) = \sum_{k=0}^3 I(f_k \Psi) - I_L^0(f_k \Psi) + \left( I(\widetilde{f_4} \Psi ) - I_L^0( \widetilde{f_4} \Psi)  \right),
\end{equation}
and to evaluate each part of the right-hand side. We begin with the remainder terms.
\begin{lemma}
There exists $C \in \R^+$ such that, for all $g \in C^4(\cB(\bnull, r_1))$, it holds that 
\[
	\left| I \left( \widetilde{f_4} \Psi \right)  - 
		I_L^0 \left( \widetilde{f_4} \Psi \right)\right| \le 
		\dfrac{C}{L^4}   \left\| g \right\|_{C^4(\cB(\bnull, r_1))}.
\]
\end{lemma}
\begin{proof}
From~\eqref{eq:widetildef4} we deduce that the function $\widetilde{f_4} \Psi$ is in $W^{4,1}(\R^3)$. Also, it holds that $\left( \widetilde{f_4} \Psi \right)(\bnull) = 0$. We deduce from Lemma~\ref{lem:RiemannR3} and~\eqref{eq:Ix0-Ix} that
\[
	\left| I \left( \widetilde{f_4} \Psi \right)  - 
		I_L^0 \left( \widetilde{f_4} \Psi \right)\right| = \left| I \left( \widetilde{f_4} \Psi \right)  - 
		I_L \left( \widetilde{f_4} \Psi \right)\right| \le C_p \dfrac{\left\| \widetilde{f_4} \Psi \right\|_{W^{4,1}}}{L^4} 
		\le C \dfrac{\left\| g \right\|_{C^4(\cB(\bnull, r_1))}}{L^4} .
\]
\end{proof}

On the other hand, from the symmetries $f_1(-\bq) = -f_1(\bq)$ and $f_3(-\bq) = -f_3(\bq)$, and the fact that $\Psi$ is radial, we easily obtain the following result.
\begin{lemma}
	It holds that
	\[
		I  \left( {f_{1}}\Psi \right) = I  \left( {f_{3}}\Psi \right) = 0
		\quad \text{and}
		\quad
		\forall L \in \N^*, \quad I_L^0 \left( {f_{1}}\Psi \right)=  I_L^0 \left( {f_{3}}\Psi \right)= 0.
	\]
\end{lemma}

It remains to study the $f_0$ and $f_2$ terms. We start with the $f_2$ term. 
\begin{lemma}
	There exists $C \in \R^+$ such that, for all $g  \in C^4(\cB(\bnull, r_1))$, it holds that
	\[
		\left| I \left( \widetilde{f_2} \Psi \right)  - 
		I_L \left( \widetilde{f_2} \Psi \right)\right| \le C \dfrac{\left\| g \right\|_{C^4(\cB(\bnull, r_1))} }{L^3} .
	\]
\end{lemma}

\begin{proof}
We introduce the function
\[
	F_g(x) := x^3 \left[ I(f_2 \Psi) - I_x^0(f_2 \Psi) \right].
\]
Our goal is to prove that $F_g$ is uniformly bounded by $C \left\| g \right\|_{C^4(\cB(\bnull, r_1))}$.
Since $\Psi$ is compactly supported, the sum~\eqref{eq:Ix0} for $I_x^0(f_2 \Psi)$ is finite for all $x \in (0, \infty)$. Moreover, since both $f_2$ and $\Psi$ are continuous away from zero, so is $F_g$. 

\medskip

It holds that $f_2(\lambda \bq) = f_2(\bq)$ for all $\lambda \in \R$ and all $\bq \in \R^3$. As a result, the change of variable $\by = 2\bq$ leads to
\begin{align*}
	x^3 I(f_2 \Psi) - (2x)^3 I(f_2 \Psi) & = \dfrac{x^3}{| \BZ |} \int_{\R^3} f_2(\bq) \Psi(\bq) \rd \bq 
	+ \dfrac{8 x^3}{8 | \BZ |} \int_{\R^3} f_2(\by) \Psi \left( \dfrac{\by}{2} \right) \rd \by \\
		& = \dfrac{x^3}{| \BZ |} \int_{\R^3} f_2(\bq) \Phi(\bq) \rd \bq = x^3  I(f_2 \Phi),
\end{align*}
where we set $\Phi(\bq) = \Psi(\bq) - \Psi(\bq/2)$. We also get
\[
	x^3 I_x^0 (f_2 \Psi) - (2x)^3 I_x^0(f_2 \Psi) = x^3 I_x^0(f_2 \Phi) = x^3 I_x(f_2 \Phi),
\]
where we used the fact that $(f_2 \Phi)(\bnull) = 0$ for the last equality. Altogether, we obtain
\[
	F_g(x) - F_g(2x) = x^3 \left( I(f_2 \Phi) - I_x(f_2 \Phi) \right).
\]
Since $\Phi$ is a $C^\infty(\R^3)$ function with support contained in $\cB(\bnull, 2r) \setminus \cB(\bnull, r/2)$, we easily deduce that $f_2 \Phi$ is a $C^\infty(\R^3)$ compactly supported function. Together with Lemma~\ref{lem:RiemannR3} with $p = 5$, we deduce that there exist $C,C' \in \R^+$ such that, for all $g \in C^4(\cB(\bnull, r_1))$,
\[
	\left| F_g(x) - F_g(2x) \right| \le C' \dfrac{\left\| f_2 \Psi \right\|_{W^{5,1}(\R^3)}}{x^2} \le C \dfrac{\left\| g \right\|_{C^4(\cB(\bnull, r_1))}}{x^2}.
\]
Let $K_g :=  \sup_{x \in [1,2]} \left| F_g(x) \right|$, so that, for all $g \in C^4(\cB(\bnull, r_1))$, it holds $| K_g|  \le K \left\| g \right\|_{C^4(\cB(\bnull, r_1))}$ for some constant $K \in \R^+$ independent of $g$. Let $x \ge 1$ and $k \in \N^*$ be chosen such that $2^k \le x \le 2^{k+1}$. We obtain from a simple cascade argument that
\begin{align*}
	\left| F_g(x)  \right| & \le \sum_{l = 0}^{k-1} \left| F_g \left( \dfrac{x}{2^{l}} \right) - F_g \left( \dfrac{x}{2^{l+1}} \right) \right| + \left| F_g \left( \dfrac{x}{2^k}\right) \right| 
	\le \left\| g \right\|_{C^4(\cB(\bnull, r_1))} \left( C \sum_{l=0}^{k-1} \left( \dfrac{2^l}{x} \right) + K_g \right) \\
		& \le \left\| g \right\|_{C^4(\cB(\bnull, r_1))} \left( C \sum_{m=0}^{\infty} \left( \dfrac{1}{2^m} \right) + K_g \right),
\end{align*}
where we used the fact that $x^{-1} \le 2^{-k}$ and performed the change of variable $m = k-l$ in the last inequality. The result follows.
\end{proof}

It remains to study the convergence for $f_0$, hence the difference
\[
	I(f_0 \Psi )  - I_L^0(f_0 \Psi) = 
		g(\bnull)	\left( \dfrac{1}{| \BZ |} \int_{\R^3} \dfrac{\Psi(\bq)}{\bq^T M \bq} \rd \bq - 
		\dfrac{1}{L^3} \sum_{\bQ \in L^{-1}\RLat \setminus \{ \bnull \}}\dfrac{\Psi(\bQ)}{\bQ^T M \bQ}
		\right).
\]

\begin{lemma}
For all $3 \times 3$ hermitian matrix $M$ satisfying $M \ge \bbI$, and for any $\Psi \in C^{\infty}(\R^3)$ that is compactly supported in $\BZ$ and that satisfies $\Psi(\bq) = 1$ for all $| \bq | \le r$ for some $r > 0$, it holds that
\begin{equation} \label{eq:forPsi}
	\dfrac{1}{| \BZ |} \int_{\R^3} \dfrac{\Psi(\bq)}{\bq^T M \bq} \rd \bq - 
		\dfrac{1}{L^3} \sum_{\bQ \in L^{-1}\RLat \setminus \{ \bnull \}}\dfrac{\Psi(\bQ)}{\bQ^T M \bQ} = \dfrac{\fa}{L} + O\left( \dfrac{1}{L^4} \right),
\end{equation}
where
\[
	\fa =  \sum_{\bk \in \RLat} \fint_{\BZ} \left( \dfrac{1}{(\bk + \bq)^T M (\bk + \bq)} - \dfrac{\mathds{1}(\bk \neq \bnull)}{\bk^T M \bk} \right) \rd \bq.
\]
\end{lemma}

\begin{proof}

By introducing $\widetilde{\Psi}(\bq) := \Psi(\sqrt{M}^{-1} \bq)$, $\widetilde{\RLat} = \sqrt{M} \RLat$ and $ \widetilde{\BZ}  = \sqrt{M} \BZ$, so that $|\widetilde{\BZ}| = \sqrt{\det{M}} | \BZ |$, we obtain by a simple change of variable that
\[
	\dfrac{1}{| \BZ |} \int_{\R^3} \dfrac{\Psi(\bq)}{\bq^T M \bq} \rd \bq - 
		\dfrac{1}{L^3} \sum_{\bQ \in L^{-1}\RLat \setminus \{ \bnull \}}\dfrac{\Psi(\bQ)}{\bQ^T M \bQ}
	= \dfrac{1}{| \widetilde{\BZ} |} \int_{\R^3} \dfrac{\widetilde{\Psi}(\bq)}{ | \bq |^2} \rd \bq - 
		\dfrac{1}{L^3} \sum_{\bQ \in L^{-1} \widetilde{\RLat} \setminus \{ \bnull \}}\dfrac{\widetilde{\Psi}(\bQ)}{ |\bQ |^2},
\]
and that
\[
	\fa = \sum_{\widetilde{\bk} \in \widetilde{\RLat}} \fint_{\widetilde{\BZ}} \left( \dfrac{1}{|\widetilde{\bk} + \widetilde{\bq} |^2 } - \dfrac{\mathds{1}(\widetilde{\bk} \neq \bnull)}{|\widetilde{\bk}|^2} \right) \rd \widetilde{\bq},
\]
so that it is enough to prove the result for the special case $M = \mathbb{I}_3$. We consider this case in the sequel. For a function $g$, we introduce the following notation for clarity:
\[
	J(g) := \dfrac{1}{| \BZ |} \int_{\R^3} \dfrac{g(\bq)}{| \bq |^2} \rd \bq
	\quad \text{and} \quad
	J_x^0(g) :=  \sum_{\bQ \in \widetilde{\RLat} \setminus \{ \bnull \}}\dfrac{ g \left( \bQ/x \right)}{ | \bQ/x |^2},
\]
so that the left-hand side of~\eqref{eq:forPsi} is also $J(\Psi) - J_L^0(\Psi)$. Let $\Phi$ be a radial $C^\infty(\R^3)$ function such that $\Phi(\bq) = 1$ for $| \bq | \le r/2$ and $\Phi(\bq) = 0$ for $| \bq | \ge r$. By writing $\Psi = \Psi \Phi + \Psi (1 - \Phi) = \Phi + \Psi (1 - \Phi)$, we obtain
\[
	J(\Psi) - J_x^0(\Psi) = \left( J(\Phi) - J_L^0(\Phi) \right) + \left( J(\Psi (1 - \Phi)) - J_L^0(\Psi (1 - \Phi)) \right).
\]
The function $\bq \mapsto \Psi(\bq)(1 - \Phi(\bq))/ \ |\bq |^2$ is a $C^\infty(\R^3)$ compactly supported function. Hence, from Lemma~\ref{lem:RiemannR3}, we obtain that for all $p \in \N^*$, it holds that
\begin{equation} \label{eq:PsiPhi}
	\left| \left( J(\Psi) - J_x^0(\Psi) \right) - \left( J(\Phi) - J_L^0(\Phi) \right) \right| = O(L^{-p}).
\end{equation}
As a result, we deduce that it is enough to prove~\eqref{eq:forPsi} for $\Psi = \Phi$ a radial function. Let us now evaluate the difference $S_L := J(\Phi) - J_L^0(\Phi)$. It holds that
\begin{align}
	S_L & := \dfrac{1}{| {\BZ} |} \int_{\R^3} \dfrac{{\Phi}(\bq)}{ | \bq |^2} \rd \bq - 
		\dfrac{1}{L^3} \sum_{\bQ \in L^{-1} {\RLat} \setminus \{ \bnull \}}\dfrac{{\Phi}(\bQ)}{ |\bQ |^2} \nonumber \\
		&	 = \dfrac{1}{| \BZ|} \sum_{\bQ \in L^{-1} \RLat} \int_{L^{-1}  \BZ} \left( \dfrac{\Phi(\bQ + \bx)}{ \left| \bQ + \bx \right|^2} - \dfrac{\Phi(\bQ) \mathds{1}(\bQ \neq \bnull)}{  \left| \bQ \right|^2} \right) \rd \bx \nonumber \\
	&  = \dfrac{1}{| \BZ|} \sum_{\bQ \in L^{-1} \RLat} \int_{L^{-1}\BZ} \left( \dfrac{1}{ \left| \bQ + \bx \right|^2} - \dfrac{\mathds{1}(\bQ \neq \bnull)}{ \left| \bQ \right|^2} \right) \rd \bx \label{eq:Lm1} \\
	& \quad + \dfrac{1}{| \BZ|} \sum_{\bQ \in L^{-1}\RLat} \int_{L^{-1}\BZ} \left( \dfrac{\Phi(\bQ + \bx) - 1}{ \left| \bQ + \bx \right|^2} - \dfrac{(\Phi(\bQ) -1)}{ \left| \bQ \right|^2} \right) \rd \bx. \label{eq:Lm2}
\end{align}

The proof that~\eqref{eq:Lm1} and~\eqref{eq:Lm2} are indeed convergent series will be given in the sequel. We will note $S_{1,L}$ and $S_{2,L}$ the terms in~\eqref{eq:Lm1} and~\eqref{eq:Lm2} respectively, so that $S_L = S_{1,L} + S_{2,L}$. Let us first evaluate $S_{1,L}$. The change of variable $\bk = L\bQ$ and $\bq = L \bx$ leads formally to
\begin{equation} \label{eq:expression_fa}
	S_{1,L}
	 = \dfrac{1}{L} \sum_{\bk \in \RLat} \fint_{\BZ} \left( \dfrac{1}{|\bk + \bq|^2} - \dfrac{\mathds{1}(\bk \neq \bnull)}{| \bk |^2} \right) \rd \bq = \dfrac{\fa}{L}.
\end{equation}
From the so-called multipole expansion, we get, for $\bk \neq \bnull$ and $\bq \in \BZ$, that there exists $C \in \R^+$ such that, for all $\bk \neq \bnull$, and all $\bq \in \BZ$,
\begin{equation} \label{eq:def:F1}
	F_1(\bk, \bq) := \dfrac{1}{|\bk + \bq|^2} - \dfrac{1}{| \bk|^2} + 2 \dfrac{\bk^T \bq}{|\bk|^4} 
\end{equation}
satisfies $\left| F_1(\bq, \bk) \right| \le C \bq^2 | \bk |^{-4}$. As a result, we obtain
\[
	\left| \fint_{\BZ} \left( \dfrac{1}{ \left| \bk + \bq \right|^2} - \dfrac{\mathds{1}(\bk \neq \bnull)}{|\bk|^2} \right) \right| = \left|  \fint_{\BZ} F_1(\bk, \bq)  \rd \bq \right|
	 \le \left( C \fint_{\BZ} | \bq |^2 \rd \bq \right) \dfrac{1}{| \bk |^4},
\]
so that the left-hand side of~\eqref{eq:expression_fa} is indeed a convergent series, and $S_{1,L} = L^{-1} \fa$. 

\medskip

We now study $S_{2,L}$. The fact that $S_{2,L}$ is a convergent series is proved similarly than for $S_{1,L}$. We introduce the function
\[
	h(\rho) := \dfrac{ (\Phi(\rho) -1) }{ \rho^2}.
\]
 It is not difficult to see that $h$ is a $C^\infty(\R)$ function such that $h(\rho) = 0$ for $\rho \le r/2$, and that $h^{(n)}(\rho) = O( \rho ^{-(2+ n)})$  for all $n \in \N$. The Taylor expansion of $h(| {\bQ} + {\bx}|)$ near $\bQ$ leads to,
\begin{align*}
	h(| {\bQ} + {\bx}|) & = h( |{\bQ}|) + h'( |{\bQ}|) \dfrac{{\bQ} \cdot {\bx}}{| {\bQ} |}
		+ \dfrac{h''(| {\bQ} |)}{2} \dfrac{({\bQ} \cdot {\bx})^2}{| {\bQ} |^2} 
		+ \dfrac{h'(| {\bQ}|)}{2} \left( \dfrac{| {\bx} |^2}{| {\bQ}|} - \dfrac{({\bQ} \cdot {\bx})^2}{| {\bQ} |^3} \right) \\
		& \quad	+ P_3({\bQ}, {\bx}) + H({\bQ}, {\bx}),
\end{align*}
where, for all ${\bQ}$, $P_3({\bQ}, {\bx})$ is an homogeneous polynomial of degree $3$ in the variables $(x_1, x_2, x_3)$, and where $H({\bQ}, {\bx})$  satisfies an inequality of the type
\begin{equation} \label{eq:boundH}
	H({\bQ}, {\bx}) \le C \dfrac{ | {\bx} |^4}{ \left( 1 + | {\bQ} | \right)^{5}}.
\end{equation}
Since $h'( |{\bQ}|) {\bQ} \cdot {\bx} / | {\bQ} |$ and $P_3(\bq, \bx)$ are odd functions in the variable $\bx$, we deduce that $S_{2,L} = S_{2,L}^1 + S_{2,L}^2 + S_{2,L}^3$, where
\begin{align}
	S_{2,L}^1 & =
	\dfrac{1}{| \BZ|} \sum_{\bQ \in L^{-1}\RLat}  \int_{L^{-1}\Gamma^\ast} \left( \dfrac{h''(| {\bQ} |)}{2} \dfrac{(  {\bQ} \cdot {\bx})^2}{| {\bQ}|^2} \right) \rd {\bx}  \label{eq:h''} \\
	S_{2,L}^2 & = \dfrac{1}{| \BZ|} \sum_{\bQ \in L^{-1}\RLat}  \int_{L^{-1} \Gamma^\ast}  \left( \dfrac{h'(| {\bQ}|)}{2} \left( \dfrac{| {\bx} |^2}{| {\bQ}|} - \dfrac{({\bQ} \cdot {\bx})^2}{| {\bQ} |^3} \right) \right) \rd {\bx}  \label{eq:h'} \\
	S_{2,L}^3 & = \dfrac{1}{| \BZ|} \sum_{\bQ \in L^{-1}\RLat}  \int_{L^{-1} \Gamma^\ast}  H({\bQ}, \bx) \rd {\bx}. \label{eq:remainder}
\end{align}
We first consider $S_{2,L}^3$ defined in~\eqref{eq:remainder}. From~\eqref{eq:boundH}, it holds that, 
\[
	\left| \int_{L^{-1} \Gamma^\ast}  H({\bQ}, \bx) \rd {\bx} \right| \le 
	\dfrac{C}{\left( 1 + | {\bQ} | \right)^5} \int_{L^{-1} \BZ} | {\bx} |^4 \rd {\bx}
	 \le \dfrac{C}{\left( 1 + | {\bQ} | \right)^5} \left( \int_{\BZ} | \bq |^4 \rd \bq \right) \dfrac{1}{L^7}.
\]
The map $\bq \mapsto {( 1 + | \bq |)^{-5}}$ satisfies the hypothesis of Lemma~\ref{lem:RiemannR3} for all $p > 3/2$, so that
\begin{equation} \label{S2L3}
	 \left| S_{2,L}^3 \right|
	\le 
	\dfrac{1}{L^4} \dfrac{C}{L^3} \sum_{\bQ \in L^{-1}\RLat} \dfrac{1}{\left( 1 + | {\bQ} | \right)^5} 
	 = \dfrac{1}{L^4} \int_{\R^3} \dfrac{\rd \bq}{(1 + | \bq |)^{-5}} + O(L^{-p-4}) = O(L^{-4}).
\end{equation}
We now consider the term $S_{2,L}^1$ defined in~\eqref{eq:h''}. We denote, for $1 \le j \le 3$,
\[
	m_{j} := \fint_{\BZ} q_j^2 \rd \bq, 
	\quad \text{and} \quad
	m := \sum_{j=1}^3 m_{j} = \fint_{\BZ} | \bq |^2 \rd \bq.
\]
It holds
\[
	\dfrac{1}{| \BZ|} \int_{L^{-1} \BZ} \left( \dfrac{h''(| {\bQ} |)}{2} \dfrac{(  {\bQ} \cdot {\bx})^2}{| {\bQ}|^2} \right) \rd {\bx} 
	= 
	\dfrac{1}{L^5} \sum_{j=1}^3 m_{j}  \left( \dfrac{h''(| {\bQ}| )}{2} \dfrac{{Q_j}^2}{| {\bQ}|^2} \right).
\]
The map $\bx \mapsto h''(| \bx | ) x_j^2 / (2 | \bx |^2)$ satisfies the hypothesis of Lemma~\ref{lem:RiemannR3} for all $p > 3/2$, so that
\[
	\dfrac{1}{L^3} \sum_{{\bQ} \in L^{-1} \RLat} \dfrac{h''(| {\bQ}|)}{2} \dfrac{{Q_j}^2}{| {\bQ}|^2}
	= \dfrac{1}{| \BZ |}\int_{\R^3} \dfrac{h''(| \bx |)}{2} \dfrac{x_j^2}{| \bx |^2} \rd \bx + O(L^{-p}).
\]
We evaluate this last integral using spheric coordinates. It holds that
\[
	\int_{\R^3} \dfrac{h''(| \bx |)}{2} \dfrac{x_j^2}{| \bx |^2} \rd \bx = \pi \int_0^\pi \rd \theta \cos^2(\theta) \sin(\theta) \int_0^\infty h''(r) r^2 \rd r 
	= \dfrac{4 \pi}{3} \int_0^\infty h(r) \rd r.
\]
Altogether, we obtain
\begin{equation} \label{eq:S2L1}
	S_{2,L}^1 = \dfrac{m}{L^2} \dfrac{4 \pi}{3 | \BZ |} \left( \int_0^\infty h(r) \rd r \right) + O(L^{-p}).
\end{equation}
We finally consider the term $S_{2,L}^2$ defined in~\eqref{eq:h'}. Similar calculations leads to
\begin{equation} \label{eq:S2L2}
	S_{2,L}^2 = -\dfrac{m}{L^2} \dfrac{4 \pi}{3 | \Gamma^\ast|} \int_0^\infty h(r) \rd r + O(L^{-p}).
\end{equation}
As a result, from~\eqref{eq:S2L1} and~\eqref{eq:S2L1}, we obtain that $ \left| S_{2,L}^1 + S_{2,L}^2 \right| = O(L^{-p})$. Together with~\eqref{S2L3}, one finally get $S_{2,L} = O(L^{-4})$, which was the desired result.

\end{proof}

\end{proof}

\section*{Acknowledgments}
We are grateful to professors \'E. Cancès and M. Lewin for their help and support throughout this work. We also thank professor A. Ern for discussions on the convergence of Riemann sums.

\end{appendices}

\bibliographystyle{plain}

\bibliography{Charged_defects}

\end{document}